\documentclass[cmp]{svjour}

\journalname{Communications in Mathematical Physics}
\usepackage{amsfonts,amssymb,amsmath,gensymb,color,epic,eepic,float,fullpage}
\numberwithin{theorem}{section}
\usepackage{empheq}
\usepackage[mathscr]{euscript}
\usepackage[T1]{fontenc}
\usepackage{mathtools}
\usepackage{physics}
\usepackage[shortlabels]{enumitem}
\usepackage{comment}
\usepackage{bbm}
\usepackage{tikz,tikz-dimline}
\usetikzlibrary{graphs,patterns,decorations.markings,arrows,matrix}
\usetikzlibrary{calc,decorations.pathmorphing,decorations.markings,
decorations.pathreplacing,patterns,shapes,arrows}
\tikzstyle{block} = [rectangle,draw,text width=10em,text centered,rounded corners,minimum height=4em]
\usetikzlibrary{positioning}
\usetikzlibrary{arrows.meta}
\tikzstyle{line} = [draw, -latex']
\allowdisplaybreaks

\numberwithin{equation}{section}

\usepackage{graphicx}
\usepackage{hyperref}
\usepackage{epstopdf}
\usepackage{comment}
\usepackage{cite}
\usepackage{appendix}
\usepackage{lipsum}
\usepackage{psfrag}
\usepackage{pstool}
\usepackage{caption}
\usepackage[normalem]{ulem}
\makeatletter 
\renewcommand{\p@enumii}{} %added in order to reference the numbers of items correctly
\makeatother

\DeclareMathOperator{\e}{e}
\DeclareMathOperator{\Ai}{Ai}
\def\dd{{\rm d}}
\def\ii{{\rm i}}
\def\i{{\infty}}
\def\be{\begin{equation}}
\def\ee{\end{equation}}
\def\bea{\begin{eqnarray}}
\def\eea{\end{eqnarray}}

\def\ket#1{|#1\rangle}

\def\Z{\mathbb{Z}}
\DeclareMathOperator{\sign}{sign}
\DeclareMathOperator{\res}{Res}
\DeclareMathOperator{\R}{Re}

\newtheorem{defn}[theorem]{Definition}
\newtheorem{prop}[theorem]{Proposition}
\newtheorem{lemman}[theorem]{Lemma}
\newtheorem{corollaryn}[theorem]{Corollary}

\definecolor{green2}{rgb}{0.0,0.5,0.0}
\definecolor{classicrose}{rgb}{0.98,0.8,0.91}
\definecolor{blizzardblue}{rgb}{0.67,0.9,0.93}

\begin{document}
\title{Limiting current distribution for a two species asymmetric exclusion process}
\date{\small\today}

\author{Zeying Chen$^{1,3}$, Jan de Gier$^1$, Iori Hiki$^2$, Tomohiro Sasamoto$^2$ and Masato Usui$^2$}

\institute{$^1$ARC Centre of Excellence for Mathematical and Statistical Frontiers (ACEMS), School of Mathematics and Statistics, The University of Melbourne, VIC 3010, Australia\\
$^2$Department of Physics, Tokyo Institute of Technology, Ookayama 2-12-1, Tokyo 152-8551, Japan \\
$^3$School of Mathematical Sciences, University of Science and Technology of China, Hefei, Anhui, 230026, PR China\\
\email{zeyingc@ustc.edu.cn, jdgier@unimelb.edu.au, hiki.i.aa@icloud.com, sasamoto@phys.titech.ac.jp, musui@stat.phys.titech.ac.jp }
}

\maketitle
\abstract{We study current fluctuations of a two-species asymmetric exclusion process, known as the Arndt-Heinzel-Rittenberg model.
For a step-Bernoulli initial condition with finite number of particles, we provide an explicit multiple integral expression for a certain joint current probability distribution. By performing an asymptotic analysis we prove that the joint current distribution is given by a product of a Gaussian and a GUE Tracy-Widom distribution in the long time limit, as predicted by non-linear fluctuating hydrodynamics.
}

\section{Introduction and main results}
\label{se:introduction}
Asymmetric exclusion processes on $\Z$, in which many particles perform asymmetric random walks with only
one particle allowed on each lattice site, are considered to be among the most fundamental processes in the
theory of stochastic interacting particle systems\cite{Liggett1985,Liggett1999,Spohn1991}. %\cite{liggett2012interacting}. 
Originally introduced as biophysical models for protein synthesis on RNA \cite{Macdonald68,Macdonald69} they also have many other applications in biology and in other disciplines such as physics and engineering \cite{Golinelli_2006}. These processes have attracted much attention  over the
years and various large scale behaviours have been studied such as the hydrodynamic limit, large
deviations  %\cite{varadhan1984large} 
and other properties \cite{kipnis1998scaling}.

More recently many studies have focused on the non-Gaussian fluctuation properties of asymmetric exclusion processes, which are related to the Kardar-Parisi-Zhang (KPZ) universality class\cite{KPZ1986,BS1995}. In 2000, Johansson showed that the current distribution
of the totally asymmetric simple exclusion process (TASEP) with the step initial condition is given by the
GUE Tracy-Widom distribution in the long time limit \cite{J2000}. As a corollary  he proved the fluctuation exponent 1/3, which is characteristic of the KPZ class in one dimension. \footnote{In two dimensions the ASEP has a $\log(t)^{2/3}$ law \cite{Y2004}.} The KPZ dynamical exponent was also observed in Bethe ansatz studies of the spectral gap of the asymmetric simple exclusion process (ASEP) with periodic and open boundary conditions \cite{PhysRevA.46.844, PhysRevLett.68.725, PhysRevE.52.3512, Gier_2005, Gier_2006, Gier_2008}. Since Johansson's result there have been a
number of generalisations of his results for various types of asymmetric exclusion processes, including the
ASEP and $q$-TASEP, under  several different initial conditions, see e.g. \cite{BaikRains00, PS2002, Sasamoto_2005, TW2009a, Sasamoto_2010a,  Sasamoto_2010b, Sasamoto_2010c, Amir_2010, CalabreseDoussal, BC2014}. 

Most of these results on limiting distributions have been established based on explicit exact formulas for appropriate
quantities before taking the long time limit. For example, for the case of TASEP with step initial condition, the current
distribution is written as a multiple integral related to random matrix theory \cite{J2000,NS2004}. For ASEP and $q$-TASEP,
certain $q$-deformed moments admit a multiple integral representation, which leads to a Fredholm determinant
expression for the $q$-deformed Laplace transform \cite{BC2014,bcs}. The reason for the existence of such explicit formulas for certain classes of models is related to the underlying integrability of such models. In fact, almost all models for which limiting distributions have been studied have turned out to be a special or a limiting case of a stochastic higher spin six vertex model (HS6VM)  \cite{MANGAZEEV201470,povolotsky2013,Corwin_2015, BP2016}, whose similarity transformed non-stochastic version has long been known to be (Yang-Baxter) integrable \cite{Kulish1981YangBaxterEA,Kirillov_1987}. Very recently a method to study asymptotic distributions without resorting to explicit formulas or integrability has been 
proposed \cite{QS2020p} (cf. also \cite{Virag2020p}). 

Despite these remarkable developments, most of the asymptotic results have so far only been obtained for models in which there is only a single species of particles. It is a quite natural and important problem to try to generalise the analysis to multi-species models. Given the above situation for single species
models, it would be natural to start our studies on integrable multi-species models. In fact several multi-species models
have been already known. Multi-species asymmetric exclusion processes were
introduced a long time ago, see e.g. \cite{PhysRevE.59.205,Mallick_1999,Derrida_exactsolution}, and several other multi-species stochastic processes have been proposed recently in \cite{kuniba2016multispecies}. Among them those with $U_q(sl_n)$ symmetry ($n$-ASEP) have been most studied. Their stationary measures were given in \cite{ferrari2007,PEM2009} and in \cite{Cantini_2015} these were put in the context of Macdonald polynomial theory, the Knizhnik-Zamolodchikov equation and bosonic solutions of the Yang-Baxter equation. The $n$-TASEP was shown to be related to the combinatorial $R$-matrix and solutions of the tetrahedron equation in \cite{kuniba2016multispecies,Kuniba_2016}. In \cite{ChenGW,Kuan_2018} methods have been developed to construct multi-species duality functionals, and many other algebraic properties and connections to non-symmetric Macdonald polynomials and partition functions are discussed in \cite{BW2018p}. The transition probability has also been discussed \cite{tracy2013asymmetric,kuan2019probability,lee2018exact}.

Much less is known rigorously about dynamic properties for multi-species models. The fluctuation exponent of $n$-ASEP was addressed in finite size scaling of the gap of its generator \cite{arita2009spectrum}, and the Bethe ansatz for the 2-ASEP with open boundaries was considered in \cite{Zhang_2019}. Limit distributions for a single second class particle have been studied in \cite{FNG2019, Nejjar2019}. To our knowledge, full limiting distributions for multi-species models have not been derived other than where there is a relation to a single species model such as shift colour-position symmetry \cite{borodin2019colorposition,borodin2020shiftinvariance}. 

In this paper we consider the two-species Arndt-Heinzel-Rittenberg model on $\mathbb{Z}$. This is a solvable lattice model and we first derive its transition probability (or Green's function) in the form of multiple contour integrals over two families of integration variables by diagonalising the time evolution generator using the nested Bethe ansatz method. Summing over initial and final positions leads to a joint current distribution in the same form. Then we use a proposition which allows us to rewrite one set of multiple integral into a Fredholm determinant.    
Finally we rigorously derive a limiting joint current distribution for late times from an exact formula for a total crossing probability. This limiting distribution factorises in a Tracy-Widom GUE distribution function and a Gaussian.   

\subsection{The Arndt-Heinzel-Rittenberg model}
In this paper we study a two-species asymmetric exclusion process which is different from $n$-ASEP, and establish a result on the long time limit for a certain joint distribution of currents. The model we consider was first introduced in 1997 by Arndt-Heinzel-Rittenberg in \cite{AHR-ptv}
and we refer to it as the AHR model below. We study (a special case of) the AHR exclusion process on the one dimensional lattice $\mathbb{Z}$, with two species of particles. Namely, each site can be either empty or occupied by one of the two kinds of particles, which are called ``$+$'' (plus) and ``$-$'' (minus) particles. The plus particle can hop forward while the minus particle can hop backward. In addition, an adjacent $+-$ pair of particles can swap to $-+$, but the swap can only occur in one direction (see the explicit jumping rates below). Obviously, since the AHR model is defined to be a Markov process, the hopping and swapping occurs according to exponential clocks. Additionally, the hopping and swapping is suppressed if the neighbouring site is occupied, i.e. a site cannot be occupied with more than one particle so that the model is a  proper exclusion process.
The explicit jumping rates are listed below.
\begin{equation}
\begin{array}{cc}
(+,0)\rightarrow (0,+)& \textrm{with rate } \beta,
\\[2mm]
(0,-)\rightarrow (-,0)& \textrm{with rate } \alpha,
\\[2mm]
(+,-)\rightarrow (-,+)& \textrm{with rate } 1.
\end{array}
\end {equation}
\begin{figure}[t]
\begin{center}
\begin{tikzpicture}[scale=0.8]
\node[draw,circle,inner sep=0pt,fill, text=white](p1) at (1,0) {$+$};
\node[draw,circle,inner sep=0pt,fill, text=white](p2) at (2,0) {$+$};
\node[draw,circle,inner sep=0pt,fill, text=white](p3) at (3,0) {$+$};
\node[draw,circle,inner sep=0pt,fill, text=white](p4) at (5,0) {$+$};
\node[draw,circle,inner sep=0pt,fill, text=white](p5) at (7,0) {$+$};

\node[draw,circle,inner sep=0pt](m1) at (4,0) {$-$};
\node[draw,circle,inner sep=0pt](m2) at (9,0) {$-$};
\node[draw,circle,inner sep=0pt](m3) at (11,0) {$-$};
\node[draw,circle,inner sep=0pt](m4) at (12,0) {$-$};

\node[draw,circle,inner sep=1pt](o1) at (6,0) {};
\node[draw,circle,inner sep=1pt](o2) at (8,0) {};
\node[draw,circle,inner sep=1pt](o3) at (10,0) {};

\draw (p1) -- (p2) -- (p3) -- (m1) -- (p4) -- (o1) -- (p5) -- (o2) -- (m2) -- (o3) -- (m3) -- (m4);

\draw [dashed] (0,0) -- (1,0);
\draw [dashed] (12,0) -- (13,0);

\draw[thick,<->] (3,0.5) arc (160:20:0.5) ;
\node[scale=1] at (3.5,1.1){$1$};
\draw[thick,->] (5,0.5) arc (160:20:0.5) ;
\node[scale=1] at (5.5,1.1){$\beta$};
\draw[thick,<-] (8,0.5) arc (160:20:0.5) ;
\node[scale=1] at (8.5,1.1){$\alpha$};

\draw[thick,color=red,<->] (5,-0.5) arc (-20:-160:0.5) ;
\node[scale=1.5,color=red] at (4.52,-.83){${\times}$};
\draw[thick,color=red,->] (7,-0.5) arc (-20:-160:0.5) ;
\node[scale=1.5,color=red] at (6.52,-.83){${\times}$};
\draw[thick,color=red,->] (9,-0.5) arc (-160:-20:0.5) ;
\node[scale=1.5,color=red] at (9.52,-.83){${\times}$};
\end{tikzpicture}
\end{center}
\caption{Configuration of $+$ and $-$ particles and hopping rates in the AHR model on $\mathbb{Z}$}
\label{ahr}
\end{figure}
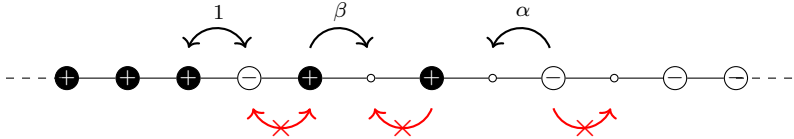
See Fig. \ref{ahr}.

The Yang-Baxter integrability of (this case of) the AHR model was proved in \cite{C2008}. In \cite{AHR-ptv}, the authors studied the stationary state of the model on a ring and observed an interesting condensation phenomena\cite{AHR1-ptv,AHR2-ptv}, which was further studied in  \cite{RSS2000} by using a connection to the Al-Salam-Chihara polynomials.
The hydrodynamics of the AHR was studied numerically in \cite{kim2007dynamic}, where it was observed that the model can be described by two coupled Burgers equations that decouple at large length scales.
Here, we specifically consider the case where the rates $\alpha$ and $\beta$ sum up to unity, $\alpha+\beta=1$.
It is known that in this case the stationary measure is factorised. In addition, this condition drastically simplifies technicalities in the
construction of the eigenfunctions of the time evolution generator in the form of a product of plane waves, and hence also in the derivation of the transition probability of the AHR model. It would be a quite interesting and challenging problem to generalise our results to the case with more general values of the parameters. 

\subsection{Transition probability}
In Section~\ref{se:greenfunction} we give a formula for the transition probability (or the Green's function) for the AHR model on $\mathbb{Z}$ in the form of a multiple integral.
This formula is a generalisation of the formula for the single species TASEP, first given by
Sch\"utz in \cite{S1997}, see also \cite{tracy2008integral}. Our proof will also be similar to the ones in \cite{S1997,tracy2008integral}, i.e. we
show explicitly that the multiple integral satisfies the correct time evolution equation and initial condition.

The origin of the integrand of the formula may be understood by considering a connection to the Bethe ansatz, as explained in Appendix~\ref{appx:bethe wave fn}.
A big difference from the single species case is that the form of the transition probability depends strongly on
the ordering of particles at initial and final times. In this paper we will focus on a special case
in which initially all $+$ particles are to the left of all $-$ particles, they swap their positions completely and at final
time all $-$ particles are to the left of all $+$ particles. For this particular case, the interaction effects between the $+$
and $-$ particles are contained in a nice product form in the integrand of the transition probability.
For the moment it seems difficult to do asymptotics for general orderings.

As for the initial condition, we consider a mix of random and step initial conditions, namely, we
consider the situation in which initially there are $n$ of $+$ particles with density $\rho$ to the left of the origin
while $m$ of $-$ particles occupy the first $m$ sites to the right of the origin, see Fig.~\ref{stepbernoulli}. Let $\mathbb{P}_{n,m}$ denote
the probability measure for this initial condition and $N_{\pm}(t)$ the number of $\pm$ particles which
passed the origin up to time $t$. For this setup, our first main result is the following theorem.
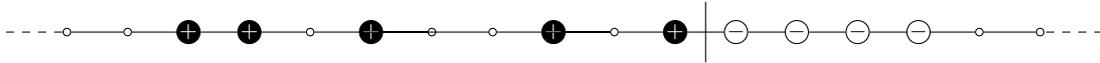
\begin{figure}[h]
\begin{center}
\begin{tikzpicture}[scale=0.8]
\node[draw,circle,inner sep=0pt,fill, text=white](p1) at (3,0) {$+$};
\node[draw,circle,inner sep=0pt,fill, text=white](p2) at (4,0) {$+$};
\node[draw,circle,inner sep=0pt,fill, text=white](p3) at (6,0) {$+$};
\node[draw,circle,inner sep=0pt,fill, text=white](p4) at (9,0) {$+$};
\node[draw,circle,inner sep=0pt,fill, text=white](p5) at (11,0) {$+$};

\node[draw,circle,inner sep=0pt](m1) at (12,0) {$-$};
\node[draw,circle,inner sep=0pt](m2) at (13,0) {$-$};
\node[draw,circle,inner sep=0pt](m3) at (14,0) {$-$};
\node[draw,circle,inner sep=0pt](m4) at (15,0) {$-$};

\node[draw,circle,inner sep=1pt](o1) at (1,0) {};
\node[draw,circle,inner sep=1pt](o2) at (2,0) {};
\node[draw,circle,inner sep=1pt](o3) at (5,0) {};
\node[draw,circle,inner sep=1pt](o4) at (7,0) {};
\node[draw,circle,inner sep=1pt](o5) at (8,0) {};
\node[draw,circle,inner sep=1pt](o6) at (10,0) {};
\node[draw,circle,inner sep=1pt](o7) at (16,0) {};
\node[draw,circle,inner sep=1pt](o8) at (17,0) {};

\draw (o1) -- (o2) -- (p1) -- (p2) -- (o3) -- (o4) -- (p3) -- (o5) -- (o6) -- (p4) -- (o6) -- (p5) -- (m1) -- (m2) -- (m3) -- (m4) -- (o7) -- (o8);

\draw [dashed] (0,0) -- (1,0);
\draw [dashed] (17.1,0) -- (18,0);
\draw (11.5,-0.5) -- (11.5,0.5);
\end{tikzpicture}
\end{center}
\caption{Step-Bernoulli initial configuration of the AHR model on $\mathbb{Z}$}
\label{stepbernoulli}
\end{figure}

\begin{theorem}
\label{prop:P1rho}
For the AHR model with $\alpha+\beta=1$ and the above step-Bernoulli initial condition, the probability that all particles passed the origin at time $t$ is given as the following multiple integral,

\begin{multline}
 \mathbb{P}_{n,m}[N_+(t)=n, N_-(t)=m] ={(-1)}^{n+m} 
 \oint \prod_{j=1}^n \frac{\dd z_j }{2\pi \ii} \prod_{k=1}^m
\frac{\dd w_k }{2\pi \ii}\ \e^{\Lambda_{n,m} t} \times \\
\frac{\displaystyle \rho^n \prod_{1\le i<j\le n} (z_i-z_j) \prod_{1\le k< \ell \le m} (w_\ell -w_k) \prod_{j=1}^n z_j^{n-j} \prod_{k=1}^m w_k^{k-1}}
{\displaystyle \prod_{j=1}^n (z_j-1)^{n+1-j} \big(1-(1-\rho) z_j\big)
\prod_{k=1}^m (w_k-1)^{k} \prod_{j=1}^n \prod_{k=1}^m \Big(\alpha z_j + \beta w_k\Big)},
\label{Prho1}
\end{multline}
where $\Lambda_{n,m}=\beta \sum_{j=1}^n (z_j^{-1}-1) + \alpha \sum_{k=1}^m (w_k^{-1}-1)$, and all contour integrals are around the origin.
\end{theorem}

\subsection{Transformation to Fredholm determinant}

Let us recall the definition of the continuous and discrete versions of the Fredholm determinant, which we will use in this paper. Let $K$ be an integral operator acting on functions $f \in L^2(s,\infty)$ with kernel $K(\zeta,\xi)$ given by 
$$
(Kf)(\zeta) := \int_{s}^\infty K(\zeta,\xi) f(\xi) \dd \xi.
$$
A Fredholm determinant of the operator $1+\lambda K$ is formally defined as the series 
\begin{align} 
\det(1+\lambda K)_{L^2(s,\infty)} &=  1 + \sum_{k=1}^{\infty}{ \frac{\lambda^k}{k!} } \int_{s}^{\infty}{ \int_{s}^{\infty}{ \cdots \int_{s}^{\infty}{ \det_{1 \leq i , j \leq k}\left[ K(\xi_i, \xi_j) \right] \dd \xi_1 \ldots \dd \xi_k } } }. \label{eq:FHCdef}
\end{align}
A discrete analogue is defined by
\begin{align}
\det(1+\lambda K)_{\ell^2(\mathbb{N})} &= 1 + \sum_{k=1}^{\infty}{\frac{\lambda^k}{k!} } \sum_{x_1 = 1}^{\infty}{ \sum_{x_2 = 1}^{\infty}{ \cdots \sum_{x_k = 1}^{\infty}{ \det_{1 \leq i , j \leq k}\left[ K(x_i, x_j) \right] } } }. \label{eq:FHDdef}
\end{align}
The series converge when $K$ satisfies certain conditions, e.g. when it is trace-class. For further details on Fredholm determinants in the context of Macdonald processes we refer to \cite{BC2014}, and to \cite{lax02,GohbergKrein69} for more general theory.

In the case of single species models, a useful approach to establish late time results is to rewrite multiple integral expressions into a Fredholm determinant for which asymptotics is easier to perform. For example, for the most standard case of the GUE random matrix eigenvalues, or in fact for general determinantal point processes, such a rewriting is well known and is explained in \cite{Mehta2004random,forrester2010random,tcASEPdet,S1997}.

In our multi-species case it has not been known whether the AHR model is associated with a determinantal process or not, but the following rewriting
in terms of auxiliary variables (or Fourier modes) has turned out to be useful in our analysis.  A very similar and related manipulation was stated in \cite{IS2019} for the case related to the $q$-Whittaker function (see also \cite{IMuS2020} for its application to a more general setting), but here we give a statement and its proof (in Section 4) for the case with several parameters but without $q$.   
In the following, we denote by $[N,M]$ the set $\{ N, \dots, M \} \subset \mathbb{N}$ for $N, M \in \mathbb{N}$ satisfying $N<M$. 

\smallskip
\begin{prop}
\label{prop:fredholm}
Set 
\begin{equation} 
g^c(\zeta,s ; \kappa) = \prod_{j=1}^\nu \frac{1}{1-u_j/\zeta}
\prod_{k=1}^\mu \frac{1}{1+v_k \zeta} \times \zeta^{\mu-\nu-s}  
{(\zeta + c)}^\kappa \e^{\gamma\zeta}, \label{gmndef} 
\end{equation}
where $\nu,\mu,\kappa\in\mathbb{N}$, $s\in\mathbb{Z}$ and $\zeta,c,u_j,v_k,\gamma\in\mathbb{C}$ for $j\in [1,\nu],k\in [1,\mu]$
and $g(\zeta,s)=g^c(\zeta,s;0)$. 
Let $I_{\nu}$ be a multiple integral of a form
\begin{equation}
\label{Inu}
I_\nu
=
\frac{1}{\nu!} \oint_C \prod_{i=1}^{\nu}
\frac{\dd \zeta_i}{2\pi \ii\,a_i^s \zeta_i}
\frac{\prod_{1\le i\neq j\le \nu} (1-\zeta_i/\zeta_j)}{\prod_{1\le i,j\le \nu}(1-a_i/\zeta_j)}
\prod_{i=1}^{\nu} \frac{g(\zeta_i,s)}{g(a_i,s)},
\end{equation}
where $a_i\in\mathbb{C}$, $i\in [1,\nu]$ and the contour $C$ includes $0,a_i,u_j,-1/v_k$ 
for $i\in [1,\nu], j\in [1,\nu],k\in [1,\mu]$. 

We assume that there exists $c\in\mathbb{C}$ such that 
an open disc of radius $r = \min_{1 \leq j \leq \nu}{( |a_j + c| )}$ centered at $-c$, denoted by $S(-c,r)$, includes the poles at $0, u_j$ and $-1/v_k$ for $j\in[1,\nu]$ and $k\in[1,\mu]$. 

Then this multiple integral can be written as a Fredholm determinant
\begin{equation}
\label{Inu_Fredform}
 I_{\nu} = \prod_{l=1}^{\mu}{ a_l^{-s }}  \det(1-K^{c})_{\ell^2(\mathbb{N})},
\end{equation}
where the kernel is written in the form
\begin{equation*}
 K^{c}(x,y) = \sum_{k=0}^{\nu-1} \phi^{c}_k(x) \psi^{c}_k(y)
\end{equation*}
with
\begin{align}
\phi^{c}_k(x)
&=
\oint_D \frac{\dd \xi}{2\pi\ii}
\frac{\xi^{k-\nu}}{ g^{c}(\xi,s ; x )
(\xi-a_1)\cdots (\xi-a_{k+1})} ,
\label{phi}
\\
\psi^{c}_k(x)
&=
a_{k+1} \oint_{C_r} \frac{\dd \zeta}{2\pi\ii\, (\zeta + c)  \zeta^{k-\nu+1}}
g^{c}(\zeta,s ; x)(\zeta-a_1)\cdots (\zeta-a_k)
\label{psi}.
\end{align}
The contour $D$ includes the poles at $a_i$ for $i\in[1,\nu]$ and 
the contour $C_r$ includes the poles at $0,u_j,-1/v_k$ for $j\in[1,\nu], k\in[1,\mu]$ 
(see Fig. \ref{fig:prop fred}). 

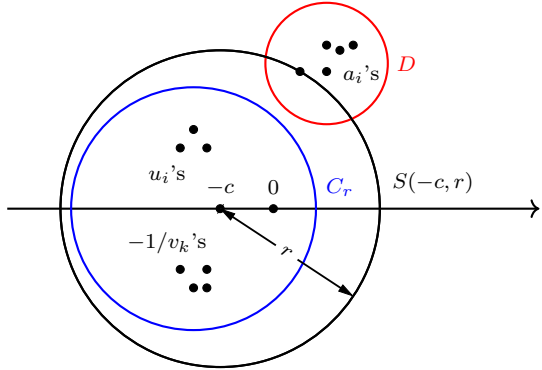
\begin{figure}[h!]
\begin{center}
\begin{tikzpicture}[scale=3.5]
			
\draw[->,thick] (-0.8,0) -- (1.2,0);
%\draw[->,thick] (0,-0.6) -- (0,0.6);
			
\draw[thick] (0,0) circle (0.6);
\node at (0.8,0.1){$S(-c,r)$};

\draw[fill=black] (0.2,0) circle (0.015);
\node at (0.2,0.08) {$0$};
\draw[fill=black] (0,0) circle (0.015);
\node at (0,0.08) {$-c$};

\draw[fill=black] (0.3,0.5196) circle (0.015);
\draw[fill=black] (0.4,0.5196) circle (0.015);
\draw[fill=black] (0.4,0.6196) circle (0.015);
\draw[fill=black] (0.5,0.6196) circle (0.015);
\draw[fill=black] (0.45,0.6) circle (0.015);
\node at (0.53,0.51){$a_i$'s};

\draw[thick,red] (0.4,0.55) circle (0.23);
\draw[thick,blue] (-0.1,0) circle (0.46);
\draw[thick] (0,0) circle (0.6);
\node at (0.45,0.08){$\color{blue}{C_r}$};
\node at (0.7,0.55){$\color{red}{D}$};

\draw[fill=black] (-0.1,0.3) circle (0.015);
\draw[fill=black] (-0.15,0.23) circle (0.015);
\draw[fill=black] (-0.05,0.23) circle (0.015);
\node at (-0.2,0.13){$u_i$'s};

\draw[fill=black] (-0.1,-0.3) circle (0.015);
\draw[fill=black] (-0.15,-0.23) circle (0.015);
\draw[fill=black] (-0.05,-0.23) circle (0.015);
\draw[fill=black] (-0.05,-0.3) circle (0.015);
\node at (-0.2,-0.13){$-1/v_k$'s};

\coordinate (O) at (0,0);
\coordinate (A) at (0.3,0.5196);
\coordinate (B) at (0.5,-0.3317);

\dimline[line style = {line width=0.7},extension start length=0,
extension end length=0] {(B)}{(O)}{$r$};			
		
\end{tikzpicture}
\end{center}
\caption{An illustration of the contours and poles in Proposition \ref{prop:fredholm}. The black circle is the boundary of the disk $S(-c,r)$, the blue circle is $C_r$ and the red circle is $D$.}
\label{fig:prop fred}
\end{figure}
\end{prop}

\bigskip
Existing applications of Proposition \ref{prop:fredholm} so far have used only the case $c=0$ where the conditions on the contours $D,C_r$ translate to those used in \cite{IS2019}. 
In this paper, $c$ is introduced to make it possible to perform the asymptotic analysis of $I_\nu$, i.e., $c$ prevents the right hand side of \eqref{Inu_Fredform} from diverging when $a_j$ goes to unity for all $j \in [1, \nu ]$ and $v_k$ approaches 1, as will be necessary in the following.

\subsection{Limiting distribution and nonlinear fluctuating hydrodynamics}
We can perform an asymptotic analysis of the multiple integral formula from Theorem~\ref{prop:P1rho} using the rewriting in terms of a Fredholm determinant as in Proposition~\ref{prop:fredholm}. A novel feature compared to the single species case is that one encounters dynamic poles arising from additional integration variables in the kernel of a Fredholm determinant. In our case we will see that, by taking certain poles at the beginning, one can evaluate the effects of the interaction and observe that in fact the two sets of variables decouple asymptotically as the parameter $t$ tends to infinity. As a consequence we can study the long time limit of the joint distribution, which is our second main result.

To state our result let us recall the definitions \eqref{eq:FHDdef} and \eqref{eq:FHCdef} of the Fredholm determinant, and set $\Ai(x)$ to denote the Airy function \cite{TW1994,bowick1991universal,forrester1993spectrum,moore1990matrix} defined by 
\begin{equation}\label{def:airy_function}
\Ai(x) = \frac{1}{2 \pi \ii} \int_C \exp(z^3/3-zx) \dd z,
\end{equation}
where the contour $C$ starts at $\i \e^{- \pi \ii /3}$ and goes to $\i \e^{\pi \ii /3}$.

\medskip

\begin{defn}[\cite{TW1994}]\label{def:GUETW_dist}
 The Fredholm determinant $F_2(s):=\det(1 - A)_{L^2(s,\infty)}$ is a distribution function called GUE Tracy-Widom distribution.
 Its kernel is the Airy kernel: 
 \begin{equation} \label{def:airy_kernel}
     A(x,y)= \int_0^{\i} \Ai(x+\lambda) \Ai(y+\lambda) \dd \lambda .
 \end{equation}
\end{defn}

\begin{theorem}
 \label{th:limitdistr}
For the AHR model with $\alpha=\beta=\frac12$ and the step-Bernoulli initial condition above, we have, in the long time limit, 
\begin{align}
	\lim_{t\to\infty}  \mathbb{P}_{n,m}[N_+(t)=n,N_-(t)=m]
	=F_{\text {2}}(s_{2})\cdot F_{\text G}(s_{\rm g}), 
	\label{PnmFF}
\end{align}
where on the left hand side we use the scaling (\ref{nmscale}) with (\ref{c2cg}),(\ref{jpm}).  

$F_2$ and $F_G$ on the right hand side are the distribution functions of the GUE Tracy-Widom and the standard Gaussian distributions respectively.
\end{theorem}
The specialisation to $\alpha=\beta=\frac12$ is just for a simplicity. One can generalise our whole calculations and arguments to the case with $\alpha+\beta=1$.

This theorem can be proved as written, but one may wonder what the meaning is of the variables $s_2,s_{\rm g}$ and the reason for the appearance of the limiting distribution in (\ref{PnmFF}). In fact, these can be understood from nonlinear fluctuating hydrodynamics (NLFHD) which is a heuristic physics theory for studying the long time behaviour of one dimensional multi-component systems.

NLFHD was first proposed in \cite{van2012exact,spohn2014nonlinear} to provide concrete predictions for the long time behaviour of one dimensional Hamiltonian dynamics with nonlinear interaction. One first writes down the hydrodynamic equations for
three conserved quantities of the original system, stretch, momentum and energy, and takes fluctuation effects into account by adding white noise. Diagonalising the advection term in the hydrodynamic part one switches to normal modes which have intrinsic propagating speeds.

It is natural to study fluctuations of the normal modes. A key idea of NLFHD is that if the speeds of the normal modes are different, then the interaction among them should be irrelevant in the long time limit and thus the fluctuations for each mode would be described by the single species noisy-Burgers equation (a.k.a. KPZ equation). Fluctuations in the long time limit would then generically be given by the Tracy-Widom type distributions.

NLFHD has also been formulated for stochastic models and gives concrete predictions for the distribution of currents for multispecies models \cite{ferrari2013coupled}. A prototypical model is in fact the AHR model, which has two obvious conserved
quantities, the number of $+$ and that of $-$ particles. The hydrodynamic equation is given by
\begin{align}
	\frac{\partial \boldsymbol \rho(x,t) }{\partial t}+\frac{\partial\boldsymbol{\mathsf j}(\boldsymbol \rho(x,t))}{\partial x}=0,
\end{align}
where $\boldsymbol \rho(x,t) = (\rho_+(x,t),\rho_-(x,t))$ is the density vector and
$\boldsymbol{\mathsf j}(\boldsymbol \rho) = (\mathsf j_{+}(\boldsymbol \rho), \mathsf j_{-}(\boldsymbol \rho))$ denotes the macroscopic current of $\pm$ particles given by
\begin{align}
	&\mathsf j_{+} (\boldsymbol \rho)=\rho_+(1-\rho_+-\rho_-) + 2\rho_+\rho_-,\label{currenth}\\[2mm]
	&\mathsf j_{-} (\boldsymbol \rho)=-(1-\rho_+-\rho_-)\rho_- -2\rho_+\rho_- \ . \label{currentc}
\end{align}
The normal modes for the AHR model follow from diagonalisation of the Jacobian matrix $\partial \boldsymbol{\mathsf j}/ \partial \boldsymbol \rho$.

In \cite{mendl2016searching}, a step initial condition was studied. As a particular case related to our study, let us consider a mix of Bernoulli and step initial conditions, namely, we consider the situation in which initially the left half of $\Z$ is filled randomly with $+$ particles with density $\rho$ \footnote{This $\rho$ should be distingished from $\boldsymbol \rho$, which appears only in this subsection.} while the right half of $\Z$ is filled with $-$ particles, as in Fig.~\ref{stepbernoulli}. This seems the simplest nontrivial initial condition
for the two species exclusion process. Note that the special case with $\rho=1$ is simpler but this case can be easily seen to be equivalent to the single species of TASEP with step initial conditions.

For the step-Bernoulli initial condition, one can consider fluctuations of the normal modes, and NLFHD predicts the following. Firstly, the variables $s_{\rm g}(n,m,t)$ and $s_2(n,m,t)$ defined below (\ref{jpm})
are just the scaling variables
for the fluctuations in the normal directions, and the resulting prediction of NLFHD relevant to our problem is the following:
\begin{multline}
    \mathbb{P}_{\infty,\infty}[N_+(t)=n,N_-(t)=m]\simeq \\
	\left|\det
	\begin{pmatrix}
		\displaystyle \frac{\partial s_{2}(n,m;t)}{\partial n}		&\displaystyle\frac{\partial s_{2}(n,m;t)}{\partial m}\\[3mm]
		\displaystyle \frac{\partial s_{{\rm g}}(n,m;t)}{\partial n}		&\displaystyle \frac{\partial s_{{\rm g}}(n,m;t)}{\partial m}
	\end{pmatrix}
	\right|
	\partial F_2(s_{2}(n,m;t))\cdot \partial F_{\text G}(s_{{\rm g}}(n,m;t)).
	\label{infdistri}
\end{multline}
For the $s_{\rm g}$ mode, one expects the Gaussian rather than Tracy-Widom due to the strong effects from the initial randomness.

In the case of single species TASEP, the study of the current distribution for the step initial condition is reduced to a problem of finite number of particles because each particle in TASEP can not affect the dynamics of particles in front. For the AHR model this kind of property does not hold any more and therefore the problem for the step type initial condition in the previous paragraph cannot be reduced to one with only finite number of particles. We can therefore not directly establish the above prediction of NLFHD using the formula \eqref{Prho1} for the transition probability because this holds only for a finite number of particles.

Instead we can generalise the prediction of the NLFHD to the case with large but finite number of particles. This generalisation is nontrivial and a full discussion will be given elsewhere.
 The main idea is that we may assume that after the two species of particles become separate the fluctuations are transported by the hydrodynamics. Then the probability of interest $\lim_{t\to\infty} \mathbb{P}_{n,m}[N_+(t)=n,N_-(t)=m]$ would be written as a sum of contributions of $\lim_{t\to\infty} \mathbb{P}_{\infty,\infty}[N_+(t)=i,N_-(t)=j]$ from a certain region only. Integrating \eqref{infdistri} over the independent modes $s_2,s_{\rm g}$ leads to a product of $F_2$ and $F_G$, which is exactly the limiting distribution appearing in Theorem~\ref{th:limitdistr}.

The main results in this paper were announced in \cite{PRL}. The purpose of this paper is to provide full details of the intricate and elaborate calculations as well as the mathematical proofs that establish our results. As already mentioned above, details of the physics oriented discussion of NLFHD for the case with finite number of particles falls outside the scope of this paper and will be given in a separate publication. 

The rest of the paper is organised as follows. In Section~\ref{se:greenfunction} we give a multiple integral formula for the transition probability, or Green's function. In Section~\ref{se:joint current}, we give a multiple integral formula for a certain joint distribution of the currents.
In Section~\ref{se:tofredholm} we explain a rewriting of a multiple integral to a Fredholm determinant. In Sections~\ref{se:scalinglimit}, \ref{sec: I1 limit} and \ref{sec:Asymptotics_second} we perform asymptotics and establish Theorem~\ref{th:limitdistr}. In the appendices we provide further technical details about asymptotic estimates as well as details related to the Bethe ansatz for the AHR model, and symmetrisation identities for multi-variable rational functions. In addition we provide a detailed derivation of a crucial decoupling identity. Details of some calculations and proofs are given in  Appendices from A through F.

\section{Transition probability}
\label{se:greenfunction}
We consider the AHR model with $n$ particles of one type, denoted by plus, and $m$ particles of a second type, denoted by minus. Define the set of coordinates by
\begin{equation}
\label{coordinates}
\mathbb{W}^k:= \{\vec{x}=(x_1,\dots,x_k)\in\mathbb{Z}^k: x_1<x_2<\dots<x_k\}.
\end{equation}
Suppose $P(\vec{x},\vec{y};t)$ is the probability that at time $t$, plus particles are sitting at position $\vec{x}=(x_1,\dots,x_n)\in \mathbb{W}^n$, while the minus particles are at $\vec{y}=(y_1,\dots,y_m)\in \mathbb{W}^m$, and let $\Omega^{n+m}$ be the physical domain of $(\vec{x},\vec{y})$, i.e. we write $(\vec{x},\vec{y}) \in \Omega^{n+m}$ if $(\vec{x},\vec{y})\in \mathbb{W}^n \times \mathbb{W}^m$ with the additional condition that $\vec{x} \cap \vec{y} = \emptyset$. From the dynamics of the system described above, we can write down the master equation for $P(\vec{x},\vec{y};t)$ in the AHR model for the cases where all particles are far apart from each other,
\begin{align}
\label{master eq AHR}
\frac{\dd}{\dd t}P(\vec{x},\vec{y};t)
=
\beta\sum_{i=1}^{n}P(\vec{x}_i^-,\vec{y};t)+
\alpha\sum_{j=1}^{m}P(\vec{x},\vec{y}_j^+;t)
-(\beta n+\alpha m)P(\vec{x},\vec{y};t),
\end{align}
where $\vec{x}_i^{\pm}:=(x_1,\dots,x_i {\pm} 1,\dots,x_n)$. We emphasise again that \eqref{master eq AHR} only describes the time evolution of hops but without swaps and exclusions.

On the right hand side of \eqref{master eq AHR}, the positive terms describe the arrival part of the change in $P(\vec{x},\vec{y};t)$. Namely, the state $(\vec{x}_i^-,\vec{y})$ turns into the state $(\vec{x},\vec{y})$ at rate $\beta$, while $(\vec{x},\vec{y}_j^+)$ turns into $(\vec{x},\vec{y})$ at rate $\alpha$. The negative term on the right hand side comes from the loss part of the change in $P(\vec{x},\vec{y};t)$. Specifically, the state $(\vec{x},\vec{y})$ evolves into the state $(\vec{x}_i^+,\vec{y})$ at rate $\beta$ and the state $(\vec{x},\vec{y}_j^-)$ at rate $\alpha$.

In order to describe the interactions between particles, i.e. the exclusions and swaps, one could include appropriate kronecker delta functions on the right hand side of \eqref{master eq AHR}. Alternatively, one can impose the following boundary conditions. For all appropriate $i,j$,
\begin{itemize}
  \item Interactions between plus particles:
  \begin{align}\label{AHR bdrycond1}
  P(x_1,\dots,x_i,x_{i+1}=x_{i},\dots,x_n;\vec{y})
  =
  P(x_1,\dots,x_i,x_{i+1}=x_{i}+1,\dots,x_n;\vec{y}).
  \end{align}
  \item Interactions between minus particles:
  \begin{align}\label{AHR bdrycond2}
  P(\vec{x};y_1,\dots,y_{i-1}=y_i,y_i,\dots,y_m)
  =
  P(\vec{x};y_1,\dots,y_{i-1}=y_i-1,y_i,\dots,y_m).
  \end{align}
  \item Interactions between plus and minus particles:
  \begin{equation}\label{AHR bdrycond3}
  \begin{split}
  P(\vec{x};
  y_1,\dots,y_j=x_i+1,\dots,y_m)
  =
  \beta P(\vec{x};
  y_1,\dots,y_j=x_i,\dots,y_m)+\\
  \alpha P(x_1,\dots,x_{i-1},x_i+1,x_{i+1},\dots,x_n;
  y_1,\dots,y_j=x_i+1,\dots,y_m).
  \end{split}
  \end{equation}
\end{itemize}

We observe that \eqref{master eq AHR}--\eqref{AHR bdrycond3} are well defined for $P(\vec{x};\vec{y})$ on $\mathbb{Z}^n\times\mathbb{Z}^m$. However, we know that the position states are defined only on the physical coordinates $(\vec{x},\vec{y})\in \Omega^{n+m}$. The solution of \eqref{master eq AHR}--\eqref{AHR bdrycond3} indeed gives the transition probability of the AHR model on $\Omega^{n+m}$ , but $P(\vec{x};\vec{y})$ is no longer a probability on $(\mathbb{Z}^n\times\mathbb{Z}^m) \setminus(\Omega^{n+m})$. We refer to Appendix \ref{appx:bethe wave fn} for further details on the boundary conditions.

Imposing the initial condition that the positions of plus and minus particles at $t=0$ are given by  $\vec{x}^{(0)}$ and $\vec{y}^{(0)}$, respectively, the probability $P(\vec{x},\vec{y};t)$ is called the \textit{transition probability} or \textit{Green's function} of the model. We emphasise that this probability is independent of the absolute values of positions, but depends only on the relative position of the initial and final states, and we denote it by $G(\vec{x}-\vec{x}^{(0)},\vec{y}-\vec{y}^{(0)},t)$.
The initial condition is given by
\begin{equation}
G(\vec{x}-\vec{x}^{(0)},\vec{y}-\vec{y}^{(0)},0)
=
\prod_{j=1}^n \delta(x_j-x_j^{(0)}) \prod_{k=1}^m \delta(y_j-y_j^{(0)}).
\label{AHR initialcond}
\end{equation}
A transition probability $G(\vec{x}-\vec{x}^{(0)},\vec{y}-\vec{y}^{(0)},t)$ is the function that satisfies the master equation \eqref{master eq AHR}, the boundary conditions \eqref{AHR bdrycond1}--\eqref{AHR bdrycond3}, as well as the initial condition \eqref{AHR initialcond}. Generally, the transition probability is constructed as an integral form
\begin{align*}
G(\vec{x}-\vec{x}^{(0)},\vec{y}-\vec{y}^{(0)},t)
=
\oint_{C} \prod_{j=1}^n \frac{\dd z_j }{2\pi \ii}
\prod_{k=1}^m \frac{\dd w_k }{2\pi \ii}
\e^{\Lambda_{n,m} t}  A(\vec{z},\vec{w})  \,
\psi(\vec{x};\vec{y}) \prod_{j=1}^{n} z_{j}^{-x^{(0)}_j-1} \prod_{k=1}^m w_{k}^{y^{(0)}_k-1},
\end{align*}
where $\psi(\vec{x};\vec{y})$ and $\Lambda_{n,m}$ are the Bethe wave function and eigenvalue of the Markov generator found by the Bethe ansatz. $A(\vec{z},\vec{w})$ and the contour $C$ are chosen such that the initial condition is satisfied where $\vec{z}$ and $\vec{w}$ denote the collection of variables $z_j$'s with $j\in [1,n]$ and $w_k$'s with $k\in [1,m]$, respectively. The transition probability for the single species TASEP was derived in \cite{S1997}, and for single species ASEP in \cite{tracy2008integral}. In Appendix \ref{appx:bethe wave fn} we provide details of the construction of the transition probability using the Bethe ansatz. Here we give the final explicit integral form for $G(\vec{x}-\vec{x}^{(0)},\vec{y}-\vec{y}^{(0)},t)$ and prove that it satisfies all desired properties \eqref{master eq AHR}--\eqref{AHR initialcond}.

\begin{theorem}
\label{th:Green1}
Consider the AHR model on $\mathbb{Z}$ with $n$ plus and $m$ minus particles. Define $r_j$ as the number of plus particles to the right of the $j$\textsuperscript{th} minus particles at time $t$, and $r_j^{(0)}$ the same quantity at $t=0$, i.e.,
\begin{align*}
r_j:=&\#\{x_i\in\vec{x}\mid x_i \geq y_j\},\qquad
r_j^{(0)}:=\#\{x_i^{(0)}\in\vec{x}^{(0)}\mid x_i^{(0)} \geq y_j^{(0)}\}.
\end{align*} Then the transition probability is given by
\begin{multline}
G(\vec{x}-\vec{x}^{(0)},\vec{y}-\vec{y}^{(0)},t)
=
\oint
\prod_{j=1}^n \frac{\dd z_j }{2\pi \ii}
\prod_{k=1}^m \frac{\dd w_k }{2\pi \ii}\,
\e^{\Lambda_{n,m} t}
\sum_{\pi\in S_n} \sign(\pi) \sum_{\sigma\in S_m} \sign(\sigma) \times\\
\prod_{k=1}^{m}
\prod_{j=1}^{r_k}
\frac{1}{\alpha z_{\pi_{n-j+1}}+\beta w_{\sigma_k}}
\prod_{k=1}^{m}(\beta w_k)^{r_k^{(0)}}
\prod_{j=1}^n \left( \frac{z_j-1}{z_{\pi_j}-1}\right)^{j} z_{\pi_j}^{x_j} z_{j}^{-x^{(0)}_j-1}
\prod_{k=1}^m  \left( \frac{w_k-1}{w_{\sigma_k}-1}\right)^{-k}  w_{\sigma_k}^{-y_k} w_{k}^{y^{(0)}_k-1},
\label{Green function}
\end{multline}
where
\begin{equation}
\label{eigval}
\Lambda_{n,m}=\beta \sum_{j=1}^n (z_j^{-1}-1) + \alpha \sum_{k=1}^m (w_k^{-1}-1),
\end{equation}
and all integral contours only contain poles at the origin.
\end{theorem}

\begin{proof}
The proof consists of four parts. We show below that \eqref{Green function} satisfies the master equation \eqref{master eq AHR}, the boundary conditions (\ref{AHR bdrycond1})-(\ref{AHR bdrycond3}), the initial condition \eqref{AHR initialcond} and finally that the solution is unique. 

\begin{enumerate}[label=(\roman*)]
\item \textbf{Proof of master equation (\ref{master eq AHR})}

Clearly the integrand is $C^1$ continuous in time. By interchanging ${\rm d}/{\rm d} t$ with the integral and the explicit form of $\Lambda_{n,m}$, the evolution equation \eqref{master eq AHR} follows.\\

\item \textbf{Proof of boundary conditions (\ref{AHR bdrycond1})--(\ref{AHR bdrycond3})}
The proof of the boundary conditions (\ref{AHR bdrycond1})--(\ref{AHR bdrycond3}) is elementary but detailed, we give it in Appendix~\ref{se:bcs}. \\

%\medskip

\item\textbf{Proof of initial condition (\ref{AHR initialcond})}
Firstly, we show by mathematical induction that \eqref{Green function} at $t=0$ has vanishing residues unless $\pi_i=i$ and $x_i=x_i^{(0)}$ for all $i$. To see this, we first prove the case for $i=1$.

If $x_1 \geq x_1^{(0)}$, then since the components of $\vec{x}$ are well ordered, $x_j \geq x_1^{(0)}$ for all $j$. Consider now an arbitrary permutation $\pi$ such that $\pi_k=1$. The exponent of $z_1$ in the integrand of \eqref{Green function} is given by $z_1^{x_k-x_1^{(0)}-1}$, and the exponent in this expression is always non-negative unless $k=1$ and $x_1=x_1^{(0)}$. If these conditions are not met, then at time $t=0$ the integrand is analytic at $z_1=0$ and therefore has a vanishing residue.

For the induction step we assume first that \eqref{Green function} vanishes at $t=0$ unless $x_i=x_i^{(0)}$ and $\pi_i=i$ for all $i\leq \ell-1$. If $\pi_k=\ell$, the exponent of $z_{\ell}$ is $x_k-x_{\ell}^{(0)}-1$, and by the induction hypothesis we have that $k \geq \ell$, and hence the exponent is non-negative. This implies that \eqref{Green function} is zero unless $\pi_{\ell}=\ell$ and $x_{\ell} = x_{\ell}^{(0)}$. Therefore, we can conclude that $G(\vec{x}-\vec{x}^{(0)},\vec{y}-\vec{y}^{(0)},0)$ is nonzero only when $x_i=x_i^{(0)}$ for all $i$ and $\pi=\textsf{id}$.

When $x_i=x_i^{(0)}$ for all $i \in  [1,n]$ and $\pi=\textsf{id}$, the after integration over the $\vec{z}$-variables the transition probability \eqref{Green function} at $t=0$ becomes
\begin{equation*}
G(\vec{x}-\vec{x}^{(0)},\vec{y}-\vec{y}^{(0)},0)
=
\oint
\prod_{k=1}^m \frac{\dd w_k }{2\pi \ii}\,
\sum_{\sigma\in S_m} \sign(\sigma)
\prod_{k=1}^{m}\left(\frac{w_k}{w_{\sigma_k}}\right)^{r_k}
\prod_{k=1}^m  \left( \frac{w_k-1}{w_{\sigma_k}-1}\right)^{-k}  w_{\sigma_k}^{-y_k} w_{k}^{y^{(0)}_k-1}.
\end{equation*}

\medskip

We now prove that the above function is nonzero only when $y_j=y_j^{(0)}$ for all $j\in [1,m]$ and $\sigma=\textsf{id}$. While plus particles hop to the right, minus particles always hop to the left and so we have that $y_j^{(0)}\geq y_j$ for all $j \in [1, m]$.  Consider an arbitrary permutation $\sigma$ such that $\sigma_k=j$ with $j>k$. At time $t=0$, the exponent of $w_j$ in the integrand in \eqref{Green function} is given by $w_j^{y_j^{(0)}-y_k+r_j-r_k-1}$. By the definition of $r_j$, we must have $y_j-y_k+r_j-r_k>0$ for any $j>k$. Hence $y_j^{(0)}-y_k+r_j-r_k-1\geq y_j-y_k+r_j-r_k-1>0$. Therefore, the exponent of $w_j$ is always positive unless $\sigma_j=j$. When $\sigma=\textsf{id}$, the exponent of $w_j$ becomes $y_j^{(0)}-y_j-1$, which results in a vanishing residue unless $y_j^{(0)}= y_j$.

It remains to show that the transition probability is normalised as $G(0,0,0)=1$. This can be easily seen by the residue theorem.

\medskip

\item\textbf{Proof of uniqueness}

The transition probability is a solution of a master equation with bounded initial condition, and the number of particle jumps for a given time $t$ in the AHR model is bounded by a Poisson random variable with parameter given by a constant times $t$. The global existence and uniqueness is therefore guaranteed by general considerations as provided in Proposition~4.9 and Appendix C of \cite{bcs}.
\end{enumerate}
This concludes the proof for the Green's function \eqref{Green function}.
\qed
\end{proof}

\section{Joint current distribution}
\label{se:joint current}
In the following, we are interested in the probability that all plus and minus particles have crossed the origin at time $t$. Exact expressions for such joint current distributions under different initial conditions can be derived using the transition probability. We will focus here on the case in which initially $n$ plus particles are distributed by the Bernoulli measure with density $\rho$ at negative integers, and the first $m$ sites at the non-negative integers are occupied by minus particles. We called such initial condition the step-Bernoulli initial condition in the Section~\ref{se:introduction}. Note that setting $\rho= 1$ corresponds to an initial condition which we may call  a step-step initial condition because initial conditions are step type for both plus and minus particles. In the following we will denote the probability that all particles passed the origin at time $t$ by $P_{n,m,\rho}(t)$, i.e.

\begin{equation}
    P_{n,m,\rho}(t) := \mathbb{P}_{n,m}[N_{+}(t)=n, N_{-}(t)=m],
\end{equation}
and call it the joint current distribution. In terms of $P_{n,m,\rho}(t)$, we recall  Theorem~\ref{prop:P1rho}:

\medskip

\noindent\textbf{Theorem~\ref{prop:P1rho}.}
\textit{Under the step-Bernoulli initial condition with density $\rho$, the joint current distribution $P_{n,m,\rho}(t)$ is given by}

\begin{multline}
P_{n,m,\rho}(t)
={(-1)}^{n+m} \oint \prod_{j=1}^n \frac{\dd z_j }{2\pi \ii} \prod_{k=1}^m
\frac{\dd w_k }{2\pi \ii}\ \e^{\Lambda_{n,m} t} \times \\
\frac{\displaystyle \rho^n \prod_{1\le i<j\le n} (z_i-z_j) \prod_{1\le k< \ell \le m} (w_\ell -w_k) \prod_{j=1}^n z_j^{n-j} \prod_{k=1}^m w_k^{k-1}}
{\displaystyle \prod_{j=1}^n (z_j-1)^{n+1-j} \big(1-(1-\rho) z_j\big)
\prod_{k=1}^m (w_k-1)^{k} \prod_{j=1}^n \prod_{k=1}^m \Big(\alpha z_j + \beta w_k\Big)},
\label{Prho}
\end{multline}
\textit{where $\Lambda_{n,m}$ is given by \eqref{eigval} and all contours are taken around the origin.}

\medskip

\noindent
This theorem is proved by summing over the initial and final coordinates of the particles in the transition probability, and by making use of the following lemma. We provide the proof of this lemma in Appendix~\ref{se:symm identities}.

\medskip

\begin{lemman}
\label{lem:symm identity}
\begin{equation*}
\sum_{\pi\in S_n}\sign(\pi)\prod_{i=1}^n
\left(\frac{z_{\pi_i}-1}{z_{\pi_i}}\right)^i
\frac{1}{1-(1-\rho)\prod_{j=1}^iz_{\pi_j}}
=
\frac{\prod_{1\leq i<j \leq n}(z_j-z_i)
\prod_{i=1}^{n}(z_i-1)}
{\prod_{i=1}^{n}z_i^n(1-(1-\rho)z_i)}.
\end{equation*}
\end{lemman}

\noindent {\it Proof of Theorem~\ref{prop:P1rho}.}
For the given initial and final coordinates, all the plus particles initially are to the left of all the minus particles, and at time $t$ end up to the right of all minus particles. We therefore have that $r_j=n$ for all $j \in [1, m]$ and the final transition probability is given by
\begin{multline}
G(\vec{x}-\vec{x}^{(0)},\vec{y}-\vec{y}^{(0)},t)
=
\oint
\prod_{j=1}^n \frac{\dd z_j }{2\pi \ii}
\prod_{k=1}^m \frac{\dd w_k }{2\pi \ii}\,
\e^{\Lambda_{n,m} t}\!
\prod_{k=1}^m \prod _{j=1}^{n} \frac{1}{\alpha z_{j} + \beta w_{k}}\times
\\
\sum_{\pi\in S_n} \sign(\pi)
\prod_{j=1}^n \left( \frac{z_j-1}{z_{\pi_j}-1}\right)^{j} z_{\pi_j}^{x_j} z_{j}^{-x^{(0)}_j-1}
\sum_{\sigma\in S_m} \sign(\sigma)
\prod_{k=1}^m  \left( \frac{w_k-1}{w_{\rho_k}-1}\right)^{-k}  w_{\rho_k}^{-y_k} w_{k}^{y^{(0)}_k-1}.
\label{G-+}
\end{multline}

In the Bernoulli measure, the distances among the initial positions of the plus particles are independently distributed and each is distributed as a geometric random variable with parameter $1-\rho$, i.e., the probability that plus particles locating at $\vec{x}^{(0)}$ is given by
\begin{equation}
\Pi(\vec{x}^{(0)};0)
=
\rho^n \prod_{j=1}^{n-1} (1-\rho)^{x_{j+1}^{(0)}-x_j^{(0)}-1} \cdot (1-\rho)^{-x_n^{(0)}-1}
=
\left(\frac{\rho}{1-\rho}\right)^n (1-\rho)^{ -x_1^{(0)}}.
\end{equation}
The joint current distribution $P_{n,m,\rho}(t)$ thus is the sum of
$$G(\vec{x}-\vec{x}^{(0)},\vec{y}-\vec{y}^{(0)},t) \Pi(\vec{x}^{(0)};0),$$ over all final coordinates $0\leq x_1<x_2<\cdots<x_n$, $y_1<y_2\cdots<y_m\leq -1$ and all initial coordinates $x_1^{(0)}<\cdots<x_n^{(0)}\leq-1$, with $G(\vec{x}-\vec{x}^{(0)},\vec{y}-\vec{y}^{(0)},t)$ given by \eqref{G-+} and $y_j^{(0)}=j-1$. The sums of initial and final coordinates are calculated by taking geometric series, and the sums of permutations $\pi,\sigma$ in \eqref{G-+} are computed using Lemma~\ref{lem:symm identity} with $\rho=0$, resulting in, 

\begin{multline*}
P_{n,m,\rho}(t)
=
\oint \prod_{j=1}^n \frac{\dd z_j }{2\pi \ii} \prod_{k=1}^m
\frac{\dd w_k }{2\pi \ii}\ \e^{\Lambda_{n,m} t} \times \\ \frac{\rho^n \prod_{1\le i<j\le n} (z_i-z_j) \prod_{1\le k< \ell \le m} (w_\ell -w_k) \prod_{j=1}^n z_j^{n-j} \prod_{k=1}^m w_k^{k-1}}
{\prod_{j=1}^n (z_j-1)^{n+1-j} (1-(1-\rho) z_1 \cdots z_j)
\prod_{k=1}^m (w_k-1)^{k} \prod_{j=1}^n \prod_{k=1}^m \Big(\alpha z_j + \beta w_k\Big)}.
\end{multline*}
The theorem follows from a symmetrisation in the $\vec{z}$-variables using Lemma~\ref{lem:symm identity} and subsequent de-symmetrisation using 
\[
\frac{\prod_{1\leq i<j\leq n}(z_i-z_j) \prod_{1\leq i<j\leq n}(z_j-z_i) }{\prod_{i=1}^n(z_i-1)^{n-1}} = \sum_{\pi\in S_n} \prod_{1\leq i<j\leq n}(z_{\pi_j}-z_{\pi_i}) \prod_{i=1}^n \left(\frac{z_{\pi_i}}{z_{\pi_i}-1}\right)^{i-1}.
\]
\qed

In order to prove our main result,  Theorem~\ref{th:limitdistr}, we rewrite the joint current distribution $P_{n,m,\rho}(t)$ in the form \eqref{Prho} as \eqref{eq:prob2terma} in Section~\ref{se:scalinglimit}. It is suitable for asymptotic analyses when $n<m$, which is compatible with our scaling \eqref{nmscale}. 
In the remaining part of this section let us consider the opposite case $n\geq m$, 
because this case is simpler to analyse and is still useful to 
understand what kind of asymptotics will be considered for the $n<m$ case.
When $n\geq m$, at late times we would only observe plus particles crossing the origin while all minus particles have already crossed. Hence
$P_{n,m,\rho}(t)$ in this region is expected to be  asymptotically close to the current distribution of the single species TASEP \cite{J2000,PS2002,BFPS2007}.

Suppose now that $n\geq m$ holds (in fact the following arguments are valid when $n \geq m-2$ holds). In this case we can evaluate the contour integrals over the $\vec{w}$-variables in \eqref{Prho} for the following reasons. The eigenvalue $\Lambda_{n,m}$ given by \eqref{eigval} introduces an essential singularity at the origin. It is therefore convenient to replace the contours as enclosing all other possible poles except the origin, and including $\infty$. First let us consider the residue at $\infty$. The degree of $w_k$ near $\infty$ in the integrand of \eqref{Prho} is $(m-1+k-1)-(k+n)=m-2-n \leq 0$ so that the residue at $\infty$ is zero. Next consider the simple poles at $w_k=-\alpha z_j/\beta$. Such residues will cancel the variable $z_j$ in $\Lambda_{n,m}$, making the integrand analytic in $z_j$ and hence give rise to a zero residue at $z_j=0$.

The only poles with non-zero contribution are therefore at $w_k=1$. The residue of the simple pole at $w_1=1$ can be easily evaluated. The factor $\prod_{\ell =2}^m (w_\ell -w_1)$ in such residues will decrease the order of each pole at $w_k=1$ by one. Hence the second order pole at $w_2=1$ becomes a simple pole and its residue can therefore be simply evaluated subsequently. The integration of $w_2$ again reduces the order of the pole at $w_3=1$ by one, and hence it also becomes a simple pole. Evaluating all poles at $w_k=1$ sequentially, we arrive at the following result,

\begin{equation}
P_{n,m,\rho}(t)
=
\frac{\rho^n}{n!}\oint \prod_{j=1}^n \frac{\dd z_j }{2\pi \ii}\,
\e^{\Lambda_{n,0} t} \frac{\prod_{1\le i<j\le n} (z_i-z_j)^2 }
{\prod_{j=1}^n (1-(1-\rho) z_j)(z_j-1)^{n}
\prod_{j=1}^n \left(\alpha z_j + \beta \right)^m}, \ \mathrm{for} \ n \geq m-2 ,
\label{Prhoz}
\end{equation}
where the contours are all around the origin, and we have made use of the symmetrisation identity
\begin{equation*}
\sum_{\pi\in S_n} \sign(\pi) \prod_{j=1}^n z_{\pi_j}^{n-j} (z_{\pi_j}-1)^{j-1}
=
\prod_{1 \leq j<k \leq n} (z_j-z_k) ,
\end{equation*}
which can be proved using a Vandermonde determinant.

Standard asymptotic analysis \cite{J2000,BFPS2007} then shows that the long time limit of $P_{n,m,\rho}(t)$ is governed by the GUE Tracy-Widom distribution \cite{TW1994,Mehta2004random,forrester2010random}, which is the same as the asymptotic current distribution for the single species TASEP, as expected.

For a special case, we also 
have an exact correspondence to the single species TASEP. 

\medskip
\begin{corollaryn}
When $n=m$, $\alpha=\beta=1/2$ and $\rho=1$, we have, by applying the change of variable $z_j = x_j/(2-x_j)$ in \eqref{Prhoz},
\begin{align}
P_{n,n,\rho}(t)
=
\frac{1}{n!}\oint \prod_{j=1}^n \frac{\dd x_j }{2\pi \ii}\ \e^{\mathcal{E} t}
\frac{\prod_{1\le i<j\le n} (x_i-x_j)^2 } {\prod_{j=1}^n (x_j-1)^{n}} ,
\label{SymTASEP}
\end{align}
where
$$
\mathcal{E}=\sum_{j=1}^n (x_j^{-1}-1),
$$
and the contours are around the origin, or alternatively around $1,\infty$. 
\end{corollaryn}
\medskip
This recovers the distribution for the single species TASEP for the step initial condition. 
This is understood easily because the AHR model for this particular setting is 
equivalent to the single species TASEP by regarding $+$ particle as a particle and and 
$-$ particle as a hole.

\section{Fredholm determinant}
\label{se:tofredholm}
In this section, we describe a general method of converting a multiple contour integral into a Fredholm determinant.  
In other words, we will give a proof of Proposition \ref{prop:fredholm}, which was stated in Section~\ref{se:introduction}.
The main idea is to transform the integrand into determinants using the Cauchy identity. The integral can then be converted into a determinant according to the Cauchy-Binet formula, and subsequently into a Fredholm determinant.

For $\nu\in\mathbb{N}$, let us now first 
write down again the $\nu$-fold integral (\ref{Inu}) we want to consider, which reads 
\begin{equation}
\label{Inu2}
I_\nu
=
\frac{1}{\nu!} \oint_C \prod_{i=1}^{\nu}
\frac{\dd \zeta_i}{2\pi \ii\,a_i^s \zeta_i}
\frac{\prod_{1\le i\neq j\le \nu} (1-\zeta_i/\zeta_j)}{\prod_{1\le i,j\le \nu}(1-a_i/\zeta_j)}
\prod_{i=1}^{\nu} \frac{g(\zeta_i,s)}{g(a_i,s)},
\end{equation}
where $g(\zeta,s)=g^c(\zeta,s;0)$ with
\begin{equation} 
g^c(\zeta,s ; \kappa) = \prod_{j=1}^\nu \frac{1}{1-u_j/\zeta}
\prod_{k=1}^\mu \frac{1}{1+v_k \zeta} \times \zeta^{\mu-\nu-s}  
{(\zeta + c)}^\kappa \e^{\gamma\zeta},
\label{gmndef2} 
\end{equation}
and $\mu,\kappa\in\mathbb{N}$, $s\in\mathbb{Z}$ and $\zeta,c,a_j, u_j,v_k,\gamma\in\mathbb{C}$ for $j\in [1,\nu]$ and $k\in [1,\mu]$. 
The contour $C$ includes $0,a_j,u_j,-1/v_k$ for $j\in [1,\nu]$ and $k\in [1,\mu]$.

Let us assume temporarily that all $a_j$'s are distinct.
The factor $\prod_{i\neq j} (1-\zeta_i/\zeta_j)/\prod_{i,j}(1-a_i/\zeta_j)$ in the integrand can be written as a product of two determinants via the Cauchy determinant identity, 
\begin{align*}
\frac{\prod_{i\neq j} (1-\zeta_i/\zeta_j)}{\prod_i \zeta_i  \prod_{i,j}(1-a_i/\zeta_j)}
 =
\frac{\prod_{i\neq j} (\zeta_i-\zeta_j)}{\prod_{i,j}(\zeta_i-a_j)}
=
\det\left(\frac{1}{\zeta_i-a_j}\right)_{1\le i,j\le \nu}
\det\left(\frac{1}{\zeta_k-a_\ell}\right)_{1\le k,\ell\le \nu}
\frac{\prod_{i,j}(\zeta_i-a_j)}{\prod_{i\neq j }(a_i-a_j)} .
\end{align*}
Recall the Cauchy-Binet identity (or Andreief identity):
\begin{equation*}
\frac{1}{\nu!}\int
\det\big(f_i(x_j)\big)_{1\le i,j\le \nu}
\det\big(h_i(x_j)\big)_{1\le i,j\le \nu}
\prod_{i=1}^\nu\dd u(x_i)
=
\det\left(\int f_i(x)h_j(x)\dd u(x)\right)_{1\le i,j\le \nu},
\end{equation*}
from which, the $\nu$-fold integral is written as a single determinant,
\begin{align*}
I_\nu
&=
\frac{1}{\nu!} \oint_C \prod_{j=1}^\nu \left[
\frac{\dd \zeta_j}{2\pi\ii\, a_j^s}
\frac{g(\zeta_j,s )}{g(a_j,s)}
\frac{\prod_{\ell}(\zeta_j-a_\ell)}{\prod_{ \ell\neq j}(a_j-a_\ell)} \right]
\det\left(\frac{1}{\zeta_i-a_j}\right)_{1\le i,j\le \nu}
\det\left(\frac{1}{\zeta_k-a_\ell}\right)_{1\le k, \ell \le \nu}
\\
&=
\det\left( \oint_C
\frac{\dd \zeta}{2\pi\ii}
\frac{a_j^{-s}}{(\zeta-a_j)(\zeta-a_k)} \frac{g(\zeta,s)}{g(a_j,s)}
\frac{\prod_{\ell} (\zeta-a_\ell)}{\prod_{\ell\neq j} (a_j-a_\ell)}
\right)_{1\le j,k\le \nu}.
 \end{align*}
From this explicit form it is clear that the poles at $\zeta=a_j$ and $\zeta=a_k$ are removable unless $j=k$, hence after evaluating residues at $\zeta=a_j$ for all $j$, we obtain
\begin{align}
I_\nu
&=\prod_{l=1}^{\nu}a_l^{-s}
\det\left(
\delta_{jk}+
\oint_{C_r}
\frac{\dd \zeta}{2\pi\ii}
\frac{1}{(\zeta-a_j)(\zeta-a_k)}
\frac{g(\zeta,s)}{g(a_j,s)}
\frac{\prod_{\ell} (\zeta-a_\ell)}{\prod_{\ell\neq j} (a_j-a_\ell)}
\right)_{1\le j,k\le \nu},
\label{prefredholm}
\end{align}
where the contour $C_r$ includes only the poles at $0, \alpha_j$ and $-1/\beta_k$, but not those at $a_j$. To transform the integral around $C_r$ into a product of two operators (or matrices)  we rewrite the factor $-1/(\zeta-a_j)$ in the form of a geometric series,
\begin{align}
-1/(\zeta-a_j) = \sum_{x=1}^{\infty} (\zeta + c  )^{x-1} /(a_j + c )^x.
\label{eq:cintro}
\end{align}
Therefore the integral $I_\nu$ becomes

\begin{align}
I_\nu
=\prod_{l=1}^{\nu}a_l^{-s}
\det\left(\delta_{jk}-
\sum_{x=1}^{\infty}
\oint_{C_r}
\frac{\dd \zeta}{2\pi\ii ( \zeta + c ) } \frac{g^c(\zeta,s ; x )}{g^c(a_j,s ; x )}
\frac{\prod_{\ell\neq k}(\zeta-a_\ell)}{\prod_{\ell\neq j} (a_j-a_\ell)}
\right)_{1\le j,k\le \nu}.
\label{eq:Inuwithx}
\end{align}
Note that, by the assumptions of Proposition \ref{prop:fredholm}, one can deform the contour $C_r$ and find $c$ such that $|(\zeta +c) /(a_j +c )|\le \epsilon<1$ holds for some positive constant $\epsilon$ and all $j \in [1, \nu]$,  so that the geometric sum converges uniformly.

The sum over $x\ge 1$ in \eqref{eq:Inuwithx} can be interpreted as representing a product of two matrices, $A^c$ and $B^c$, of dimensions $\nu\times \infty$ and $\infty\times \nu$. Namely 
one can write 
\begin{equation*}
I_\nu
=\prod_{l=1}^{\nu}a_l^{-s}
\det(1 - A^c B^c )_{1\le j,k\le \nu} ,
\end{equation*}
where the matrices $A^c$ and $B^c$ are given explicitly by
\begin{equation*}
A^c(k,x)
=
\oint_{C_r} \frac{\dd \zeta}{2\pi\ii ( \zeta +c )}
g^c(\zeta,s ; x ) \prod_{\ell\neq k} (\zeta-a_\ell) ,
\qquad\quad
B^c(x,j)
=
\left(a_j^{s}g^c(a_j, s ; x )\prod_{\ell\neq j}(a_j-a_\ell)\right)^{-1} .
\end{equation*}
Then by swapping the product order of these two matrices, the integral $I_\nu$ is converted into a Fredholm determinant,  
\begin{equation*}
I_\nu
=\prod_{l =1}^{\nu}a_l^{-s}
\det(1 - K^c )_{\ell^2(\mathbb{N})},
\end{equation*}
where the kernel of the operator $K^c:=B^c A^c$ is given by
\begin{equation}
K^c(x,y)=\sum_{j=1}^\nu B^c(x,j)A^c(j,y).
\end{equation}
Note that, although $A^c$, $B^c$ and the kernel $B^c A^c$ depend on $c$, $I_\nu$ is independent of $c$ because $c$  stemmed from the non-unique way that $1/(\zeta- a_j)$ can be represented as a geometric series in \eqref{eq:cintro}.

For the purpose of later discussion, we will show that the kernel of the Fredholm determinant can be written into a product of two contour integrals, so that the method like the steepest descent can be applied for asymptotic analyses. The kernel $K^c$ is given by
\begin{align}
\nonumber
K^c(x,y)
&=
\sum_{j=1}^\nu B^c(x,j)A^c(j,y)
\\
\nonumber
&=
\sum_{j=1}^\nu \frac{1}{g^c(a_j, s ; x )
\prod_{\ell \neq j} (a_j-a_\ell)}
\oint_{C_r} \frac{\dd \zeta}{2\pi \ii ( \zeta +c ) } g^c(\zeta, s ; y )
\prod_{\ell\neq j} (\zeta-a_{\ell})
\\
&=
\oint_{C_r} \frac{\dd \zeta}{2\pi \ii ( \zeta +c )} \oint_D \frac{\dd \xi}{2\pi\ii} \frac{1}{\zeta-\xi}
\prod_{\ell=1}^\nu \frac{\zeta-a_\ell}{\xi-a_\ell} \frac{g^c(\zeta, s; y )}{g^c(\xi,  s ; x )} 
\label{kernel form 2}
\\
\nonumber
&=
\oint_{C_r} \frac{\dd \zeta}{2\pi \ii ( \zeta +c )} \oint_D \frac{\dd \xi}{2\pi\ii} \frac{1}{\zeta-\xi}\left(\frac{\xi^{\nu}}{\zeta^{\nu}}\prod_{\ell=1}^\nu \frac{\zeta-a_\ell}{\xi-a_\ell}-1\right) 
\frac{\zeta^{\nu}g^c(\zeta, s ; y )}{\xi^{\nu}g^c(\xi, s ; x )} .
\end{align}
In the sequel the condition $|(\zeta +c) /(a_j +c )|\le \epsilon<1$ is not 
necessary for the contour $C_r$. 
The contour $D$ for the $\xi$ integration includes the poles at $a_j$ and should therefore be separated from $C_r$. The term $-1$ inside the parentheses is inserted for convenience of the next step and does not change the value of the integral.

Using a simple identity in \cite{IS2019},
\begin{equation}
\frac{1}{\zeta-\xi}
\left(\frac{\xi^\nu}{\zeta^\nu}\prod_{\ell=1}^\nu \frac{\zeta-a_\ell}{\xi-a_\ell}-1\right)
=
\sum_{k=0}^{\nu-1} a_{k+1}\frac{(\zeta-a_1)\cdots (\zeta-a_k)\xi^k}{(\xi-a_1)\cdots (\xi-a_{k+1})\zeta^{k+1}},
\label{tosum}
\end{equation}
the kernel can be written in the form,
\begin{equation}
\label{kernel form summation}
 K^c(x,y) = \sum_{k=0}^{\nu-1} \phi^c_k(x) \psi^c_k(y),
\end{equation}
with
\begin{align}
\phi^c_k(x)
&=
\oint_D \frac{\dd \xi}{2\pi\ii}
\frac{\xi^{k-\nu}}{ g^c(\xi, s ; x )
(\xi-a_1)\cdots (\xi-a_{k+1})} ,
\label{phi}
\\
\psi^c_k(x)
&=
a_{k+1} \oint_{C_r} \frac{\dd \zeta}{2\pi\ii\, (\zeta + c) \zeta^{k-\nu+1}}
g^c(\zeta, s ; x )(\zeta-a_1)\cdots (\zeta-a_k) 
\label{psi}.
\end{align}
Note that at this stage one can set some of $a_j$'s to be the same in these expressions. In the proof we have assumed that the $a_j$ are distinct but Proposition \ref{prop:fredholm} remains valid also when some or all of the $a_j$ are equal. Also after the use of (\ref{tosum}) the contours $C_r$ and $D$ do not need to be separate for the final result.

\section{Asymptotics: preparations}
\label{se:scalinglimit}
In this section, we will rewrite the joint current distribution in Proposition~\ref{prop:P1rho} to a form which is suitable for an asymptotic analysis. For technical simplicity, from now on, we set $\beta=\tfrac12$.
Using shorthand $\rho'=1-\rho$, recall the integral formula of the joint current distribution given in Proposition \ref{prop:P1rho} (with $\alpha=\beta=\frac12$):
\begin{align}
\label{eq:probmaster}
P_{n,m,\rho}(t) &= (-1)^{n+m}  \oint_0 \prod_{j=1}^n \frac{\dd z_j }{2\pi \ii} \prod_{k=1}^m
\frac{\dd w_k }{2\pi \ii}\ \e^{\Lambda_{n,m} t} \frac{\displaystyle \rho^n \prod_{1\le i<j\le n} (z_i-z_j) \prod_{1\le k< \ell \le m} (w_\ell -w_k) \prod_{j=1}^n z_j^{n-j} \prod_{k=1}^m w_k^{k-1}}
{\displaystyle \prod_{j=1}^n (z_j-1)^{n+1-j} (1- \rho' z_j)
\prod_{k=1}^m (w_k-1)^{k} \prod_{j=1}^n \prod_{k=1}^m \big(\tfrac12(z_j + w_k)\big)} \nonumber \\
&=  (-1)^{n+m} \frac{\rho^n }{n!m!} \oint_0 \dd^n z\, \dd^m w\, \frac{ \e^{\Lambda_{n,m} t}  \Delta_n(z) \Delta_n(-z) \Delta_m(w) \Delta_m(-w)}  
{\displaystyle \prod_{j=1}^n (1- \rho' z_j)  \prod_{j=1}^n (z_j-1)^{n} \prod_{k=1}^m (w_k-1)^{m} \prod_{j=1}^n \prod_{k=1}^m \big(\tfrac12(z_j + w_k)\big)},
 \end{align}
where the second line follows from symmetrisation in the $\vec{z}$-variables as well as in the $\vec{w}$-variables, and we defined the following abbreviations to make the formulas more compact,
\begin{equation*}
\dd^n z\, \dd^m w= \prod_{j=1}^n \frac{\dd z_j }{2\pi \ii} \prod_{k=1}^m
\frac{\dd w_k }{2\pi \ii},\qquad
\Delta_n(z) = \prod_{1 \le i < j\le n} (z_i-z_j).
\end{equation*}
We define
\[
\Lambda_{n,m} = \frac12 \sum_{j=1}^n (z_j^{-1}-1) + \frac12 \sum_{k=1}^m (w_k^{-1}-1),
\qquad
S_{n,m}(z,w)=\prod_{j=1}^n \prod_{k=1}^m \big(\tfrac12(z_j + w_k)\big).
\]
Before performing an asymptotic analysis, we first perform some rearrangements in \eqref{eq:probmaster} to disentangle the integral over the $\vec{z}$-variables from those over the $\vec{w}$-variables.

The contours in \eqref{eq:probmaster} are all around the origin and we can choose the $\vec{z}$-contours to be enclosing the $\vec{w}$-contours. By deforming the $\vec{z}$-contours we can place them around all the other poles, including infinity, in the complex plane. Since we will be interested in the asymptotics with the scaling 
(\ref{nmscale}), we can focus on the case 
when $n<m+4$, for which one can check that there is no pole at $z_j=\infty$ and the only other poles are located at $z_j=1$ and $z_j=1/\rho'$. Therefore in this case 
the target distribution becomes
\begin{align*}
%\label{eq:PP}
P_{n,m,\rho}(t)  =  \frac{ {(-1)}^m \rho^n  }{n!m!}\oint_{ 0} \dd^m w\,  \oint_{1,1/\rho'} \dd^nz\,  \frac{  \e^{\Lambda_{n,m} t} \Delta_n(z)  \Delta_n(-z)  \Delta_m(w) \Delta_m(-w)}  
{\prod_{j=1}^n (1- \rho' z_j)   \prod_{j=1}^n (z_j-1)^{n} \prod_{k=1}^m (w_k-1)^{m} S_{n,m}(z,w) }.
\end{align*}
The poles at $z_j=1/\rho'$ are simple poles and hence can be easily evaluated. Due to the symmetry of the integrand, the contribution of all poles at $z_j=1/\rho'$ is just $n$ times that of the pole $z_n=1/\rho'$. It can be easily seen that after evaluating the pole at $z_n=1/\rho'$, the Vandermonde product produces a factor $\prod_{j\neq n}(z_j-1/\rho')^2$, which cancels all other poles at $z_j=1/\rho'$. Therefore we obtain, when $n<m+4$, 
\begin{multline}
\label{eq:prob2terma}
P_{n,m,\rho}(t)  = \frac{ (-1)^m \rho^n  }{n!m!}\oint_{ 0} \dd^m w\,  \oint_{1} \dd^n z\,  \frac{ \e^{\Lambda_{n,m} t} \Delta_n(z) \Delta_n(-z) \Delta_m(w) \Delta_m(-w)} 
{  \prod_{j=1}^n (1- \rho' z_j) \prod_{j=1}^n (z_j-1)^{n} \prod_{k=1}^m (w_k-1)^{m} S_{n,m}(z,w) } \\
- \frac{  (-1)^{n+m-1}\e^{-\rho t/2}}{(\rho')^{n-1} (n-1)!m!}\oint_{ 0} \dd^m w\,  \oint_{1} \dd ^{n-1}z\,  \frac{ \e^{\Lambda_{n-1,m} t} \Delta_{n-1}(z) \Delta_{n-1}(-z) \Delta_m(w) \Delta_m(-w)} 
{ \prod_{j=1}^{n-1} (z_j-1)^{n} \prod_{k=1}^m (w_k-1)^{m} }\times \\
\frac{  \prod_{j=1}^{n-1}(1-\rho' z_j)}
{ S_{n-1,m}(z,w) \prod_{k=1}^m \big(\tfrac12( 1/\rho' + w_k)\big)}.
\end{multline}
In the next two sections, we will determine the asymptotic behaviours of each term of \eqref{eq:prob2terma}. To that end we first rewrite them into a form which is more suitable for asymptotics. 
Let $\mathcal{I}_1$ and $\mathcal{I}_2$ denote the first term and the second term of \eqref{eq:prob2terma}, respectively. Then the 
result of rewritings are summarised as follows.

\medskip

\begin{lemman}
\label{lem:prob2term}
The probability $P_{n,m,\rho}(t)$ can be written as

\begin{equation}
\label{eq:prob2termb}
P_{n,m,\rho}(t)  = \mathcal{I}_1 - \mathcal{I}_2,
\end{equation}
where $\mathcal{I}_1$ and $\mathcal{I}_2$ are given by

\begin{equation}
\label{eq:prob2term1}
\mathcal{I}_1 =  \frac{ (-1)^m }{m!}\oint_{ 0} \dd^m w\, \frac{  \e^{\Lambda_{0,m} t} \Delta_m(w) \Delta_m(-w)  }
{ \prod_{k=1}^m (w_k-1)^{m} S_{n,m}(1,w)} ,
\end{equation}

\begin{align}
\label{eq:prob2term2}
\mathcal{I}_2
=&
\frac{\e^{-\rho t/2}}{(\rho')^{n-1}} \left(\frac{2(1-\rho)}{2-\rho}\right)^m \frac{(-1)^{n-1}}{(n-1)!}  \oint_{1} \dd^{n-1} z\,  \frac{ \e^{\Lambda_{n-1,0} t} \Delta_{n-1}(z) \Delta_{n-1}(-z)  \prod_{j=1}^{n-1}(1-\rho' z_j)}
{\prod_{j=1}^{n-1} (z_j-1)^n S_{n-1,m}(z,1) }
I_w(\vec{z})\\
=& \oint_{1} \dd^{n-1} z\,  L(\vec{z})
I_w (\vec{z}),
\label{I2 Iz antisym}
\end{align}
with
\begin{align}
I_w (\vec{z}) := & \frac{{(-1)}^m}{m!}
\oint_{0}
\dd^m w
\frac{\e^{\Lambda_{0,m}t} \Delta_{m}(w) \Delta_m(-w) S_{n-1,m}(z,1) }{\prod_{k=1}^{m}{ (w_k-1)^m} S_{n-1,m}(z,w) }
\prod_{k=1}^{m}\frac{1+1/\rho'}{w_k+1/\rho'},
\label{Iw} \\
L(\vec{z}) := &(-1)^{n-1} \frac{\e^{-\rho t/2}}{(\rho')^{n-1}  } \left(\frac{2(1-\rho)}{2-\rho}\right)^m  \frac{ \e^{\Lambda_{n-1,0} t} \Delta_{n-1}(-z)  \prod_{j=1}^{n-1}(1-\rho' z_j)}
{\prod_{j=1}^{n-1} (z_j-1)^{j+1} S_{n-1,m}(z,1) }. \label{def:L(z)} 
\end{align}
Here and in the following $\vec{z}$ denotes the collection of variables $z_j$'s with $j\in [1,n-1]$.
\end{lemman}

\begin{proof}

First, we show that the first term $\mathcal{I}_1$ is written as \eqref{eq:prob2term1}.
Consider the symmetrisation identity
\begin{equation}
\label{Vandermonde_1}
\sum_{\pi \in S_{N}}{ \sign(\pi)\prod_{j=1}^{N}{ {\left( \frac{1}{z_{\pi_j}-1} \right)}^{j-1} } } = \frac{\Delta_{N}(z)}{ \prod_{j=1}^{N}{ {(z_j - 1)}^{N-1} } } ,
\end{equation}
which can be proved using a Vandermonde determinant.
Since the right hand side of \eqref{Vandermonde_1} for $N=n$ appears in the first term of \eqref{eq:prob2terma}, substituting the left hand side of \eqref{Vandermonde_1} into it and considering that a Vandermonde product $\Delta_{n}(-z)$ is anti-symmetric under the exchange $z_i \leftrightarrow z_j$ for any pairs $(i,j) \in {[1,n]}^2$, we obtain

\begin{multline}
\label{antisym_z-integral_of_I1}
 \frac{ (-1)^m  \rho^n  }{n!m!}\oint_{ 0} \dd^m w\,  \oint_{1} \dd^n z\,  \frac{ \e^{\Lambda_{n,m} t} \Delta_n(z) \Delta_n(-z) \Delta_m(w) \Delta_m(-w)}  
{ \prod_{j=1}^n (z_j-1)^{n} \prod_{k=1}^m (w_k-1)^{m} \prod_{j=1}^n (1- \rho' z_j)  S_{n,m}(z,w)} \\
= \frac{ (-1)^m  \rho^n  }{m!}\oint_{ 0} 
\dd^m w\,  \oint_{1} \dd^n z\,  \frac{  \e^{\Lambda_{n,m} t} \Delta_n(-z)  \Delta_m(w) \Delta_m(-w) }
{ \prod_{j=1}^n (z_j-1)^{j}\prod_{j=1}^n \prod_{k=1}^m (w_k-1)^{m} (1- \rho' z_j)  S_{n,m}(z,w)}.
\end{multline}
On the right hand side, the pole at $z_1=1$ is first order and its residue can be evaluated in an easy way. In doing so, the Vandermonde product produces a factor $z_{2}-1$, making the pole at $z_{2}=1$ first order. This pole can be subsequently evaluated in an easy way. Proceeding successively all poles at $z_j=1$ can be easily evaluated from $j=1$ to $j=n$, giving \eqref{eq:prob2term1}. 

Second, we show that the second term $\mathcal{I}_2$ can be written as either \eqref{eq:prob2term2} or \eqref{I2 Iz antisym}.
\eqref{eq:prob2term2} follows immediately from the second term of \eqref{eq:prob2terma}.
Since the right hand side of \eqref{Vandermonde_1} for $N=n-1$ appears in the right hand side of \eqref{eq:prob2term2}, in the same fashion as the derivation of \eqref{antisym_z-integral_of_I1}, we can obtain the anti-symmetric formula \eqref{I2 Iz antisym}.
Unlike the right hand side of \eqref{antisym_z-integral_of_I1}, the pole whose order is the lowest is at $z_1 = 1$ and its order is $2$ so that we can not evaluate the poles at $z_j=1$ explicitly in a compact manner.  \qed
\end{proof}

\subsection{Scaling limit}\label{sec:Scaling limit}
Nonlinear fluctuating hydrodynamics  \cite{landau2013course,MF1973,SH1977,DM1986} for this model, predicts that integrated currents for the scaled normal modes defined by 
\begin{equation}
\begin{split}
\eta_2(t) &= \frac{1}{c_2 t^{1/3}} \Big( (1+\rho) N_+(t)-(3-\rho) N_-(t)+\tfrac12(1-\rho) (1-(1-\rho)^2/4) t\Big),
\\
\eta_{\rm g}(t) &= \frac{1}{c_{\rm g} t^{1/2}} \Big( -2(2-\rho) N_+(t)+2\rho  N_-(t)+(2-\rho)(1-\rho)\rho t \Big),
\end{split}
\label{s+-}
\end{equation}
where the constants $c_2$ and $c_{\rm g}$ are given by
\begin{equation}
\begin{split}
c_2 &=(3/32)^{1/3}(1-\rho)(3-\rho)^{2/3}(1+\rho)^{2/3},
\\
c_{\rm g} &= 2^{-1/2} 3(1-\rho)^{3/2}\sqrt{\rho(2-\rho)},
\end{split}
\label{c2cg}
\end{equation}
would tend to finite limiting random variables in the large time limit. 
Put differently the joint distribution for the scaled normal modes, $\mathbb{P}[\eta_2 \leq s_2, \eta_{\rm g} \leq s_{\rm g}]$
is equivalent to that 
for the original variable $\mathbb{P}[N_+(t)\geq n, N_-(t)\geq m]$ 
if we take the scaling for $n$ and $m$ as  
\begin{equation}
\begin{split}
n &= j_+(\rho)t -\frac{1}{12(1-\rho)}(2\rho c_2 s_2 t^{1/3}+(3-\rho)c_{\rm g} s_{\rm g} t^{1/2}),
\\
m &= j_-(\rho)t -\frac{1}{12(1-\rho)}(2(2-\rho)c_2 s_2 t^{1/3}+(1+\rho) c_{\rm g} s_{\rm g} t^{1/2}),
\end{split}
\label{nmscale}
\end{equation}
where the macroscopic currents $j_\pm$ are defined by
\begin{equation}
     j_+(\rho) = \frac{\rho(3-\rho)^2}{16}, \qquad j_-(\rho)=\frac{(1+\rho)^2(2-\rho)}{16}.
     \label{jpm}
\end{equation}
Denote the result of solving (\ref{nmscale}) for $s_{\rm g}$ and $s_2$ as functions of $n,m,t$  
by $s_{\rm g}(n,m,t)$ and $s_2(n,m,t)$ respectively. Then they are nothing but the right hand 
sides of (\ref{s+-}) with $N_+(t),N_-(t)$ replaced with $n,m$. 

We will show that in this scaling limit, the probability \eqref{eq:prob2termb} tends to a product of the Gaussian distribution and the GUE Tracy-Widom distribution in the long time limit. In other words, we will prove the following theorem.

\medskip

\begin{theorem}
\label{thm_I1I2 limit}
With the scaling \eqref{nmscale}, the terms $\mathcal{I}_1$ and $\mathcal{I}_2$ given by \eqref{eq:prob2term1} and \eqref{eq:prob2term2}, respectively, tend to the following limits,
\begin{subequations}
\begin{align}
\lim_{t \to \infty} \mathcal{I}_1 & = F_2(s_{2}), \label{claim_I1 limit} \\
\lim_{t \to \infty} \mathcal{I}_2 & = \left[ 1 - F_G(s_{{\rm g}}) \right] F_2(s_{2}). \label{claim_I2 limit} 
\end{align}
\end{subequations}
Here we recall that $F_2$ and $F_G$ denote the cumulative distribution functions of the GUE Tracy-Widom distribution in Definition \ref{def:GUETW_dist} and the Gaussian distribution, respectively.
\end{theorem}
As a corollary of this theorem, we obtain the final result.

\medskip

\begin{corollaryn}
\label{cor_Pnm limit}
With the scaling \eqref{nmscale}, the probability $P_{n,m,\rho}(t)$ converges to the product of the Gaussian distribution and the GUE Tracy-Widom distribution, i.e.,
\begin{equation}
\lim_{t \to \infty} P_{n,m,\rho}(t) =  \lim_{t \to \infty} \left( \mathcal{I}_1 - \mathcal{I}_2 \right) = F_G(s_{{\rm g}}) F_2(s_{2}).
\end{equation}
\end{corollaryn}
The proofs of \eqref{claim_I1 limit} and \eqref{claim_I2 limit} will constitute the contents of 
sections \ref{sec: I1 limit} and \ref{sec:Asymptotics_second} respectively.

\section{Limit of $\mathcal{I}_1$: proof of \eqref{claim_I1 limit} in Theorem~\ref{thm_I1I2 limit} }\label{sec: I1 limit}

Let us start from the simplified formula of $\mathcal{I}_1$ \eqref{eq:prob2term1}.
The asymptotic behaviour of the right hand side of \eqref{eq:prob2term1} can be obtained in a standard way with an extra parameter $c$ where we first transform it to a Fredholm determinant according to the procedure outlined in Section~\ref{se:tofredholm}, and then perform a steepest descent analysis of the Fredholm kernel. To fit \eqref{Inu} into \eqref{eq:prob2term1}, we first change variables $w_k \rightarrow 1/w_k$ so that the contours lie around the poles at $w_k=\infty$, and then deform these contours to surround the poles other than the ones at $w_k=\infty$, resulting in
\begin{equation*}
\mathcal{I}_1=  \frac{{(-1)}^m  }{m!}\oint_{ 0} \dd^m w\,  \frac{  \e^{\Lambda_{0,m} t} \Delta_m(w) \Delta_m(-w)  }
{ \prod_{k=1}^m (w_k-1)^{m} S_{n,m}(1,w)}  =
\frac{1}{m!}
\oint_{0,1,-1} \dd^m w\,
\frac{  \e^{\tilde{\Lambda}_{0,m} t} \prod_{1\leq k \neq \ell\leq m}(w_\ell-w_k)}
       {\prod_{k=1}^m w_k^m (w_k-1)^m \prod_{k=1}^m (\frac12(1+1/w_k))^n},
\end{equation*}
where $\tilde{\Lambda}_{n,m}=\frac12\sum_{j=1}^n(z_j-1)+\frac12\sum_{k=1}^m(w_k-1)$. This expression can be written as the standard $m$-fold integral \eqref{Inu} in Proposition \ref{prop:fredholm},
\begin{equation*}
	\mathcal{I}_1 = \frac{1}{m!}\oint_{0,1,-1} \frac{\dd^m w}{ \prod_{k=1}^m w_k} \frac{\prod_{1\leq k\neq \ell\leq m}(1-w_k/w_\ell)}{\prod_{k,\ell=1}^m (1-a_\ell/w_k)} \prod_{k=1}^m \frac{g(w_k,0 )}{g(a_k,0)},
\end{equation*}
with
\[
\begin{array}{rll}
	\nu &=m, \quad &\mu=n, \quad \gamma = t/2,\quad s=0,\\
	u_i &=0, & 1\le i \le m, \\
	a_i &=1, & 1\le i\le m,\\
	v_k &=1, & 1\le k\le n,
\end{array}
\]
where $g(w_k,x) = g^c(w_k,x;0)$ and $g^c(w_k,x;y)$ is given in \eqref{gmndef},
\begin{equation*}
	\label{g_first_term}
	g^c(w,x;y)= \left(\frac{w}{w+1}\right)^n w^{-m -x} {(w + c)}^y \e^{wt/2}.
\end{equation*}

We choose $c>0$ and let $S(-c,|1+c|)$ be the open disc of radius $|1+c|$ centred at $-c$. Then $S(-c,|1+c|)$ includes poles at $0,-1$, since $|-c|<|1+c|$ and $|-1+c|<|1|+|c|=|1+c|$. Thus applying Proposition \ref{prop:fredholm}, the integral $\mathcal{I}_1$ defined in \eqref{eq:prob2term1} can thus be written as a Fredholm determinant:

\begin{equation}
\mathcal{I}_1 = \det(1-K^c)_{\ell^2(\mathbb{N})},\qquad K^c(x,y)=\sum_{k=0}^{m-1} \phi^c_k(x) \psi^c_k(y),
\label{eq:fredholm1}
\end{equation}
with
\begin{subequations}
\label{phipsiform1}
\begin{align}
\phi^c_k(x) &= \oint_1 \frac{\dd z}{2\pi\ii} \frac{1}{ {(z+c)}^{x}  (z-1)} \left(\frac{1+z}{z}\right)^n \left(\frac{z}{1-z}\right)^k \e^{-zt/2} ,
\label{phiform1}
\\
\psi^c_k(x) &= \oint_{0,-1}\frac{\dd w}{2\pi\ii} \frac{  {(w+c)}^{x-1}  }{w} \left(\frac{w}{1+w}\right)^n \left(\frac{1-w}{w}\right)^{k} \e^{ wt/2},
\label{psiform1}
\end{align}
\end{subequations}
in which the pole at $w=1$ is separated from the poles at $w=-1,0$.
As will be mentioned, we choose the contours of the integrals with respect to $z$ and $w$ as $\Gamma$ and $\Sigma$ that will be introduced in Lemma~\ref{descent contour 1}.
In this case we require $c$ to satisfy $c > (3+ \rho )/2$ such that the condition $| (w+c)/(z+c) | < 1$, which is imposed to describe $I_\nu$ as a Fredholm determinant, holds for any $(z, w) \in \Gamma \times \Sigma$.

A straightforward approach to study the long time behaviour of a Fredholm determinant, is to evaluate the asymptotics of $\phi_k^c(x)$ and $\psi_k^c(x)$ separately, then translate the sum over an integer $k$ into an integral over a real number $\kappa$, via $k=m-\lambda_1t^{1/3}\kappa$, where $\lambda_1$ is some positive constants. The uniform convergence of the Fredholm determinant requires that $\phi_k^c(x)\psi_k^c(y)$ decays exponentially with respect to $x,y$ and $\kappa$. The term involving $\kappa$ is $\left(\frac{w(1-w_c)}{w_c(1-w)}\right)^\kappa$, where $w_c=\rho'/2$ is added manually due to saddle point analysis\footnote{This will be shown in Lemma~\ref{descent contour 1}}. However, due to the singularity of $w=-1$ in $\psi_k^c(x)$, the term $\frac{w(1-w_c)}{w_c(1-w)}$ can not be bounded below 1. Because if we parameterise the contour around $w=-1$ by $w=-1+r\e^{\ii \theta }$ for some constant $0<r<1$, then the upper bound of  $\left|\frac{w(1-w_c)}{w_c(1-w)}\right|$ is given by $\frac{(1+r)(1+\rho)}{(2-r)(1-\rho)}$, which is always greater than 1 for any  $\rho \in ((1-2r)/3,1)$. Hence we can not find a bound that is exponentially decaying with respect to $\kappa$. To avoid the divergence due to $w=-1$ in $\kappa$, we perform the sum in $k$ before asymptotic analysis. Note that, for the case of $I_z$ defined in \eqref{def:Iz}, we can use a straightforward way to show that it converges to the Gaussian distribution. In other words, as will be seen in sub-Subsection~\ref{se:Iz}, we will evaluate the asymptotic behaviours of $\phi_k(x)$ and $\psi_k(x)$ given by \eqref{phiandpsi2} separately, and then perform the sum in $k$.

By performing the sum $\sum_{k=0}^{m-1}$ in (\ref{eq:fredholm1}), we obtain a double contour integral expression of $K^c(x,y)$ as
\begin{multline}\label{KernelI1withc}
	K^c(x, y)  = \oint_1 \frac{\mathrm{d} z}{2 \pi \ii} \frac{1}{{\left( z + c \right)}^x} {\left( \frac{1 + z}{z} \right)}^n {\left( \frac{z}{1 - z} \right)}^m  \mathrm{e}^{ - z t/2 } \times \\\oint_{0 , -1} \frac{\mathrm{d} w}{2 \pi \ii} {\left( w + c \right)}^{y-1}  {\left( \frac{w}{1 + w} \right)}^n {\left( \frac{1 - w}{w} \right)}^m \mathrm{e}^{ w t /2 } \frac{1}{w - z}.
\end{multline}

In the following, we will show a rigorous asymptotic analysis on the first term $\mathcal{I}_1$ \eqref{eq:fredholm1}, by computing the scaling limit of the kernel $K^c(x, y)$ \eqref{KernelI1withc}.  Recall that the scaling of $n$ and $m$ we consider has been set in \eqref{nmscale}. Accordingly, we rescale $(x,y)$ as
\begin{equation}\label{xyscale}
	x=\lambda_ct^{1/3}\xi, \qquad y=\lambda_ct^{1/3}\zeta,
\end{equation}
where 
\begin{equation}\label{DefLambdac}
	\lambda_c = (1 - \rho + 2c ) {\left( \frac{3}{4 (1 + \rho) (3 - \rho)} \right)}^{1/3}.
\end{equation}

The scaling power $t^{1/3}$ in \eqref{xyscale} is chosen according to the scaling \eqref{nmscale} arising from the prediction of nonlinear fluctuating hydrodynamics (NLFHD), and the constant $\lambda_c$ in \eqref{DefLambdac} is chosen for convenience to ease notation in Proposition \ref{uniform convergence kernelc}. 
The basic strategies can be taken from the previous works \cite{borodin2008transition,BFPS2007} and are given as follows:

\begin{itemize}
	\item First in Section \ref{Sec:SteepestDesecnt}, we define a steepest descent contour of $K^c$ with the scaling \eqref{nmscale} and \eqref{xyscale}.
	\item Using the steepest descent, we prove the uniform convergence of $K^c(x, y)$ to the Airy kernel \eqref{def:airy_kernel} for bounded $\xi,\zeta$ in Section \ref{Sec:KernelUniform}.
	\item Then we evaluate the bounds of $K^c(x, y)$ for large $\xi, \zeta$ in Section \ref{Sec:KernelEstimate} .
	\item Finally by the convergence and bounds of the kernel $K^c(x, y)$, we are able to obtain the long time behaviour of $\mathcal{I}_1$ with \eqref{nmscale} given in \eqref{I1 limit}.
\end{itemize}
However, we can not simply follow the arguments in the previous works. A novelty in our treatment of Fredholm determinants is that  
we introduce an additional parameter $c$ to avoid divergence of the Fredholm determinant stemming from the singularity at $-1$ and we evaluate the whole kernel $K^c$ instead of $\phi^c_k$ and $\psi^c_k$ to show that it is bounded by an exponential function.

Before continuing to a rigorous analysis, let us first introduce the following notations for convenience. We rescale kernel by 
\begin{equation}\label{DefRescKc}
	\bar{K}_t^{c} (\xi , \zeta) :=  {(w_c + c)}^{\lambda_c t^{1/3} (\xi - \zeta)} \lambda_c t^{1/3} K^c(\lambda_c t^{1/3} \xi , \lambda_c t^{1/3} \zeta) ,
\end{equation}
where $K^c(x,y)$ is given in \eqref{KernelI1withc}, and $w_c:=(1-\rho)/2$ is the saddle point (which we will see in Lemma \ref{descent contour 1}). Then with the scaling  \eqref{nmscale}, the rescaled kernel \eqref{DefRescKc} can be rewritten by collecting the terms according to the order of $t$:
\begin{equation}\label{RescKernelc}
	\bar{K}_t^{c}(\xi , \zeta)  = \lambda_c t^{1/3} \oint_{1} \frac{\mathrm{d} z}{2 \pi \ii} \mathrm{e}^{ f(z , t , \xi) - f(w_c , t , \xi) + g_{\phi}(z) } \oint_{0 , -1} \frac{\mathrm{d} w}{2 \pi \ii} \mathrm{e}^{ - f(w , t , \zeta) + f(w_c , t , \zeta) + g_{\psi}(w) }   \frac{1}{w-z},
\end{equation}
where $f(z , t , \xi)=g_1(z)t+g_2(z)t^{1/2}+g_3(z,\xi)t^{1/3}$ with the functions $g_1(z)$, $g_2(z)$, $g_3(z,\xi)$, $g_{\phi}(z)$ and $g_{\psi}(z)$ defined by
\begin{subequations}\label{gDefine}
	\begin{align}
		g_1(z)=&
		\frac{\rho(3-\rho)^2}{16}
		\ln\left(\frac{1+z}{z}\right)+
		\frac{(1+\rho)^2(2-\rho)}{16}
		\ln\left(\frac{z}{1-z}\right)-\frac{z}{2},
		\label{g1 def1}
		\\
		g_2(z)=&
		-\frac{(3-\rho)c_{\rm g} s_{\rm g}}{12(1-\rho)}
		\ln\left(\frac{1+z}{z}\right)
		-\frac{(1+\rho)c_{\rm g} s_{\rm g}}{12(1-\rho)}
		\ln\left(\frac{z}{1-z}\right),
		\label{g2 def1}
		\\
		g_3(z,\xi)=&
		\frac{-\rho c_2 s_2}{6(1-\rho)}
		\ln\left(\frac{1+z}{z}\right)
		-
		\frac{(2-\rho)c_2 s_2}{6(1-\rho)}
		\ln\left(\frac{z}{1-z}\right)
		-\xi\lambda_c\ln(z+c),
		\label{g3 def1} \\
		g_{\psi}(z)=&-\ln(z+c),\,\,\,\,
		g_{\phi}(z)=0.
	\end{align}
	\label{g def1}
\end{subequations}

In vicinity of a saddle point $w_c$, the functions $g_1(z)$, $g_2(z)$, $g_3(z, \xi)$ and $g(z)$ can be represented as the following Taylor expansions which are useful in the later analysis: 
\begin{subequations}\label{gExpands}
	\begin{align}
		g_1(z) &= g_1(w_c) + 2 a_1 {(z - w_c)}^3 + {(z - w_c)}^4 h_1(z), \label{g1Expand} \\
		g_2(z) & = g_2(w_c) + b_2 {(z - w_c)}^2 + {(z - w_c)}^3 h_2(z), \\
		g_3(z , \xi) & = g_3(w_c , \xi) + b_{3,\xi} {(z - w_c)} + {(z - w_c)}^2 h_3(z, \xi), \\
		g_{\psi}(z) & = g_{\psi}(w_c) + b_{\psi} (z - w_c) + {(z - w_c)}^2 h_{\psi}(z),\\
		g_{\phi}(z) & = g_{\phi}(w_c) + b_{\phi} (z - w_c) + {(z - w_c)}^2 h_{\phi}(z),
	\end{align}
\end{subequations}
where $h_1,h_2,h_3,h_{\psi}$ and $h_{\phi}$ are some functions, and constants $a_1$, $b_2$, $b_{3,\xi}$, $b_{\psi}$ and $b_{\phi}$ are given by

\begin{subequations}\label{gCoefs}
	\begin{align}
		a_1 &= \frac{1}{(3 - \rho)(1 + \rho)}, \label{Defa1} \\
		b_2 &= - \frac{4 c_{\rm g} s_{\rm g}}{3 {(1 - \rho)}^2 (1 + \rho) (3 - \rho) }, \\
		b_{3,\xi} &= - \frac{4 c_2 s_2}{(1 - \rho)(1 + \rho)(3 - \rho)} - \frac{\xi\lambda_c}{w_c+c}, \\
		b_{\psi} &= - \frac{1 }{w_c + c},
		\,\,\,\,
		b_{\phi}=0.
	\end{align}
\end{subequations}
Note that $a_1$ is positive.

\subsection{Steepest descent}\label{Sec:SteepestDesecnt}

To obtain the scaling limit of $\bar{K}^c_t$ \eqref{RescKernelc}, we first need to find the critical point as well as the corresponding steepest descent contour. Since $t\rightarrow\infty$, we only consider the term in the integrand of \eqref{RescKernelc} that is of the largest order of $t$:
\[\bar{K}_t^c(\xi,\zeta)=\lambda_c t^{1/3}
\oint_1 \frac{\mathrm{d} z}{2 \pi \ii}
\oint_{0,-1} \frac{\mathrm{d} w}{2 \pi \ii}
\exp(g_1(z)t - g_1(w)t+\order{t^{1/2}}),\]
where $g_1(z)$ is given in \eqref{g1 def1}. Regarding to $g_1(w)$, we define the following descent contour, along which $\R(g_1(w))$ has a global maximum.

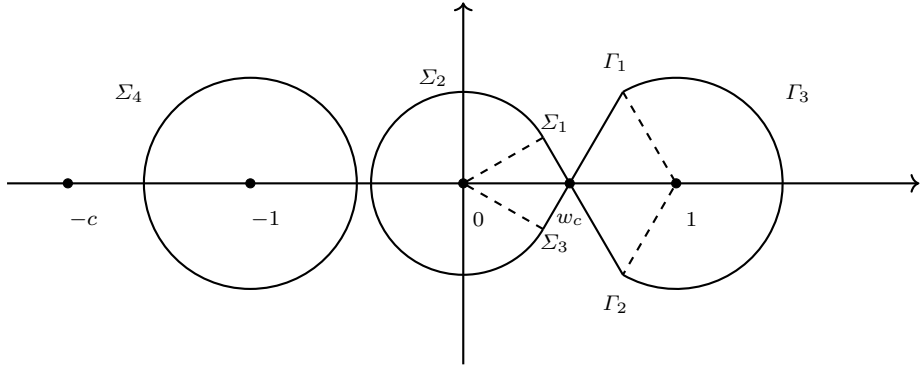
\begin{figure}[h!]
\begin{center}
\begin{tikzpicture}[scale=4]

\draw[->,thick] (-1.5,0) -- (1.5,0);
\draw[->,thick] (0,-0.6) -- (0,0.6);

\draw[thick] (-0.303109,0) arc (180:30:0.303109);
\draw[thick] (-0.303109,0) arc (180:330:0.303109);
\draw[thick] (-0.7,0) circle (0.35);
\draw[fill=black] (-0.7,0) circle (0.015);
\node at (-0.65,-0.12) {$-1$};
\node at (-1.25,-0.12) {$-c$};

\draw[fill=black] (0,0) circle (0.015);
\draw[fill=black] (0.7,0) circle (0.015);
\draw[fill=black] (0.35,0) circle (0.015);
\draw[fill=black] (-1.3,0) circle (0.015);
\node at (0.05,-0.12) {$0$};
\node at (0.75,-0.12) {$1$};
\node at (0.35,-0.12) {$w_c$};
\node at (0.5,0.4) {$\Gamma_1$};
\node at (0.5,-0.4) {$\Gamma_2$};
\node at (1.1,0.3) {$\Gamma_3$};
\node at (0.3,0.2) {$\Sigma_1$};
\node at (-0.1,0.35) {$\Sigma_2$};
\node at (0.3,-0.2) {$\Sigma_3$};
\node at (-1.1,0.3) {$\Sigma_4$};

\draw[thick] (0.525,-0.303109) arc (-120:120:0.35);			

\coordinate (A) at (0.525,0.303109);
\coordinate (B) at (0.7,0);
\coordinate (C) at (0.525,-0.30310);
\coordinate (D) at (0.35,0);
\coordinate (E) at (0.4,0.0866025);
\coordinate (F) at (0.4,-0.0866025);
\draw[thick,dashed] (A) -- (B) -- (C);

\coordinate (G) at (0.2625,0.151554);
\coordinate (H) at (0.2625,-0.151554);
\coordinate (O) at (0,0);

\draw[thick,dashed] (G) -- (O) -- (H);

\draw[thick] (A) -- (H) ;
\draw[thick] (G) -- (C) ;

\end{tikzpicture}
\end{center}
\caption{Steepest descent contour of $\bar{K}^c_t$}
\label{fig:contour tw 1}
\end{figure}

\medskip

\begin{lemman}
	\label{descent contour 1}
	{\rm \textbf{(Steepest  descent contour of $\boldsymbol{\bar{K}_t^c}$)}}
	For $0<\rho<1$ set 
	\[g_1(w)=\frac{\rho(3-\rho)^2}{16}
	\ln\left(\frac{1+w}{w}\right)+
	\frac{(1+\rho)^2(2-\rho)}{16}
	\ln\left(\frac{w}{1-w}\right)-\frac{w}{2}.\]
	Then $g_1'(w)=0$ has a double root at $w_c=(1-\rho)/2=\rho'/2$ and a single root $w_2=\rho-1=-\rho'$. The path $\Gamma=\bigcup_{i=1}^3 \Gamma_i$ (see Fig. \ref{fig:contour tw 1}) given by \eqref{DefGamma} below is a steepest descent path of $g_1(w)$ passing through $w_c$. Namely, $w=w_c$ is the strict global maximum point of ${\rm Re}(g_1)$ along $\Gamma$, i.e., ${\rm Re}(g_1(w))<{\rm Re}(g_1(w_c))$ except when $w=w_c$. Moreover, ${\rm Re}(g_1)$ is monotone along $\Gamma$ except two points where it reaches its maximum and minimum. Meanwhile, the path $\Sigma=\bigcup_{i=1}^4 \Sigma_i$ (see Fig. \ref{fig:contour tw 1}) given by \eqref{DefSigma} below is a steepest descent path of $-g_1(w)$ passing through $w_c$.

	\begin{subequations}\label{DefGamma}
		\begin{align}
			\Gamma_1 & = \left\{ w \in \mathbb{C} \middle| w = \frac{1 - \rho}{2} - s \mathrm{e}^{\ii \pi / 3} , s \in \left[ - \frac{1 + \rho}{2} , 0 \right] \right\} \label{DefGamma1}, \\
			\Gamma_2 & = \left\{ w \in \mathbb{C} \middle| w = \frac{1 - \rho}{2} + s \mathrm{e}^{-\ii \pi / 3} , s \in \left[ 0, \frac{1 + \rho}{2} \right] \right\} \label{DefGamma2}, \\
			\Gamma_3 & = \left\{ w \in \mathbb{C} \middle| w = 1 + \frac{1 + \rho}{2} \mathrm{e}^{\ii \theta} , \theta \in \left[ - \frac{2 \pi }{3} , \frac{2 \pi }{3} \right] \right\} \label{DefGamma3}, 
		\end{align}
	\end{subequations}
	
	\begin{subequations}\label{DefSigma}
		\begin{align}
			\Sigma_1 & = \left\{ w \in \mathbb{C} \middle| w = \frac{1-\rho}{2} + s \mathrm{e}^{2 \pi \ii /3} , s \in \left[ 0, \frac{(1 - \rho)}{4} \right] \right\} \label{DefSigma1}, \\
			\Sigma_2 & = \left\{ w \in \mathbb{C} \middle| w = \frac{\sqrt{3}}{4}(1 - \rho) \mathrm{e}^{\ii \theta} , \theta \in \left[ \frac{\pi}{6} , \frac{11 \pi}{6} \right] \right\} \label{DefSigma2}, \\
			\Sigma_3 & = \left\{ w \in \mathbb{C} \middle| w = \frac{1-\rho}{2} - s \mathrm{e}^{- 2 \pi \ii /3} , s \in \left[ -\frac{(1 - \rho)}{4} , 0 \right]  \right\} \label{DefSigma3}, \\
			\Sigma_4 & = \left\{ w \in \mathbb{C} \middle| w = -1 + \frac{1 + \rho}{2} \mathrm{e}^{\ii \theta} , \theta \in \left[ 0 , 2 \pi \right] \right\} \label{DefSigma4}.
		\end{align}
	\end{subequations}
	
\end{lemman}
The proof is given in the Appendix \ref{appxs:SteepestContour1}.

\subsection{$\bar{K}^c_t$ on a bounded set}\label{Sec:KernelUniform}

With the steepest descent contour given above, we arrive at the uniform convergence of $\bar{K}^c_t$ for bounded $\xi,\zeta$. We first show the contribution from the contour away from the saddle point $w_c=\rho'/2$ vanishes as $t\rightarrow\infty$, then the convergence is obtained by a Taylor expansion near $w_c=\rho'/2$ and a change a variable.

\medskip

\begin{prop}
\label{uniform convergence kernelc}
{\rm \textbf{(Uniform convergence of $\boldsymbol{\bar{K}_t^c}$ on a bounded set)}}
Let $m,n$ be scaled as \eqref{nmscale}. Then for any fixed $L>0$, the rescaled kernel $\bar{K}_t^{c} (\xi , \zeta)$ defined in \eqref{RescKernelc} converges uniformly on $\xi,\zeta \in [-L,L]$ to
\begin{equation}
\label{phi1 limit}
\lim_{t \rightarrow \infty}
\bar{K}_t^{c} (\xi , \zeta)
= A(\xi + s_2 , \zeta + s_2) , 
\end{equation}
where $A(x,y)$ is the Airy kernel defined in \eqref{def:airy_kernel} and $s_2$ is given in \eqref{s+-}. 
\end{prop}

\begin{proof}
We give a sketch of the proof idea and provide a rigorous proof in Appendix~\ref{Unif_conv_Kc}.
Owing to \eqref{gExpands} and \eqref{gCoefs}, the function $f(z,t,\xi)$ can be expanded with respect to $z$ in vicinity of a saddle point $w_c$ as

\begin{equation*}
f(z ,t , \xi ) = f(w_c,t, \xi) -
b_{3,\xi}(z - w_c)t^{1/3} + (z-w_c)^2\mathcal{O}(t^{1/2}) +
2 a_1 {(z - w_c)}^3 (t+\mathcal{O}(t^{1/2}))  + \ldots ,
\end{equation*}
where $a_1$ and $b_{3,\xi}$ is given in \eqref{gCoefs}. The same is true for $f(w ,t , \zeta )$ .

Let us define the scaled variables $v$ and $u$ along the contours $\Gamma_{1,2}$ and $\Sigma_{1,3}$, respectively, by
\begin{align*}
z - w_c &= \frac{v}{ \lambda t^{1/3}}, & w - w_c & = \frac{u}{\lambda t^{1/3}}.
\end{align*}
We choose $\lambda$ such that $6a_1=\lambda^3$, namely $\lambda=(6a_1)^{1/3}$. Recall the coeffcients given in \eqref{gCoefs}, we obtain 
\begin{equation}\label{DefNormalLambda}
	\lambda = {\left( \frac{6}{(1 + \rho)(3 - \rho)} \right)}^{1/3}.
\end{equation}
Therefore by simple calculation, we have $b_{3, \xi} \lambda^{-1} = - s_2 - \xi \lambda_c (w_c + c)^{-1} \lambda^{-1}$. We now choose $\lambda_c$ such that $\lambda_c (w_c + c)^{-1} \lambda^{-1}=1$, i.e., 
\begin{equation*}
	\lambda_c = (w_c + c) \lambda = (1 - \rho + 2c ) {\left( \frac{3}{4 (1 + \rho) (3 - \rho)} \right)}^{1/3} ,
\end{equation*}
which agrees with \eqref{DefLambdac}.
From the integral form of the Airy function \eqref{def:airy_function} and the identity $1/ a = \int_{0}^{\infty} \dd x \e^{-a x}$ for $a > 0$, we arrive at the final result,

\begin{equation*}
\begin{split}
& \lim_{t \rightarrow \infty} \lambda_c t^{1/3} \oint_{1} \frac{\mathrm{d} z}{2 \pi \ii} \mathrm{e}^{ f(z , t , \xi) - f(w_c , t , \xi) } \times \lambda t^{1/3} \oint_{0 , -1} \frac{\mathrm{d} w}{2 \pi \ii} \mathrm{e}^{ - f(w , t , \zeta) + f(w_c , t , \zeta) + g_{\psi}(w)} \frac{1}{\lambda t^{1/3} (w-z)} \\
= & \lim_{t \rightarrow \infty} \int \frac{\dd v}{2 \pi \ii} \e^{ \frac{1}{3} v^3 - (s_2 + \xi) v + \mathcal{O}(t^{-1/6}) } \int \frac{\dd u}{2 \pi \ii} \e^{ - \frac{1}{3} u^3 + (s_2 + \zeta) u + \mathcal{O}(t^{-1/6}) } \int_{0}^{\infty} \dd \kappa \e^{ - \kappa(v - u) } \\
= & \int_0^{\infty} \Ai(s_2+\xi+\kappa)\Ai(s_2+\zeta+\kappa) \dd \kappa.
\end{split} 
\end{equation*}
Above, we showed only pointwise convergence on a bounded set ${[-L, L]}^2$.
In Appendix~\ref{Unif_conv_Kc}, we prove uniform convergence rigorously. \qed

\end{proof}

\medskip

\subsection{Estimate of kernel}\label{Sec:KernelEstimate}
In the following, we will give an estimate of the rescaled function $\bar{K}^c_t(\xi, \zeta)$ for unbounded $\xi,\zeta$.

\medskip

\begin{prop}
	\label{exp_bound}
	{\rm \textbf{(Estimate of  $\boldsymbol{\bar{K}_t^{c}(\xi , \zeta)}$ for unbounded $\boldsymbol{\xi , \zeta}$ )} }
	Let $m,n$ be scaled as \eqref{nmscale}. Then for any large enough $L$ and $t$, the rescaled kernel $\bar{K}_t^{c}(\xi , \zeta)$ defined in \eqref{RescKernelc} is bounded by
	\begin{equation}
		\left| \bar{K}_t^{c}(\xi , \zeta) \right| \leq \e^{-(\zeta + \xi)},
	\end{equation}
	when $(\xi , \zeta) \in [0, \i )^2 \backslash [0 , L]^2$.	
\end{prop}

\begin{proof}
We first need to deform the contours. In the proof of Proposition \ref{uniform convergence kernelc}, 
which is provided in Appendix~\ref{Unif_conv_Kc}, we deformed only $\Gamma$ to $\Gamma'$, since we only 
needed to bound $|w-z|$ there. Here we need to deform both of $\Gamma, \Sigma$ to $\Gamma', \Sigma'$ to be away from $w_c$ (see Fig. \ref{Fig.CotourWhole1_text}), because now we want to estimate terms which involve both $\zeta$ and $\xi$, i.e. $\left| \frac{w_c+c}{z+c}\right|$ and $\left| \frac{w+c}{w_c+c}\right|$.
The deformed contour $\Gamma'$ is obtained by replacing, the segments which are within $2 w_c \delta$ from the saddle point $w_c$, with a vertical line through $w_c ( 1+ \delta)$. Similarly we have the deformed contour $\Sigma'$. Explicitly they are given by 
\begin{subequations}
\begin{align}
	\Gamma' =& \left\{ z \in \Gamma \middle| | z - w_c | > 2 w_c \delta  \right\} \cup \left\{ z \in \mathbb{C} \middle| z = w_c + w_c \delta (1 - s \ii) , s \in \left[  -\sqrt{3} , \sqrt{3} \right]  \right\} , \label{Gamma_prime}  
	\\
	\Sigma' =& \left\{ z \in \Sigma \middle| | z - w_c | > 2 w_c \delta  \right\} \cup \left\{ z \in \mathbb{C} \middle| z = w_c - w_c \delta (1 - s \ii) , s \in \left[  -\sqrt{3} , \sqrt{3} \right]  \right\} . \label{Sigma_prime} 
\end{align}
\end{subequations}

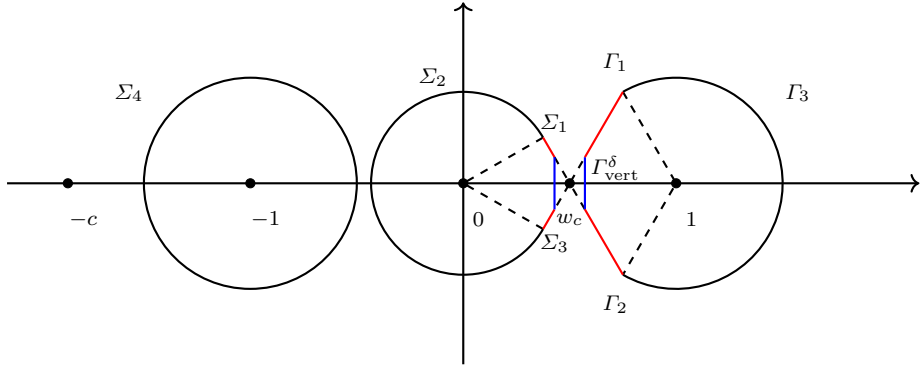
\begin{figure}[h]
\begin{center}
\begin{tikzpicture}[scale=4.0]
\draw[->,thick] (-1.5,0) -- (1.5,0);
\draw[->,thick] (0,-0.6) -- (0,0.6);

\draw[thick] (-0.303109,0) arc (180:30:0.303109);
\draw[thick] (-0.303109,0) arc (180:330:0.303109);
\draw[thick] (-0.7,0) circle (0.35);
\draw[fill=black] (-0.7,0) circle (0.015);
\node at (-0.65,-0.12) {$-1$};
\node at (-1.25,-0.12) {$-c$};

\draw[fill=black] (0,0) circle (0.015);
\draw[fill=black] (0.7,0) circle (0.015);
\draw[fill=black] (0.35,0) circle (0.015);
\draw[fill=black] (-1.3,0) circle (0.015);
\node at (0.05,-0.12) {$0$};
\node at (0.75,-0.12) {$1$};
\node at (0.35,-0.12) {$w_c$};
\node at (0.5,0.4) {$\Gamma_1$};
\node at (0.5,-0.4) {$\Gamma_2$};
\node at (1.1,0.3) {$\Gamma_3$};
\node at (0.3,0.2) {$\Sigma_1$};
\node at (-0.1,0.35) {$\Sigma_2$};
\node at (0.3,-0.2) {$\Sigma_3$};
\node at (-1.1,0.3) {$\Sigma_4$};
\node at (0.5,0.05) {$\Gamma^\delta_\mathrm{vert}$};
			
\draw[thick] (0.525,-0.303109) arc (-120:120:0.35);			

\coordinate (A) at (0.525,0.303109);
\coordinate (B) at (0.7,0);
\coordinate (C) at (0.525,-0.30310);
\coordinate (D) at (0.35,0);
\coordinate (E) at (0.4,0.0866025);
\coordinate (F) at (0.4,-0.0866025);
\draw[thick,dashed] (A) -- (B) -- (C);
\draw[thick,dashed] (F) -- (D) -- (E);
\draw[thick,blue] (F) -- (E);
\draw[thick,red] (C) -- (F);
\draw[thick,red] (E) -- (A);
			
\coordinate (G) at (0.2625,0.151554);
\coordinate (H) at (0.2625,-0.151554);
\coordinate (I) at (0.3,0.0866025);
\coordinate (J) at (0.3,-0.0866025);
\coordinate (O) at (0,0);

\draw[thick,dashed] (G) -- (O) -- (H);
\draw[thick,dashed] (I) -- (D) -- (J);
\draw[thick,red] (G) -- (I);
\draw[thick,red] (J) -- (H);
\draw[thick,blue] (I) -- (J);		
\end{tikzpicture}
\end{center}
\caption{Deformed contours $\Gamma'$ and $\Sigma'$. $\Gamma'$ is obtained by replacing the edge of $\Gamma (=\bigcup_{i=1}^3 \Gamma_i)$ near the point $w_c$ by a vertical line $\Gamma_{\rm vert}$ (the blue line), and $\Sigma'$ is replacing the edge of $\Sigma (=\bigcup_{i=1}^4 \Sigma_i)$ near the point $w_c$ by $\Sigma_{\rm vert}$ (the blue line).}
\label{Fig.CotourWhole1_text}
\end{figure}

Then we will separate the integrand into two parts: those contain $\zeta$ or $\xi$, and those independent of $\zeta$ and $\xi$. We first estimate the second one (independent of $\zeta$ and $\xi$) along the deformed contour in three parts: $\Gamma'=(\Gamma' \backslash \Gamma'^{\Delta})  \cup (\Gamma'^{\Delta} \backslash \Gamma_{\rm vert}) \cup \Gamma_{\rm vert}$, where 
\[
\Gamma_{\rm vert} = \left\{ z \in \mathbb{C} \middle| z = w_c + w_c \delta (1 - s \ii) , s \in \left[  -\sqrt{3} , \sqrt{3} \right]  \right\} ,
\,\,\,\,
\Gamma^{\Delta} = \left\{ z \in \Gamma' \middle| |z-w_c| \leq \Delta  \right\}.
\]
Respectively, we have $\Sigma'=(\Sigma' \backslash \Sigma'^{\Delta})  \cup (\Sigma'^{\Delta} \backslash \Sigma_{\rm vert}) \cup \Sigma_{\rm vert}$. We require that $2w_c\delta < \Delta$, so that $\Gamma_{\rm vert}$ is inside $\Gamma'^{\Delta}$ (and $\Sigma_{\rm vert}$ is inside $\Sigma'^{\Delta}$). The reason we separate the contour $\Gamma'$ into 3 parts is that $g_1(z)$ decays exponentially far away from the saddle point $z=w_c$ (the black arcs in Fig. \ref{Fig.CotourWhole1_text}), while for those near the saddle point, the contour along the vertical line (the blue line in Fig. \ref{Fig.CotourWhole1_text}) and the one along the direction $\e^{\pm \ii \pi / 3}$ (the red part in Fig. \ref{Fig.CotourWhole1_text}) need to be considered individually. Here we choose $\Delta = t^{-1/9}$ for $t$ is large enough. 
Then we bound the factors which contain $\zeta$ or $\xi$ by its maximal value along $\Gamma'$ (and $\Sigma'$). Combining these two results, one can arrive with the bound $\e^{-\zeta-\xi}$ by choosing appropriate value of $\delta$ according to $\zeta$ and $\xi$. 

We let $
\bar{g}_3(z) := g_3(z,\xi) + \xi \lambda_c \ln(z+c)$, i.e., $\bar{g}_3$ is the terms in $g_3$ that do not involve $\zeta$ or $\xi$. Correspondingly, $\bar{f}(z,t) = g_1(z) t +g_2(z) t^{1/2} + \bar{g}_3(z) t^{1/3}$. Then the Taylor expansion in \eqref{gExpands} now becomes $\bar{g}_3(z) = \bar{g}_3(w_c) + \bar{b}_3(z - w_c) + (z - w_c)^3 \bar{h}_3(z)$, where $\bar{b}_3 = b_{3,\xi} +\xi \lambda_c / (w_c +c) = - 4 c_2 s_2 / [(1 - \rho)(1 + \rho)(3 - \rho)]$. Now the rescaled kernel is written as
\begin{equation}\label{resc K suitable for asymptotics}
	\begin{split}
		\bar{K}_t^{c}(\xi , \zeta)  = \lambda_c t^{1/3} \oint_{1} \frac{\mathrm{d} z}{2 \pi \ii} \mathrm{e}^{ \bar{f}(z , t) - \bar{f}(w_c , t) + g_{\phi}(z) } \oint_{0 , -1} \frac{\mathrm{d} w}{2 \pi \ii} \mathrm{e}^{ - \bar{f}(w , t) + \bar{f}(w_c , t) + g_{\psi}(w) } \frac{1}{w-z} \times \\  
		\left(\frac{w_c+c}{z+c}\right)^{\xi \lambda_c t^{1/3}}
		\left(\frac{w+c}{w_c+c}\right)^{\zeta \lambda_c t^{1/3}}.
	\end{split}
\end{equation}
Clearly from Fig. \ref{Fig.CotourWhole1_text},  
\begin{equation}\label{wzBound}
	\left|\frac{1}{w-z}\right| \leq \frac{1}{2 w_c \delta}
\end{equation}
holds along $\Gamma' \times \Sigma'$. In the following, we will give estimate of $\bar{f}(z,t)$ along $\Gamma'$ and $\Sigma'$ first, and then the estimate of $\left| \frac{w_c+c}{z+c}\right|$ and $\left| \frac{w+c}{w_c+c}\right|$.

\paragraph{(i) Estimate of terms independent of  $\xi,\zeta$ }

Recall the Taylor expansion of the $g$'s functions given in \eqref{gExpands}. Assuming that $z$ is not close to the critical point $w_c$, i.e., given a restriction on $\delta$, one can bound the term independent of $\xi,\zeta$ by $c'\e^{c|z-w_c|^3 t}$, where $c$ and $c'$ are some constants to be fixed. We will give the results directly and a detailed calculation can be found in Appendix \ref{Upper_Bound_Kc}. When $\delta > \sqrt{\frac{12|\bar{b_3}|}{a_1 w_c^2}} t^{-1/3} := c_1' t^{-1/3}$ where $a_1$ is given in \eqref{gCoefs}, 

\begin{subequations}\label{fbar bound}
\begin{align}
\label{fbar bound1}
\left| \e^{\bar{f}(z,t) - \bar{f}(w_c, t) + g_{\phi}(z) } \right| &\leq \e^{28 a_1 (w_c \delta)^3t} \e^{g_{\phi}(w_c)}, 
&\textrm{for }& z\in \Gamma_{\rm vert},
\\
\label{fbar bound2}
\left|	\e^{\bar{f}(z,t) - \bar{f}(w_c, t) + g_{\phi}(z) }  \right| &\leq  \e^{-a_1 v^3 t /2} \e^{g_{\phi}(w_c)},
&\textrm{for }& z\in \Gamma'^{\Delta} \backslash \Gamma_{\rm vert},
\\
\label{fbar bound3}
\left| \e^{f(z,t) - f(w_c,t) + g_{\phi}(z)}  \right| &\leq \e^{- a_1 \Delta^3 t /2},
&\textrm{for }& z\in \Gamma' \backslash \Gamma'^{\Delta},
\end{align}
\end{subequations}
where in the second inequality, $z = w_c +v \e^{ \pm \ii \pi /3}$. It follows that 
\begin{equation*}
	\begin{split}
		t^{1/3}\int_{\Gamma'} \left|\frac{\dd z }{2 \pi \ii}\right|	\left| \e^{\bar{f}(z,t)-\bar{f}(w_c,t) + g_{\phi}(z)} \right| \leq t^{1/3} |\Gamma' \backslash \Gamma'^{\Delta}| \e^{- a_1 \Delta^3 t /2} + 2 t^{1/3} e^{g_{\phi}(w_c)} \int_0^{\Delta} \e^{-a_1 v^3 t /2} \dd v + \\ t^{1/3} e^{g_{\phi}(w_c)} \e^{28 a_1 (w_c \delta)^3} \left| \Gamma_{\rm vert}\right|.
	\end{split}
\end{equation*}
We consider the first term on the right hand side of the equation. Choosing $t$ large enough and substituting $\Delta=t^{-1/9}$, we can see  $t^{1/3}\e^{-a_1\Delta^3t/2} $ is bounded by some constant. Then for the second term, we use the change of variable $vt^{1/3} = u$. Thus we have $t^{1/3} \int^{t^{-1/9}}_0 \e^{-a_1 v^3 t/2} \dd v = \int^{t^{2/9}}_0 \e^{-a_1u^3/2} \dd u < \i$. The second term is also bounded by some constant. We assume the sum of the first two terms is bounded by a positive constant $r_3$. Let us consider the last term $t^{1/3} e^{g_{\phi}(w_c)} \e^{28 a_1 (w_c \delta)^3} \left| \Gamma_{\rm vert}\right| = t^{1/3} e^{g_{\phi}(w_c)} \e^{28 a_1 (w_c \delta)^3 t} 2 \sqrt{3} w_c \delta := r_1 t^{1/3}\delta \e^{ r_2 \delta^3 t}$, where $r_1,r_2$ are some positive constants. Specifically,  $r_1=2\sqrt{3} w_ce^{g_{\phi}(w_c)} $, $r_2=28 a_1 w_c^3$. Collecting all the results, we have
\begin{equation*}
	\begin{split}
		t^{1/3} \oint_{\Gamma'} \left|\frac{\dd z }{2 \pi \ii}\right|	\left| \e^{\bar{f}(z,t)-\bar{f}(w_c,t) + g_{\phi}(z)} \right| \leq r_3+r_1 \delta t^{1/3} \e^{r_2 \delta^3 t} \leq
		(r_3+r_1 \delta t^{1/3})\e^{r_2 \delta^3 t}  \leq R(1 + \delta t^{1/3})\e^{r_2 \delta^3 t},
	\end{split}
\end{equation*}
where $R = \max\{r_1,r_3\}$. In order to put the $\delta$ and $t$ into the exponential, we use the inequality $1 + x <\e^{x} $ and $\e^x<\e^{x^3}$ when $x>1$. Namely, if $\delta > t^{-1/3}$, 
\begin{equation*}
	\begin{split}
		t^{1/3} \oint_{\Gamma'} \left|\frac{\dd z }{2 \pi \ii}\right|	\left| \e^{\bar{f}(z,t)-\bar{f}(w_c,t) + g_{\phi}(z)} \right| \leq R \e^{\delta t^{1/3}} \e^{r_2 \delta^3 t} \leq R\e^{r \delta^3 t},
	\end{split}
\end{equation*}
where $r=r_2+1$ is a positive constant.
Please bare in mind that the above inequality holds only when $\delta$ satisfies the restriction: $\delta > c_1 t^{-1/3}$ where $c_1 = \max\{c_1',1\}$. 

Such estimate can be repeated for the deformed contour $\Sigma'$. Combining these two results together, one obtain
\begin{equation}\label{fbarBoundTotal}
	\begin{split}
		t^{2/3} \oint_{\Sigma'} \left|\frac{\dd w }{2 \pi \ii}\right|	\left| \e^{-\bar{f}(w,t)+\bar{f}(w_c,t) + g_{\psi}(w)} \right| \oint_{\Gamma'} \left|\frac{\dd z }{2 \pi \ii}\right|	\left| \e^{-\bar{f}(z,t)+\bar{f}(w_c,t) + g_{\phi}(z)} \right|  \leq R\e^{r \delta^3 t},
	\end{split}
\end{equation} 
where $R$ and $r$ absorb the constants from the deformed contour $\Sigma'$.

\paragraph{(ii) Estimate of terms dependent on $\xi,\zeta$} The estimate of these two terms is straightforward. We simply take the minimum value of $|z+c|$ along $\Gamma'$, and the maximum value of $|w+c|$ along $\Sigma'$. The reason why we introduce the parameter $c$ is that $\max_{w\in \Sigma'}|w| > \max_{w\in \Sigma'}|w_c|$, i.e., $|\frac{w}{w_c}|$ can not be bounded by some number that is smaller than 1. But introducing an appropriate extra parameter $c$, we have $|\frac{w+c}{w_c+c}|<1$ for $w \in \Sigma'$ (see Fig. \ref{Fig.CotourVicinityw_c1_text}).

From Fig. \ref{Fig.CotourVicinityw_c1_text}, one can see that (we refer to Appendix \ref{appx:zw+cBound})
\begin{align}\label{zw+cBound}
\left| \frac{w_c + c}{z + c}\right|  \leq \e^{-\frac{1}{2}\frac{w_c}{w_c + c} \delta}, 
\,\,\,
\left|\frac{w + c}{w_c + c}\right| \leq  \e^{-\frac{1}{2}\frac{w_c}{w_c + c}\delta},
&&
\textrm{for } (z,w) \in \Gamma' \times \Sigma'.
\end{align}

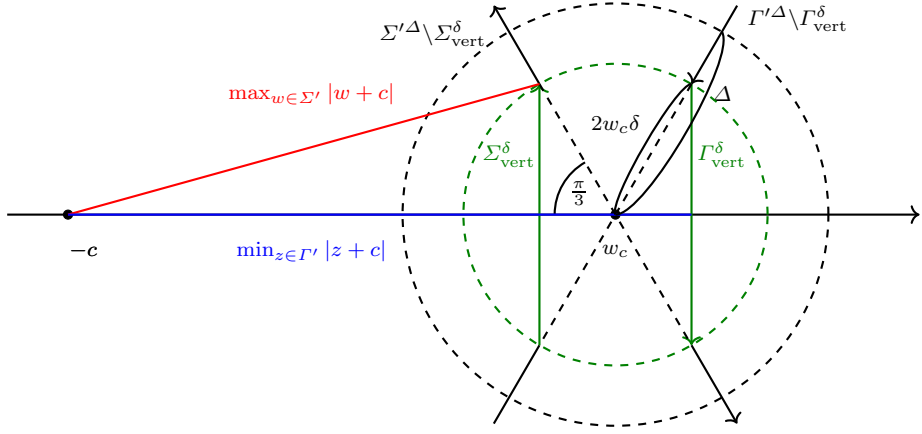
\begin{figure}[h]
\begin{center}
\begin{tikzpicture}[scale=4]
\draw[->,thick] (-1,0) -- (2,0);
			
\node at (-0.75,-0.12) {$-c$};
			
\draw[fill=black] (1,0) circle (0.015);
\draw[fill=black] (-0.8,0) circle (0.015);
			
\node at (1.0,-0.12) {$w_c$};
\node at (-0.75,-0.12) {$-c$};
\node at (0.88,0.07) {$\frac{\pi}{3}$};
\node at (1.35,0.4) {$\Delta$};
\node at (1.,0.3) {$2 w_c \delta$};
\node at (0,-0.12) {$\color{blue} \min_{z \in \Gamma'}{|z + c|} \color{black}$};
\node at (0,0.4) {$\color{red} \max_{w \in \Sigma'}{|w + c|}  \color{black}$};
\node at (0.4,0.6) {$\Sigma'^\Delta \backslash \Sigma_{\rm vert}^\delta$};
\node at (1.6,0.65) {$\Gamma'^\Delta \backslash \Gamma^\delta_{\mathrm{vert}}$};
\node at (0.65,0.2) {$\color{green2} \Sigma_{\rm vert}^\delta \color{black}$};
\node at (1.35,0.2) {$\color{green2} \Gamma^\delta_{\mathrm{vert}} \color{black}$};

\draw[thick,dashed] (1.7,0) arc (0:360:0.7);
\draw[thick,dashed,green2] (1.5,0) arc (0:360:0.5);
\draw[thick] (0.8,0) arc (-180:-240:0.2);
			
\coordinate (A) at (1,0);
\coordinate (B) at (0.75,0.433013);
\coordinate (C) at (0.75,-0.433013);
\coordinate (D) at (1.25,0.433013);
\coordinate (E) at (1.25,-0.433013);
\coordinate (F) at (0.60,0.692820);
\coordinate (G) at (0.60,-0.692820);
\coordinate (H) at (1.40,0.692820);
\coordinate (I) at (1.40,-0.692820);
\coordinate (J) at (-0.80,0);
\coordinate (K) at (1.35,0.606218);
\coordinate (P) at (0.75,0.005);
\coordinate (Q) at (1.25,0);
			
\draw[->,thick] (B) -- (F);
\draw[thick] (G) -- (C);
\draw[->,thick,green2] (C) -- (B);
\draw[thick,dashed] (B) -- (A) -- (C);
\draw[->,thick] (H) -- (D);
\draw[->,thick,green2] (D) -- (E);
\draw[->,thick] (E) -- (I);
\draw[thick,dashed] (D) -- (A) -- (E);
\draw[thick,red] (J) -- (B);
\draw[thick,blue] (J) -- (Q);

\draw[thick] (D) .. controls (1.2,0.433013) and (0.975,0.0433013) .. (A);
\draw[thick] (K) .. controls (1.4,0.519615) and (1.1,0) .. (A);
			
\end{tikzpicture}
\end{center}
\caption{The dotted green and black curves are circles of radius $2 w_c \delta$ and $\Delta$, respectively, centred at $w_c$. The green lines represent the paths $\Gamma^{\delta}_{\mathrm{vert}} $ and $ \Sigma^{\delta}_{\mathrm{vert}}$, while the black lines in the black dotted circle are the paths ${\Gamma'}^\Delta \backslash \Gamma^\delta_{\mathrm{vert}}$ and $\Sigma^\Delta \backslash \Sigma^\delta$. The length of blue line is the minimum values of $|z + c|$ on $\Gamma'$ and that of red line is the maximum value of $|w + c|$ on $\Sigma$.}
\label{Fig.CotourVicinityw_c1_text}
\end{figure}

In conclusion, the integrand depending on $\xi,\zeta$ has the following bound along $\Gamma'\times \Sigma'$. Given $0<\delta<1/4$, we have
\begin{equation}\label{XiZetaTermBound}
	\left| 
	\left(\frac{w_c+c}{z+c}\right)^{\xi \lambda_c t^{1/3}}
	\left(\frac{w+c}{w_c+c}\right)^{\zeta \lambda_c t^{1/3}}
	\right| \leq \e^{-\frac{1}{2}\lambda_c \frac{w_c}{w_c + c}(\xi+\zeta)\delta t^{1/3}}.
\end{equation}

\paragraph{(iii) Total estimate}
In part (i), we have the bound of $\e^{\bar{f}(z,t)-\bar{f}(w_c,t)+g_{\phi}(z)}$ (and  $\e^{-\bar{f}(w,t)+\bar{f}(w_c,t)+g_{\psi}(z)}$) along $\Gamma'$ (and $\Sigma'$), given in \eqref{fbarBoundTotal}. Then the bound of $|(w_c + c)/(z+c)|$ and $|(w + c)/(w_c + c)|$ is analysed in part (ii) and given in \eqref{XiZetaTermBound}. The integrand of the rescaled kernel is left with $1/(w-z)$, which is bounded by $(2 w_c \delta)^{-1}$ given in \eqref{wzBound}.
It follows that when $c_1 t^{-1/3} < \delta <1/4 $, 
\begin{equation}\label{KerscBoundTotal1}
	\begin{split}
		\left| \bar{K}_t^{c}(\xi , \zeta) \right|
		\leq \lambda_c t^{-1/3} R \e^{r \delta^3 t} \frac{1}{2w_c\delta} \e^{-\frac{1}{2}\lambda_c \frac{w_c}{w_c + c}(\xi+\zeta)\delta t^{1/3}} \leq   \e^{r \delta^3 t}  \e^{-s(\xi+\zeta)\delta t^{1/3}},
	\end{split}
\end{equation}
where $s:= \frac{\lambda_c w_c}{2(w_c + c)}$ and we impose another restriction that $\delta > \frac{\lambda_c R}{2 w_c} t^{-1/3}$. The condition on $\delta$ now becomes $c't^{-1/3} < \delta < 1/4$ where $c'$ is given by $c' =\max\{\frac{\lambda_c R}{2 w_c}, c_1\} $. We are only left with showing that the left hand side of \eqref{KerscBoundTotal1} is smaller than $\e^{-\xi-\zeta}$. To achieve this result, we choose a different value of $\delta$ depending on the value of  $\xi+\zeta$ (see Fig.\ref{Fig.RegionModerate} and Fig.\ref{Fig.RegionLarge}).  

\begin{figure}[h]
\begin{tabular}{c}
\begin{minipage}{0.47\hsize}
\begin{center}
\begin{tikzpicture}[scale=2.0]

\draw[->,thick] (-0.2,0) -- (2.5,0);
\draw[->,thick] (0,-0.2) -- (0,2.5);

\draw[fill=black] (1.8,0) circle (0.02);
\draw[fill=black] (0.6,0) circle (0.02);
\draw[fill=black] (0,1.8) circle (0.02);
\draw[fill=black] (0,0.6) circle (0.02);

\node at (-0.1,-0.1) {$O$};
\node at (0.6,-0.1) {$L$};
\node at (-0.1,0.6) {$L$};
\node at (1.5,-0.15) {$(2r/s) t^{1/3}$};
\node at (-0.5,1.8) {$(2r/s) t^{1/3}$};
\node at (0.1,2.5) {$\zeta$};
\node at (2.5,-0.1) {$\xi$};

\coordinate (A) at (-0.2,0.8);
\coordinate (B) at (0.8,-0.2);
\coordinate (C) at (-0.2,2.0);
\coordinate (D) at (2.0,-0.2);
\coordinate (E) at (0,1.8);
\coordinate (F) at (1.8,0);
\coordinate (H) at (0,0.6);
\coordinate (I) at (0.6,0);
\coordinate (J) at (0.6,0.6);

\draw[thick] (A) -- (B);
\draw[thick] (C) -- (D);
\draw[thick,dashed] (H) -- (J) -- (I);

\fill[classicrose] (E) -- (F) -- (I) -- (J) -- (H);

\end{tikzpicture}
\end{center}
\caption{ ${[0 , \infty)}^2 \backslash {[0,  L]}^2$, $L < \xi + \zeta \leq (2r/s) t^{1/3}$ }
\label{Fig.RegionModerate} 

\end{minipage}

\begin{minipage}{0.47\hsize}

\begin{center}
\begin{tikzpicture}[scale=2.0]

\draw[->,thick] (-0.2,0) -- (2.5,0);
\draw[->,thick] (0,-0.2) -- (0,2.5);

\draw[fill=black] (1.8,0) circle (0.02);
\draw[fill=black] (0.6,0) circle (0.02);
\draw[fill=black] (0,1.8) circle (0.02);
\draw[fill=black] (0,0.6) circle (0.02);

\node at (-0.1,-0.1) {$O$};
\node at (0.6,-0.1) {$L$};
\node at (-0.1,0.6) {$L$};
\node at (1.5,-0.15) {$(2r/s) t^{1/3}$};
\node at (-0.5,1.8) {$(2r/s) t^{1/3}$};
\node at (0.1,2.5) {$\zeta$};
\node at (2.5,-0.1) {$\xi$};

\coordinate (A') at (-0.2,2.0);
\coordinate (B') at (2.0,-0.2);
\coordinate (E) at (0,1.8);
\coordinate (F) at (1.8,0);
\coordinate (H) at (0,0.6);
\coordinate (I) at (0.6,0);
\coordinate (J) at (0.6,0.6);
\coordinate (K) at (0,2.4);
\coordinate (L) at (2.4,0);
\coordinate (M) at (2.4,2.4);

\draw[thick] (A') -- (B');
\draw[thick,dashed] (H) -- (J) -- (I);

\fill[blizzardblue] (K) -- (M) -- (L) -- (F) -- (E);

\end{tikzpicture}
\end{center}
\caption{ $ {[0 , \infty)}^2 \backslash {[0 , L]}^2  $, $\xi + \zeta > (2r/s) t^{1/3}$ }
\label{Fig.RegionLarge}
\end{minipage}
\end{tabular}
\end{figure}

$\bullet$ When $\xi + \zeta \leq \frac{2r}{s} t^{1/3}$,  we choose $\delta = \sqrt{\frac{(\xi +\zeta ) s}{2 r}} t^{-1/3}$. Then $ \sqrt{L s/(2r)} t^{-1/3} < \delta < t^{-1/6}$. Then restriction $c' t^{-1/3} < \delta <1/4 $ is satisfied if we choose $L$ large enough such that $L>(2c'r)^2$ and choose $t$ large enough such that $t^{-1/6} < 1/4$. Therefore from \eqref{KerscBoundTotal1}, 
\begin{align*}
	\left| \bar{K}_t^{c}(\xi , \zeta) \right| \leq & 
	\e ^{(\xi +\zeta)^{3/2}(s/2)^{3/2}r^{-1/2}}
	\e^{-(\xi +\zeta)^{3/2}(s)^{3/2}(2r)^{-1/2}}=
	\e^{-(s(\xi + \zeta)/2)^{3/2}r^{-1/2}}\\
	\leq & \e^{-(\xi+\zeta)(s/2)^{3/2}(L/r)^{1/2}}
	\leq \e^{-\xi - \zeta},
\end{align*}
where in the second line, $L < \xi+\zeta$ since $(\xi , \zeta) \in [0, \i )^2 \backslash [0 , L]^2$. The final inequality follows by choosing $L$ large enough such that $L>r(2/s)^3$. 

$\bullet$ When $\xi + \zeta > \frac{2r}{s} t^{1/3}$, i.e., $r<s(\xi +\zeta)/(2t^{1/3})$.  We choose $\delta = t^{-1/6}$. If $t$ is large enough, then the restriction $c' t^{-1/3} < t^{-1/6} < 1/4$ is automatically satisfied. From \eqref{KerscBoundTotal1}, we have
\begin{align*}
	\left| \bar{K}_t^{c}(\xi , \zeta) \right| \leq & 
	\e^{rt^{1/2}}\e^{-s(\xi +\zeta)t^{1/6}}
	\leq 
	\e^{s(\xi +\zeta)t^{1/2}/(2t^{1/3})}\e^{-s(\xi +\zeta)t^{1/6}}=\e^{-(\xi+\zeta)t^{1/6}s/2}\leq \e^{-\xi-\zeta},
\end{align*}
where the last inequality holds when $t$ is large enough such that $t>(2s)^6$. \qed
	
\end{proof}

\subsection{Long time limit of the first term in \eqref{eq:prob2termb}}

We now are ready to conclude the limiting kernel is the \textit{Airy kernel} \cite{TW1994} given in \eqref{def:airy_kernel}.
From Subsections~\ref{Sec:KernelUniform} and \ref{Sec:KernelEstimate}, we obtain the following theorem.

\begin{theorem}
\label{kernel1 limit}
Consider the rescaled kernel defined in \eqref{DefRescKc}
\begin{equation*}
 \bar{K}^c_t( \xi, \zeta)=(w_c + c)^{ \lambda_c (\xi-\zeta) t^{1/3}}
  \lambda_c t^{1/3}K^c(\lambda_c t^{1/3} \xi,\lambda_c t^{1/3} \zeta),
\end{equation*}
with $\lambda_c$ given in \eqref{DefLambdac}. Then we have
\begin{enumerate}[(i)]
\item \label{RescaledFredholm} $\displaystyle \lim_{t \to \infty} \det(1-  \bar{K}^c_t)_{L^2(0, \infty)}= \lim_{t \to \infty} \det(1-\bar{K}^c_t)_{\ell^2(\mathbb{N}/(\lambda_c t^{1/3}))}$.

Here note that the operator $\bar{K}^c_t$ on the right hand side is regarded as acting on $\ell^2(\mathbb{N}/(\lambda_c t^{1/3}))$ while the one on the left hand side on  $L^2(0,\infty)$.
\item \label{UnifConv}
  For any fixed $L>0$
  \begin{equation*}
    \lim_{t\rightarrow \infty}\bar{K}^c_t( \xi, \zeta)= A(\xi + s_2 , \zeta + s_2),
  \end{equation*}
  uniformly on $(\xi,\zeta)\in[-L,L]^2$.
\item \label{UnifBound} For $t$ large enough,
\begin{equation*}
  |\bar{K}^c_t( \xi, \zeta)| \leq C \e^{
  - (\xi + \zeta)},
\end{equation*}
for some constant $C>0$ and $(\xi,\zeta) \in {[0,\infty )}^2$.
\end{enumerate}

\end{theorem}

\begin{proof}
One can show \ref{RescaledFredholm} by rescaling $(x,y)=(\lambda_c t^{1/3} \xi,\lambda_c t^{1/3} \zeta)$ in the Fredholm kernel in \eqref{eq:FHDdef} and replacing the Riemann sums $\displaystyle \lim_{t \to \infty} {(\lambda_c t^{1/3})}^{-1} \sum_{x \in \mathbb{N}}$ with the integral $\int_{0}^{\infty} \dd \xi$.
Statements \ref{UnifConv} and \ref{UnifBound} are nothing but Propositions \ref{uniform convergence kernelc} and \ref{exp_bound}, respectively.
\end{proof}

A general statement on the convergence of a Fredholm determinant is given below.
\begin{lemman}
\label{lem:kernel lim}
Suppose a kernel $K_t$ satisfies
\begin{enumerate}[(i)]
\item
  For any fixed $L>0$
  \begin{equation*}
    \lim_{t\rightarrow \infty}K_t( \xi, \zeta)=H(\xi,\zeta),
  \end{equation*}
  uniformly on $(\xi,\zeta)\in[-L,L]^2$.
\item For any fixed $L>0$ and $t$ large enough,
\begin{equation*}
  |K_t( \xi, \zeta)| \leq C \e^{
  -\max\{0,\xi\}-\max\{0,\zeta\}},
\end{equation*}
for some constant $C>0$ and $\xi,\zeta\geq -L$.
\end{enumerate}
Then we have
\begin{equation*}
\lim_{t \rightarrow \infty} \det(1-K_t)=\det(1-H).
\end{equation*}
\end{lemman}
See Lemma C.2 in \cite{IS2019} for a detailed proof. Consequently, we can conclude that, with the scaling \eqref{nmscale} in Section \ref{sec:Scaling limit}, the kernel $K$ defined in \eqref{eq:fredholm1} satisfies 

\begin{equation}
\label{I1 limit}
\lim_{t\rightarrow\infty}\mathcal{I}_1
= \lim_{t\rightarrow\infty}\det(1-\bar{K}^c_t)_{\ell^2(\mathbb{N}/(\lambda_c t^{1/3}))} = \lim_{t\rightarrow\infty}\det(1-\bar{K}^c_t)_{L^2(0,\infty)} = \det(1-A)_{L^2(s_2 , \infty)}  
=F_2(s_2),
\end{equation}
where $A$ is the Airy kernel defined in \eqref{def:airy_kernel} and $s_2$ is given in \eqref{s+-}.
Obviously, the equality \eqref{I1 limit} itself coincides with the first claim of Theorem~\ref{thm_I1I2 limit}.

\section{Limit of $\mathcal{I}_2$: proof of \eqref{claim_I2 limit} in Theorem~\ref{thm_I1I2 limit} }
\label{sec:Asymptotics_second}

In this section, which consists of three subsections \ref{subsec:Deform_I2}, \ref{subsec:Ev_I21} and \ref{subsec:Ev_I22}, we consider the second term in \eqref{eq:prob2termb}, i.e., $\mathcal{I}_2$, given by \eqref{eq:prob2term2}.
In Subsection~\ref{subsec:Deform_I2}, we will rewrite $\mathcal{I}_2$ of \eqref{eq:prob2term2} to a form suitable for studying asymptotics  with the scaling \eqref{nmscale}.
We will state and prove Proposition~\ref{rank1 perturbation det2} which allows to decouple the $z_j$ dependence of $K$ into a fairly simple rank one perturbation. This is the key to establish the asymptotic decoupling of the two modes in our paper. Using this we can divide $\mathcal{I}_2$ into a main contribution $\mathcal{I}_2^{(1)}$ and a remaining part
$\mathcal{I}_2^{(2)}$ as in (\ref{I2div}).
In Subsections~\ref{subsec:Ev_I21} and \ref{subsec:Ev_I22}, we will show that the main term converges to the GUE Tracy-Widom distribution whilst the remainder tends to zero.

\subsection{Rewriting of $\mathcal{I}_2$} 
\label{subsec:Deform_I2}

As we just saw in the analysis of $\mathcal{I}_1$ in Section~\ref{sec: I1 limit}, a Fredholm determinant formula which appears as a result of procedures in
Section \ref{se:tofredholm} is often suitable for studying asymptotics with the scaling like \eqref{nmscale}.
For $\mathcal{I}_2$, this approach can not be applied at least directly,
due to the presence of a factor involving both $w_j$ and $z_k$ variables. 
But as we will see below, $I_w(\vec{z})$ can be transformed into a Fredholm determinant following the procedures in Section~\ref{se:tofredholm}, 
and we can show that the long time limit of $I_w(\vec{z})$ is dominated by the 
same multiple integral with $\vec{z}=\vec{1}$.
It follows then that $\mathcal{I}_2$ can be asymptotically the same as the product of 
$I_z$ defined in \eqref{def:Iz} below and $I_w(\vec{1})$, and asymptotics
of each multiple integral can be valuated independently by following similar 
arguments as in the previous section. 

Let us start from the multi-fold integral formula of $I_w(\vec{z})$, namely \eqref{Iw}.
In order to rewrite $I_w(\vec{z})$ to a Fredholm determinant, we first show that the right hand side of \eqref{Iw} fits the standard form \eqref{Inu}.
Changing the variables $w_k\rightarrow1/w_k$ in $I_w(\vec{z})$, the $\vec{w}$-contours around the origin are deformed to contours around the infinity.
Hence, we can replace these contours to surround the poles other than the ones at the infinity, namely  $w_k=0,1,-\rho',-z_j^{-1}$ for all $j\in[1,n-1]$, resulting in

\begin{align}
I_w (\vec{z}) = & \frac{{(-1)}^m}{m!}
		\oint_{0}
		\dd^m w
		\frac{\e^{\Lambda_{0,m}t} \Delta_{m}(w) \Delta_m(-w) S_{n-1,m}(z,1) }{\prod_{k=1}^{m}{ (w_k-1)^m} S_{n-1,m}(z,w) }
		\prod_{k=1}^{m}\frac{1+1/\rho'}{w_k+1/\rho'} 
		\nonumber \\ %\label{Iw} \nonumber \\
		= & \frac{1}{m!}\oint_{0,1,-\rho', {\{ -z_j^{-1} \}}_{j=1}^{n-1}} \dd^m w
		\frac{\e^{\tilde{\Lambda}_{0,m}t } \prod_{1 \leq k \neq \ell \leq m}{(w_\ell - w_k)} \prod_{j=1}^{n-1}{ {(1 + z_j)}^m } }{ \prod_{ k=1 }^{m}{ w_k^{m-n-1} {(w_k - 1)}^m \prod_{j=1}^{n-1}{ \prod_{k=1}^{m}{ (1 + z_j w_k) } } }  } \prod_{k=1}^{m}{ \frac{1+1/\rho'}{1+w_k /\rho'} }, \nonumber
\end{align}
where $\tilde{\Lambda}_{n,m}=\frac12\sum_{j=1}^n(z_j-1)+\frac12\sum_{k=1}^m(w_k-1)$.
This expression can be written as the standard $m$-fold integral \eqref{Inu} in Proposition \ref{prop:fredholm},

\begin{equation*}
I_w(\vec{z}) = \frac{1}{m!}\oint_{0,1,-\rho', {\{ -z_j^{-1} \}}_{j=1}^{n-1}} \frac{\dd^m w}{ \prod_{k=1}^m w_k} \frac{\prod_{1\leq k\neq \ell\leq m}(1-w_k/w_\ell)}{\prod_{k,\ell=1}^m (1-a_\ell/w_k)} \prod_{k=1}^m \frac{g(w_k,0 )}{g(a_k,0)},
\end{equation*}
with
\[
\begin{array}{rll}
\nu &=m, \quad &\mu=n, \quad \gamma = t/2,\quad s=0, \\
u_i &=0, & 1\le i \le m, \\
a_i &=1, \quad & 1\leq i\leq m, \\
v_k &=z_j, & 1\le k\le n-1,\quad v_n =1/\rho',
\end{array}
\]
where $g(w_k,x) = g^c(w_k,x;0)$ and $g^c(w,x;y)$ is given in \eqref{gmndef}, namely

\begin{equation*}
g^c(w, x; y)
= \prod_{j=1}^{n-1}\frac{w}{1+z_jw} \frac{w}{1 + w/ \rho'} w^{-m-x} {(w + c)}^y \e^{wt/2}.
\end{equation*}

\noindent As is clear from \eqref{eq:prob2term2} and \eqref{I2 Iz antisym}, the contours of $\vec{z}$-integrals include only the poles at unity, and hence it is allowed to choose them as ones which lie in the vicinity of unity. 
When we choose such contours, $-1/z_j$'s are close to $-1$ and then, likewise the case of $\mathcal{I}_1$ treated in Section~\ref{sec: I1 limit}, there exists $c \in \mathbb{C}$ for which $S(-c,|1+c|)$ includes $0,-\rho',-z_j^{-1}$ for all $j \in [1,n-1]$ and excludes unity.
Therefore, the conditions stated in the claim of Proposition~\ref{prop:fredholm} are satisfied.
Thus applying Proposition \ref{prop:fredholm}, the integral $I_w(\vec{z})$ defined in \eqref{Iw} can be written as a Fredholm determinant:

\begin{equation}
I_w(\vec{z}) = \det(1-K^c(\vec{z}))_{\ell^2(\mathbb{N})} \label{eq:fredholmdet2z}
\end{equation}
with kernel

\begin{equation}
K^c(x,y,\vec{z})=\sum_{k=0}^{m-1} \phi^c_k(x,\vec{z}) \psi^c_k(y,\vec{z}),
\label{eq:fredholm2z}
\end{equation}
and
\begin{align}
\phi^c_k(x,\vec{z}) &= \oint_1 \frac{\dd z}{2\pi\ii}
\frac{1+z/\rho'}{ {(z + c)}^x z (z-1)}
\prod_{j=1}^{n-1}\frac{1+z_jz}{z}
\left(\frac{z}{1-z}\right)^k \e^{-zt/2} ,
\label{phiform2a}
\\
\psi^c_k(x,\vec{z}) &=
\oint_{0,-\rho',\{-z_j^{-1}\}_{j=1}^{n-1}}
\frac{\dd w}{2\pi\ii}
\frac{ {(w + c)}^x }{w(1+w/\rho')}
\prod_{j=1}^{n-1}\frac{w}{1+z_jw}
\left(\frac{1-w}{w}\right)^{k} \e^{ wt/2}.
\label{psiform2a}
\end{align}
In the right hand side of \eqref{eq:fredholm2z}, the summation over $k \in [0,m-1]$ can be performed easily as

\begin{equation*}
    \sum_{k=0}^{m-1}{ {\left[ \frac{z(1-w)}{w(1-z)} \right] }^k} = \frac{w(1-z)}{w - z} \left[ 1 - {\left( \frac{z(1-w)}{w(1-z)} \right)}^m \right] .
\end{equation*}
The first term does not contribute to the kernel because the pole at $z=1$ of the integral with respect to $z$ is removed.
Hence, by performing the sum $\sum_{k=0}^{m-1}$ in \eqref{eq:fredholm2z}, we obtain 

\begin{equation*}
\label{Kxyz}
K^c(x,y,\vec{z}) =   \oint_{1} \frac{\dd z}{2 \pi \ii} F^c(z,x) \prod_{j=1}^{n-1}{ \frac{1 + z_j z}{1 + z} } \oint_{0 , - \rho' ,  \{ - z_j^{-1} \}_{j=1}^{n-1} } \frac{\dd w}{2 \pi \ii} G^c(w,y) \prod_{j=1}^{n-1}{\frac{1 + w}{1 + z_j w}} \frac{1}{w-z},
\end{equation*}
where the functions $F^c(z,x)$ and $G^c(w,y)$ are defined as
\be
\label{def:FcGc}
F^c(z , x) = \frac{z + \rho'}{z + 1} \e^{f^c(z , x)} , \qquad  G^c(w , y) = \frac{w + 1}{(w + \rho')(w + c)} \e^{ - f^c(w , y)}
\ee
with $f^c(z,x)$ given by
\be
\label{def:fc}
f^c(z,x) = n \ln{ \left( \frac{1 + z}{z} \right) } + m \ln{\left( \frac{z}{1 - z} \right)} - x \ln{(z + c)} - \frac{zt}{2}.
\ee

Let us substitute the Fredholm determinant formula of $I_w(\vec{z})$ \eqref{eq:fredholmdet2z} into the anti-symmetric expression of $\mathcal{I}_2$ given by \eqref{I2 Iz antisym}. Recalling the definition of the Fredholm determinant \eqref{eq:FHDdef}, $\mathcal{I}_2$ now reads

\begin{equation}
\label{detdivide_1_pre1}
    \mathcal{I}_2 = \oint_{1} \dd^{n-1} z\, L(\vec{z}) \left[1 + \sum_{k=1}^{m}{ \frac{{(-1)}^k}{ k! } } \sum_{x_1 = 1}^{\infty}{ \sum_{x_2 = 1}^{\infty}{  \cdots \sum_{x_k = 1}^{\infty}{ \det\left[ K^c(x_i , x_j , \vec{z}) \right]_{1 \leq i , j \leq k}  } } } \right] .
\end{equation}
Here $L(\vec{z})$ is defined in \eqref{def:L(z)}. 
Next we show that the summations over $x_i\in\mathbb{N}$ for any $i\in[1,m]$ and
$z$-integrations commute, i.e., we prove the following Lemma.

\medskip

\begin{lemman}
\label{commutativity sum and int}
For any $\rho \in (0,1)$, $t>0$ and $n, m \in \mathbb{N}$, the following equality holds:

\begin{equation*}\label{eq:com_integrand_sum}
\begin{split}
&  \oint_{1} \dd^{n-1} z L(\vec{z}) \sum_{k=1}^{m}{ \frac{{(-1)}^k}{ k! } } \sum_{x_1 = 1}^{\infty}{ \sum_{x_2 = 1}^{\infty}{  \cdots \sum_{x_k = 1}^{\infty}{ \det\left[ K^c(x_i , x_j , \vec{z}) \right]_{1 \leq i , j \leq k}  } } } \\
= & \sum_{k=1}^{m}{ \frac{{(-1)}^k}{ k! } } \sum_{x_1 = 1}^{\infty}{ \sum_{x_2 = 1}^{\infty}{  \cdots \sum_{x_k = 1}^{\infty}{   \oint_{1} \dd^{n-1} z L(\vec{z}) \det\left[ K^c(x_i , x_j , \vec{z}) \right]_{1 \leq i , j \leq k}  } } }. 
\end{split}
\end{equation*}

\end{lemman}
\begin{proof}
See Appendix~\ref{ap:commutativity}.
\qed
\end{proof}

Owing to this Lemma, we can rewrite the right hand side of \eqref{detdivide_1_pre1} as

\begin{equation}
\label{detdivide_1_pre2}
\mathcal{I}_2 =  I_z + \sum_{k=1}^{m}{ \frac{{(-1)}^k}{k!} } \sum_{x_1 = 1}^{\infty}{ \sum_{x_2 = 1}^{\infty}{ \cdots \sum_{x_k = 1}^{\infty}{ \oint_1 \dd^{n-1} z\,  L(\vec{z}) \det\left[ K^c(x_i, x_j , \vec{z}) \right]_{1 \leq i , j \leq k} } } },
\end{equation}
where
\begin{equation}
\label{def:Iz}
I_z := \oint_1 \dd^{n-1} z\,  L(\vec{z}).
\end{equation}

In the following, we will prove that the asymptotic behaviour of $\mathcal{I}_2$ given by the right hand side of \eqref{detdivide_1_pre2} is same as that with $K^c(x_i,x_j,\vec{z})$ replaced by $K^c(x_i,x_j,\vec{1})$. To establish this, the following Proposition for each term in \eqref{detdivide_1_pre2} plays a central role. 

\medskip

\begin{prop}
\label{rank1 perturbation det2}
For any $(x_1 , x_2 , \dots , x_k) \in \mathbb{N}^k$, $\rho \in (0,1)$, $t> 0$ and $n, m \in \mathbb{N}$, the following equality holds:

\begin{equation*}\label{eq:rank1lem5}
\begin{split}
&  \oint_{1} \dd^{n-1} z\, L(\vec{z}) \det\left[ K^c(x_i ,x_j , \vec{z})\right]_{1 \leq i , j \leq k} \\
=& \oint_{1} \dd^{n-1} z\, L(\vec{z}) \det\left\{ K^c_{\mathrm{W}}(x_i , x_j) - \left[ \sum_{l=1}^{n-1}{ \prod_{k=1}^{l}{ {(z_k - 1)} A^c_l(x_i) }} \right] B^c(x_j)   \right\}_{ 1 \leq i , j \leq k } .
\end{split}
\end{equation*}
Here the functions, $K^c_{\mathrm{W}}(x,y)$, $A^c_j(x)$ and $B^c(x)$ are defined as 
\begin{align}
K^c_{\mathrm{W}}(x,y) := & K^c(x ,y, \vec{1})= \oint_{1} \frac{\dd z}{2 \pi \ii} F^c(z, x)  \oint_{0 , - \rho' , -1  } \frac{\dd w}{2 \pi \ii} G^c(w, y)  \frac{1}{w-z}, \nonumber \\ 
A_j^c(x) := & \oint_1 \frac{\dd w}{2 \pi \ii} F^c(w, x)  { \left( \frac{w}{1 + w} \right) }^{j-1} \frac{1}{1 + w}, \label{def:Acj}\\
B^c(y) := &\oint_{0 , - \rho' , -1} \frac{ \dd w }{2 \pi \ii}  G^c(w , y) \frac{1}{w+1},\label{def:Bc}
\end{align}
where the functions $F^c(z,x)$ and $G^c(w,y)$ are defined in \eqref{def:FcGc}.
\end{prop}

\medskip

In order to prove Proposition~\ref{rank1 perturbation det2}, we need to take some steps which consist of proving Lemmas~\ref{rank1 perturbation term1}, \ref{rank1 perturbation term2} and \ref{kernel holomorphic}.
In the proofs of these Lemmas, the asymmetric part of $L(\vec{z})$, $\Delta_{n-1}(-z)/\prod_{j=1}^{n-1}{(z_j - 1)}^{j+1}$, plays a crucial role in algebraic manipulations.
We utilise the properties that $\Delta_{n-1}(-z)$ is anti-symmetric under exchanges $z_i \leftrightarrow z_j$ for any pairs $(i , j) \in {[1,n-1]}^2$ and  $z$-integrand, namely $L(\vec{z})$, has the poles of the order $j+1$ at $z_j = 1$ for any $j \in [1,n-1]$.
On the other hand, the detailed form of the other part, which is symmetric function of $z_j$s, is not important as long as it does not produce any additional singularities. Therefore, for convenience, we define \begin{equation*}
\label{def:h}
h(z) := (1 - \rho' z) {\left( \frac{2}{1 + z} \right)}^m \mathrm{e}^{ (z^{-1} - 1) \frac{t}{2} }
\end{equation*}
and, in terms of $h(z)$, we rewrite $L(\vec{z})$ defined in \eqref{def:L(z)} as
\begin{equation}
\label{simple:L(z)}
	    L(\vec{z}) =  \frac{\e^{-\rho t/2}}{(\rho')^{n-1}  } \left(\frac{2(1-\rho)}{2-\rho}\right)^m  \frac{ {(-1)}^{n-1} \Delta_{n-1}(-z) }
		{\prod_{j=1}^{n-1} (z_j-1)^{j+1} } \prod_{j=1}^{n-1}{h(z_j)} ,
\end{equation}
where $\Delta_{n-1}(z) = \prod_{1 \leq i < j \leq n-1}{(z_i - z_j)}$.

As a first step, we separate each entry of the determinant into the main contribution coming from $z_j = 1$ and the other.

\medskip

\begin{lemman}
\label{rank1 perturbation term1}
For any $(x , y) \in \mathbb{N}^2$, $\rho \in (0,1)$, $t > 0$ and $n, m \in \mathbb{N}$, the following equality holds:

\begin{equation}\label{eq:rank1lem1}
\oint_{1} \dd^{n-1} z\, L(\vec{z}) K^c(x ,y , \vec{z}) = \oint_{1} \dd^{n-1} z\, L(\vec{z})  \left[ K_{\mathrm{W}}^c(x , y) -  (z_1 - 1) A^c(x, \vec{z}) B^c(y) \right],
\end{equation}
where $A^c(x , \vec{z})$ is defined as
\begin{align*}
A^c(x, \vec{z}) := & \oint_{1} \frac{ \dd z }{2 \pi \ii} F^c(z, x)  \frac{1}{ 1+z} \prod_{j=2}^{n-1}{ \frac{1 + z_j z}{1+z} }. 
\end{align*}
\end{lemman}

\begin{proof}

We carefully consider the contributions of the poles at $z_j=1$. Considering $L(\vec{z})$ is expressed as \eqref{simple:L(z)}, the left hand side of \eqref{eq:rank1lem1} except for the constant factor is explicitly given by 

\begin{multline}
\label{eq:rank1lem1_1}
 \oint_{1} \prod_{j=1}^{n-1}{ \frac{\dd z_j}{2 \pi \ii} } \frac{ \Delta_{n-1}(-z)}{ \prod_{j=1}^{n-1}{{(z_j -1)}^{j+1}} } \prod_{j=1}^{n-1}{h(z_j)} \times \\
  \oint_{1} \frac{\dd z}{2 \pi \ii} F^c(z , x) \prod_{j=1}^{n-1}{ \frac{1 + z_j z}{1 + z} } \oint_{C_1} \frac{\dd w}{2 \pi \ii} G^c(w , y) \prod_{j=1}^{n-1}{\frac{1 + w}{1 + z_j w}} \frac{1}{w-z},
\end{multline}
where the contour $C_k$ lies around the poles in $w$ at 0, $ - \rho'$ , $ -1$ , and $ \{ - z_j^{-1} \}_{j=k}^{n-1}$.

In the following, we choose the contours of $\vec{z}$-integrals such that they are included in the open polydisc ${S(1,1)}^{n-1}$ where $S(a,r) := \{ z \in \mathbb{C} \mid |z - a| < r \}$.
Since the integral with respect to $w$ has an essential singularity at the infinity and its contour is chosen such that it excludes the infinity, the contour of the integral in terms of the variable $1/w$ includes $\infty, -1/\rho', -1$, $-z_j$'s and excludes the origin. 
For $\vec{z} \in S(1,r)^{n-1}$ with $r \in (0,1)$, to compute the residues at $z_j=1$, we choose the contour such that $r < |1+1/w|$ holds. On such a contour, the following expansion converges uniformly,
\be
\label{eq:Taylorexp}
\frac{1 + w}{1 + z_j w} = \sum_{i=0}^{\infty}{ {\left( \frac{w(1 - z_j)}{1 + w} \right)}^i },  \qquad \underset{1 \leq j \leq n-1}{ \max}  \left| \frac{w (1 -z_j)}{w +1} \right|  \le \varepsilon < 1 ,
\ee
and sum in \eqref{eq:Taylorexp} and integration with respect to $z_j$'s can be interchanged.
First, we perform the expansion in $z_1$.  Because the pole at $z_1=1$ is of order 2, only the first two terms give a nonzero contribution to \eqref{eq:rank1lem1_1}. Therefore \eqref{eq:rank1lem1_1} can be written as $I_1^{(2)} + I_2^{(2)}$ where we define $I_1^{(k)}$ and $I_2^{(k)}$ for $ k \in [2, n] $ as

\begin{multline*}
\label{def:I1k}
I_1^{(k)} =   \oint_{1} \prod_{j=1}^{n-1}{ \frac{\dd z_j}{2 \pi \ii} } \frac{ \Delta_{n-1}(-z)}{ \prod_{j=1}^{n-1}{{(z_j -1)}^{j+1}} } \prod_{j=1}^{n-1}{h(z_j)} \times \\
 \oint_{1} \frac{\dd z}{2 \pi \ii} F^c(z , x) \prod_{j=1}^{n-1}{ \frac{1 + z_j z}{1 + z} } \oint_{C_{k}} \frac{\dd w}{2 \pi \ii} G^c(w , y)  \prod_{j=k}^{n-1}{\frac{1 + w}{1 + z_j w}}  \frac{1}{w-z}, 
\end{multline*}
\begin{multline*}
I_2^{(k)} =   - \oint_{1} \prod_{j=1}^{n-1}{ \frac{\dd z_j}{2 \pi \ii} } \frac{ \Delta_{n-1}(-z) }{ \prod_{j=2}^{n-1}{{(z_j -1)}^{j+1}} (z_1-1) } \prod_{j=1}^{n-1}{h(z_j)} \times \\
 \oint_{1} \frac{\dd z}{2 \pi \ii} F^c(x, z) \prod_{j=1}^{n-1}{ \frac{1 + z_j z}{1 + z} } \oint_{C_{k}} \frac{\dd w}{2 \pi \ii} G^c(w, y)  \frac{w}{1+w}  \prod_{j=k}^{n-1}{\frac{1 + w}{1 + z_j w}}  \frac{1}{w-z} .
\end{multline*}
As a convention   we set the factor $\prod_{j=k}^{n-1} \frac{1+w}{1+z_j w}$ to unity for $k=n-1$.

Secondly, we will show that $I_1^{(2)} + I_2^{(2)}$ is the same as $I_1^{(n)}+I_2^{(n)}$, in which $w$-integrands do not depend on $z_j$'s.
To show this, we prove that the equality

\be
\label{eq:rank1lem1_3}
I_1^{(k)} + I_2^{(k)} = I_1^{(k+1)} + I_2^{(k+1)}
\ee
holds for any $k \in [2,n-1]$. We expand $\frac{1 + w}{1 + z_{k} w}$ in $I_1^{(k)}$ around $z_{k}=1$ as in \eqref{eq:Taylorexp}.  The terms of order $k+1$ or higher in this expansion make the integrand analytic in $z_{k}$ and hence give zero residue. The remaining terms become

\begin{multline*}
I_1^{(k)} =  \sum_{i=0}^{k}{} \oint_{1} \prod_{j=1}^{n-1}{ \frac{\dd z_j}{2 \pi \ii} } \frac{ \Delta_{n-1}(-z)}{ \prod_{j=1, j \neq k}^{n-1}{{(z_j -1)}^{j+1}} {(z_{k} -1 )}^{k+1 - i}  } \prod_{j=1}^{n-1}{h(z_j)} \times \\ \oint_{1} \frac{\dd z}{2 \pi \ii} F^c(z, x) \prod_{j=1}^{n-1}{ \frac{1 + z_j z}{1 + z} } 
 \oint_{C_{k+1} } \frac{\dd w}{2 \pi \ii} G^c(w, y)  {\left( \frac{-w}{1 + w} \right)}^i  \prod_{j=k+1}^{n-1}{\frac{1 + w}{1 + z_j w}}  \frac{1}{w-z}.
\end{multline*}
The terms with $i=1,\ldots, k-1$ vanish because for each of those terms the integrand is anti-symmetric under exchange $z_{k-i} \leftrightarrow z_{k}$. Hence only the terms with $i=0$ and $i=k$ survive,
\begin{multline}
\label{eq:rank1lem1_4}
I_1^{(k)}  =  \oint_{1} \prod_{j=1}^{n-1}{ \frac{\dd z_j}{2 \pi \ii} } \frac{ \Delta_{n-1}(-z)}{ \prod_{j=1, j \neq k}^{n-1}{{(z_j -1)}^{j+1}}  (z_{k}-1) } \prod_{j=1}^{n-1}{h(z_j)} \oint_{1} \frac{\dd z}{2 \pi \ii} F^c(z, x) \prod_{j=1}^{n-1}{ \frac{1 + z_j z}{1 + z} } \times \\
\oint_{C_{k+1} } \frac{\dd w}{2 \pi \ii} G^c(w, y)   \left( \frac{1}{(z_{k}-1)^{k}} +\left( \frac{-w}{1 + w} \right)^{k}  \right) \prod_{j=k+1}^{n-1}{\frac{1 + w}{1 + z_j w}}  \frac{1}{w-z}.
\end{multline}
Likewise, expanding $ \frac{1 + w}{1 + z_{k} w}$ around $z_{k}=1$  in $I_2^{(k)}$ and making use of analytic and symmetry properties of the integrand, we obtain
\begin{multline}
\label{eq:rank1lem1_5}
I_2^{(k)} =  - \oint_{1} \prod_{j=1}^{n-1}{ \frac{\dd z_j}{2 \pi \ii} } \frac{ \Delta_{n-1}(-z) }{ \prod_{j=2,  j \neq k}^{n-1}{{(z_j -1)}^{j+1}} (z_{k}-1)^2 (z_1-1) } \prod_{j=1}^{n-1}{h(z_j)} \oint_{1} \frac{\dd z}{2 \pi \ii} F^c(z, x) \prod_{j=1}^{n-1}{ \frac{1 + z_j z}{1 + z} } \times \\
\oint_{C_{k+1}} \frac{\dd w}{2 \pi \ii} G^c(w, y) \frac{w}{1+w} \left(  \frac{1}{(z_{k}-1)^{k-1}} + \left(\frac{-w}{1 + w} \right)^{k-1}  \right) \prod_{j=k+1}^{n-1}{\frac{1 + w}{1 + z_j w}}  \frac{1}{w-z} .
\end{multline}
The second terms of \eqref{eq:rank1lem1_4} and \eqref{eq:rank1lem1_5} differ only in terms of the orders of the poles at $z_1=1$ and $z_k=1$. It is easy to see that the integrand of the sum of these terms is anti-symmetric under the exchange $z_1 \leftrightarrow z_{k}$ and thus it is equal to zero.
On the other hand, the first terms of  \eqref{eq:rank1lem1_4} and \eqref{eq:rank1lem1_5} coincide with $I_1^{(k+1)}$ and $I_2^{(k+1)}$, respectively.
Thus, it turns out that the equality \eqref{eq:rank1lem1_3} holds for any $k \in [2,n-1]$.
Using the equality \eqref{eq:rank1lem1_3} from $k=2$ to $k=n-1$ consecutively, we obtain
\begin{equation*}
\label{eq:rank1lem1_7}
I_1^{(2)} + I_2^{(2)} = I_1^{(n)} + I_2^{(n)}.
\end{equation*}

Thirdly, we further manipulate $I_1^{(n)}$ and $I_2^{(n)}$.  Considering the equality $\frac{1 + z_1 z}{1 + z} =  1 + \frac{ z(z_1 - 1) }{1 + z} $, $I_1^{(n)}$ can be divided into two terms as
\begin{equation}
\label{eq:rank1lem1_8}
\begin{split}
I_1^{(n)} = & \oint_{1} \prod_{j=1}^{n-1}{ \frac{\dd z_j}{2 \pi \ii} } \frac{ \Delta_{n-1}(-z) }{ \prod_{j=1}^{n-1}{{(z_j -1)}^{j+1}} } \prod_{j=1}^{n-1}{h(z_j)}  \times \\
& \hspace{1in} \oint_{1} \frac{\dd z}{2 \pi \ii} F^c(z, x) \prod_{j=2}^{n-1}{ \frac{1 + z_j z}{1 + z} }  \oint_{0 , - \rho' , -1  } \frac{\dd w}{2 \pi \ii} G^c(w, y)   \frac{1}{w-z} \\
+ &  \oint_{1} \prod_{j=1}^{n-1}{ \frac{\dd z_j}{2 \pi \ii} } \frac{ \Delta_{n-1}(-z) }{ \prod_{j=2}^{n-1}{{(z_j -1)}^{j+1}} (z_1 - 1) } \prod_{j=1}^{n-1}{h(z_j)} \times \\
& \hspace{1in}  \oint_{1} \frac{\dd z}{2 \pi \ii} F^c(z, x) \left( \frac{z}{1 + z} \right)  \prod_{j=2}^{n-1}{ \frac{1 + z_j z}{1 + z} } 
\oint_{0 , - \rho' , -1  } \frac{\dd w}{2 \pi \ii} G^c(w, y)  \frac{1}{w-z} .
\end{split}
\end{equation}
Applying the expansion

\begin{equation}
\label{eq:rank1lem1_9}
\prod_{j=2}^{n-1}{ \frac{1 + z_j z}{1 + z} } =   1  +   \sum_{j=1}^{n-2}{ { \left( \frac{z}{1 + z} \right) }^{j} \sum_{2 \leq \ell_1 < \ell_2 < \dots < \ell_j \leq n-1 }{ \prod_{i=1}^{j}{ (z_{\ell_i} - 1) } } }
\end{equation}
to the first term of \eqref{eq:rank1lem1_8}, it turns out that unity remains and all terms in the sums vanish because for each term there exists a pair $(i , j) \in {[1,n-1]}^2$ such that $z$-integrand is anti-symmetric under exchange $z_i \leftrightarrow z_j$. Thus, the first term of \eqref{eq:rank1lem1_8} reduces to

\begin{equation*}
\label{eq:rank1lem1_10}
\oint_{1} \prod_{j=1}^{n-1}{ \frac{\dd z_j}{2 \pi \ii} } \frac{ \Delta_{n-1}(-z)}{ \prod_{j=1}^{n-1}{{(z_j -1)}^{j+1}} } \prod_{j=1}^{n-1}{h(z_j)} \oint_{1} \frac{\dd z}{2 \pi \ii} F^c(z, x)  \oint_{0 , - \rho' , -1  } \frac{\dd w}{2 \pi \ii} G^c(w, y)  \frac{1}{w-z},
\end{equation*}
and it coincides with the first term of the right hand side of \eqref{eq:rank1lem1}.

Likewise, considering the equality $\frac{1 + z_1 z}{1 + z} =  1 + \frac{ z(z_1 - 1) }{1 + z} $, $I_2^{(n)}$ can be divided into two terms and the second term vanishes because $z_1$-integrand is regular in the vicinity of $z_1 = 1$, i.e.

\begin{multline*}
I_2^{(n)} =  - \oint_{1} \prod_{j=1}^{n-1}{ \frac{\dd z_j}{2 \pi \ii} } \frac{ \Delta_{n-1}(-z)}{ \prod_{j=2}^{n-1}{{(z_j -1)}^{j+1}} (z_1-1) } \prod_{j=1}^{n-1}{h(z_j)} \times \\
 \oint_{1} \frac{\dd z}{2 \pi \ii} F^c(z, x)   \prod_{j=2}^{n-1}{ \frac{1 + z_j z}{1 + z} }  \oint_{0 , - \rho' ,  -1  } \frac{\dd w}{2 \pi \ii} G^c(w, y) \frac{w}{1+w}   \frac{1}{w-z} .
\end{multline*}
From the above, the sum of $I_2^{(n)}$ and the second term of \eqref{eq:rank1lem1_8} can be described as the product of a function of $x$ and that of $y$, i.e.

\begin{align}
& \oint_{1} \prod_{j=1}^{n-1}{ \frac{\dd z_j}{2 \pi \ii} } \frac{ \Delta_{n-1}(-z)}{ \prod_{j=2}^{n-1}{{(z_j -1)}^{j+1}} (z_1 - 1) } \prod_{j=1}^{n-1}{h(z_j)} \times \nonumber \\
& \hspace{1in} \oint_{1} \frac{\dd z}{2 \pi \ii} F^c(z, x)   \prod_{j=2}^{n-1}{ \frac{1 + z_j z}{1 + z} } \oint_{0 , - \rho' , -1  } \frac{\dd w}{2 \pi \ii} G^c(w, y) \left[ \frac{z}{1 +z} - \frac{w}{1 + w} \right]  \frac{1}{w-z} \nonumber \\
= & -  \oint_{1} \prod_{j=1}^{n-1}{ \frac{\dd z_j}{2 \pi \ii} } \frac{ \Delta_{n-1}(-z)}{ \prod_{j=2}^{n-1}{{(z_j -1)}^{j+1}} (z_1 - 1) } \prod_{j=1}^{n-1}{h(z_j)} \times \nonumber \\
& \hspace{1in}   \oint_{1} \frac{\dd z}{2 \pi \ii} F^c(z, x) \frac{1}{1 + z}   \prod_{j=2}^{n-1}{ \frac{1 + z_j z}{1 + z} }   \oint_{0 , - \rho' , -1  } \frac{\dd w}{2 \pi \ii} G^c(w, y)  \frac{1}{1 + w}, \nonumber 
\end{align}
and it coincides with the second term of the right hand side of \eqref{eq:rank1lem1}. Putting everything together, it follows that
\begin{align}
&\oint_{1} \prod_{j=1}^{n-1}{ \frac{\dd z_j}{2 \pi \ii} } \frac{ \Delta_{n-1}(-z) }{ \prod_{j=1}^{n-1}{{(z_j -1)}^{j+1}}} \prod_{j=1}^{n-1}{h(z_j)} K^c(x , y, \vec{z})  \nonumber \\
=& I_1^{(2)} + I_2^{(2)} = I_1^{(n)} + I_2^{(n)} \nonumber  \\
=& \oint_{1} \prod_{j=1}^{n-1}{ \frac{\dd z_j}{2 \pi \ii} } \frac{ \Delta_{n-1}(-z) }{ \prod_{j=1}^{n-1}{{(z_j -1)}^{j+1}}} \prod_{j=1}^{n-1}{h(z_j)} \left[ K^c_{\mathrm{W}}(x , y) - (z_1 - 1) A^c(x, \vec{z})  B^c(y) \right]. \nonumber
\end{align}
\qed
\end{proof}
From the claim of Lemma~\ref{rank1 perturbation term1}, we can easily obtain the following Lemma~\ref{rank1 perturbation term2}.
Although this looks similar to Lemma~\ref{rank1 perturbation term1}, they will be applied in different manners  in the proof of Proposition~\ref{rank1 perturbation det2}.

\medskip

\begin{lemman}
\label{rank1 perturbation term2}

For any $(x , y) \in \mathbb{N}^2$, $\rho \in (0,1)$, $t > 0$ and $n, m \in \mathbb{N}$, the following equality holds:

\begin{multline}\label{eq:rank1lem2}
\oint_{1} \dd^{n-1} z\, L(\vec{z}) \left[ K_{\mathrm{W}}^c(x, y) - (z_1 - 1) A^c(x , \vec{z} )  B^c(y) \right] \\
= \oint_{1} \dd^{n-1} z\, L(\vec{z}) \left\{ K_{\mathrm{W}}^c(x, y) - \left[ \sum_{j=1}^{n-1}{ \prod_{i=1}^{j}{ {(z_i - 1)} A_j^c(x) }} \right] B^c(y) \right\},
\end{multline}
where $A_j^c(x)$ is defined as \eqref{def:Acj}.

\end{lemman}

\begin{proof}

Applying the expansion \eqref{eq:rank1lem1_9} to $A^c(x, \vec{z})$ in the left hand side of \eqref{eq:rank1lem2}, it turns out that, for any $j \in [1, n-2]$, the only terms of which $(\ell_1, \ell_2, \dots, \ell_j)$ is the consecutive numbers $(2 , 3, \ldots , j+1)$ remain and the others vanish.
That is, the second term of the left hand side of \eqref{eq:rank1lem2} except for the constant factor can be described as

\begin{equation}
\label{deformation_second_term}
\begin{split}
& - \oint_{1} \prod_{j=1}^{n-1}{ \frac{\dd z_j}{2 \pi \ii} } \frac{ \Delta_{n-1}(-z) \prod_{j=1}^{n-1}{h(z_j)} }{ \prod_{j=2}^{n-1}{{(z_j -1)}^{j+1}} (z_1 - 1) }  \oint_{1} \frac{\dd w}{2 \pi \ii} \frac{F^c(w, x)}{1+w} \left\{ 1  +   \sum_{j=1}^{n-2}{ { \left( \frac{w}{1 + w} \right) }^{j} \prod_{i=2}^{j + 1}{ (z_{i} - 1) } }  \right\} B^c(y) \\
= & - \oint_{1} \prod_{j=1}^{n-1}{ \frac{\dd z_j}{2 \pi \ii} } \frac{ \Delta_{n-1}(-z)\prod_{j=1}^{n-1}{h(z_j)} }{ \prod_{j=1}^{n-1}{{(z_j -1)}^{j+1}}} \left[ \sum_{j=1}^{n-1}{ \prod_{i=1}^{j}{(z_i - 1)} \oint_{1} \frac{\dd w}{2 \pi \ii} \frac{F^c(w, x)}{1+w} { \left( \frac{w}{1 + w} \right) }^{j-1} } \right] B^c(y).
\end{split}
\end{equation}
From \eqref{def:Acj}, the definition of $A_j^c(x)$, it turns out that the right hand side of \eqref{deformation_second_term} coincides with the second term of the right hand side of \eqref{eq:rank1lem2}.
Therefore, the equality \eqref{eq:rank1lem2} holds.
\qed
\end{proof}
As a second step, we see that $K^c(x_i, x_j ,\vec{z})$ is a holomorphic function in vicinity of $\vec{z} = \vec{1}$.
It guarantees that Lemma~\ref{rank1 perturbation term2} is extended to the determinant as will be mentioned in the proof of Proposition~\ref{rank1 perturbation det2}. 

\medskip

\begin{lemman}
\label{kernel holomorphic}
For any $(x , y) \in \mathbb{N}^2$, $\rho \in (0,1)$, $t > 0$ and $n, m \in \mathbb{N}$, the kernel $K^c(x,y,\vec{z})$ defined as \eqref{Kxyz} is a holomorphic function of $\vec{z} = (z_1,\ldots, z_{n-1})$ in ${S(1,1)}^{n-1}$, where $S(a,r) := \{ z \in \mathbb{C} \mid |z - a| < r \}$ for $a \in \mathbb{C}$.

\end{lemman}

\begin{proof}
See Appendix~\ref{ap:kernel holomorphic}.
\qed
\end{proof}

\noindent {\it Proof of Proposition~\ref{rank1 perturbation det2}.} It now follows from Lemma~\ref{rank1 perturbation term1}, \ref{rank1 perturbation term2} and \ref{kernel holomorphic}. Details will be given in 
Appendix~\ref{ap:rank1_determinant}. \qed

\medskip

In the following, we will divide $\mathcal{I}_2$ given by \eqref{eq:prob2term2} into the dominant contribution and the subdominant term, which converge to the GUE Tracy-Widom distribution $F_2(s_2)$ and zero in the long time limit, respectively, with the scaling \eqref{nmscale}.
Applying Proposition~\ref{rank1 perturbation det2} to each term of the series in the right hand side of \eqref{detdivide_1_pre2}, we obtain

\begin{equation}
\begin{split}
\label{detdivide_1}
\mathcal{I}_2 
& =  I_z + \sum_{k=1}^{m}{ \frac{{(-1)}^k}{k!} } \sum_{x_1 = 1}^{\infty}{ \sum_{x_2 = 1}^{\infty}{ \cdots \sum_{x_k = 1}^{\infty}{  \oint_1 \dd^{n-1} z\,  L(\vec{z}) \det\left[ K^c(x_i, x_j , \vec{z}) \right]_{1 \leq i , j \leq k} } } } \\
& \begin{multlined} = I_z + \sum_{k=1}^{m}{ \frac{{(-1)}^k}{k!} } \sum_{x_1 = 1}^{\infty}{} \sum_{x_2 = 1}^{\infty}{} \cdots \sum_{x_k = 1}^{\infty}{ }  \oint_1 \dd^{n-1} z\,  L(\vec{z}) \times  \\ \qquad\qquad\qquad\qquad\qquad\qquad \det\left\{ K_{\mathrm{W}}^c(x_i, x_j ) - \left[ \sum_{p=1}^{n-1}{ \prod_{q=1}^{p}{ (z_q - 1) } A^c_p(x_i) }  \right] B^c(x_j) \right\}_{1 \leq i , j \leq k} .
\end{multlined}
\end{split}
\end{equation}
The matrices in the determinant parts for each term of \eqref{detdivide_1} can be expressed as a product of the kernel and the identity matrix minus a rank-one matrix, i.e.,

\begin{multline*}
K^c_{\mathrm{W}}(x_i , x_j) - \left[ \sum_{p=1}^{n-1}{ \prod_{q=1}^{p}{ {(z_q - 1)} A^c_{p}(x_i) }} \right] B^c(x_j)  \\
= \sum_{s = 1}^{k}{  K^c_{\mathrm{W}}(x_i , x_s) \left\{  \delta_{s,j} - \sum_{\ell =1}^{k}{ {\left( K^c_{\mathrm{W}} \right)}^{-1}(x_s, x_\ell) \left[ \sum_{p=1}^{n-1}{ \prod_{q=1}^{p}{ {(z_q - 1)} A^c_{p}(x_\ell) }} \right] B^c(x_j) }  \right\}  }.
\end{multline*}
Since a determinant of a product of two matrices equals a product of determinants of the two matrices and a determinant of the identity matrix minus a rank-one matrix equals $1$ minus the trace of the rank-one matrix, the determinant seen in \eqref{detdivide_1} is divided into two terms.

\begin{equation}
\label{detdivide_2}
\begin{split}
& \det\left\{ K^c_{\mathrm{W}}(x_i , x_j) - \left[ \sum_{p=1}^{n-1}{ \prod_{q=1}^{p}{ {(z_q - 1)} A^c_{p}(x_i) }} \right] B^c(x_j)   \right\}_{ 1 \leq i , j \leq k } \\
= & \det\left[  K^c_{\mathrm{W}}(x_i , x_j) \right]_{1 \leq i , j \leq k}   \times  \det\left\{ \delta_{i,j} - \sum_{\ell=1}^{k}{ {(K^c_{\mathrm{W}})}^{-1}(x_i , x_\ell) \left[ \sum_{p=1}^{n-1}{ \prod_{q=1}^{p}{ {(z_q - 1)} A^c_{p}(x_\ell) }} \right] } B^c(x_j)   \right\}_{ 1 \leq i , j \leq k } \\
=& \det\left[  K^c_{\mathrm{W}}(x_i , x_j) \right]_{1 \leq i , j \leq k}  \times \left\{ 1 -  \sum_{r, \ell =1}^{k}{   {(K^c_{\mathrm{W}})}^{-1}(x_r , x_\ell) \left[ \sum_{p=1}^{n-1}{ \prod_{q=1}^{p}{ {(z_q - 1)} A^c_{p}(x_\ell) }} \right]  B^c(x_r)  } \right\}  \\
=& \det\left[ K^c_{\mathrm{W}}(x_i , x_j) \right]_{1 \leq i , j \leq k} - \sum_{r, \ell =1}^{k}{  {(-1)}^{r+\ell} \det\left[  K^c_{\mathrm{W}}(x_i , x_j) \right]_{ \substack{ 1 \leq i , j \leq k \\  i \neq \ell , j \neq r } }  \left[ \sum_{p=1}^{n-1}{ \prod_{q=1}^{p}{ {(z_q - 1)} A^c_{p}(x_\ell) }} \right]  B^c(x_r)  }.
\end{split}
\end{equation}
It follows from \eqref{detdivide_1} and \eqref{detdivide_2} that $\mathcal{I}_2$ can be written as

\begin{multline*}
\label{detdivide_3}
\mathcal{I}_2 =  I_z + \sum_{k=1}^{m}{ \frac{{(-1)}^k}{k!} } \sum_{x_1 = 1}^{\infty}{  \cdots \sum_{x_k = 1}^{\infty}{ I_z \times \det\left[  K^c_{\mathrm{W}}(x_i , x_j) \right]_{1 \leq i , j \leq k}   }  } \\
-  \sum_{k=1}^{m}{ \frac{{(-1)}^k}{k!} } \sum_{x_1 = 1}^{\infty}{ } \cdots \sum_{x_k = 1}^{\infty}{ } \oint_1 \dd^{n-1} z\,  L(\vec{z}) \sum_{r, \ell =1}^{k}{   {(-1)}^{r+\ell} } \times \\
\det\left[  K^c_{\mathrm{W}}(x_i , x_j) \right]_{ \substack{ 1 \leq i , j \leq k \\  i \neq \ell , j \neq r } }  \left[ \sum_{p=1}^{n-1}{ \prod_{q=1}^{p}{ {(z_q - 1)} A^c_{p}(x_\ell) }} \right]  B^c(x_r) .
\end{multline*}
As for the second term, $I_z$ can be moved to outside of the sums from $x_i = 1$ to $\infty$ for any $i \in [1, n-1]$ because $\det\left[  \bar{K}^c_{\mathrm{W},t}(x_i , x_j) \right]_{1 \leq i , j \leq k}$ does not depend on $z_j$'s.
For the same reason, about the final term, $\oint_1 \dd^{n-1} z\,  L(\vec{z})$ can be moved to inside the sum from $p=1$ to $n-1$.
Then, we arrive at the following proposition.

\medskip

\begin{prop}
The integral $\mathcal{I}_2$ given in \eqref{eq:prob2term2} and \eqref{I2 Iz antisym} can be written into
\begin{equation}
\mathcal{I}_2=\mathcal{I}^{(1)}_2 - \mathcal{I}^{(2)}_2,
\label{I2div}
\end{equation}
where
\begin{equation}
\mathcal{I}^{(1)}_2 = 
I_z \times {\det\left( 1 - K^c_{\mathrm{W}} \right)}_{\ell^2(\mathbb{N})},
\label{def:I2(1)}
\end{equation}
\begin{multline}
\mathcal{I}^{(2)}_2 =  \sum_{k=1}^{m} \frac{(-1)^k}{ k! }
\sum_{x_1 = 1}^{\infty} { \!\! \dots \!\!\! \sum_{x_{k} = 1}^{\infty}
\sum_{r,\ell=1}^{k}{ {(-1)}^{r+\ell} }} \times \\
\det\left[  K^c_{\mathrm{W}}(x_i , x_j) \right]_{ \substack{ 1 \leq i , j \leq k \\  i \neq \ell , j \neq r } } \left[ \sum_{p=1}^{n-1}{  \oint_1 \dd^{n-1} z\,  L(\vec{z}) \prod_{q=1}^{p}{ {(z_q - 1)} A^c_{p}(x_\ell) }} \right]  B^c(x_r).\label{def:I2(2)} 
\end{multline}

\end{prop}

\subsection{Evaluation of $\mathcal{I}^{(1)}_2$}
\label{subsec:Ev_I21}
We can separate the estimate of $\mathcal{I}^{(1)}_2$ into $I_z$ and ${\det\left( 1 - K^c_{\mathrm{W}} \right)}_{\ell^2(\mathbb{N})}$ since $K^c_{\mathrm{W}}(x,y)$ does not contain $\vec{z}$. The analysis of ${\det\left( 1 - K^c_{\mathrm{W}} \right)}_{\ell^2(\mathbb{N})}$ follows exactly the same arguments as in Section \ref{sec: I1 limit}, while the limit of $I_z$ is obtained by the same idea in \cite{borodin2008transition,borodin2007fluctuationsintasep,BFPS2007}.

\subsubsection{Evaluation of ${\det\left( 1 - K^c_{\mathrm{W}} \right)}_{\ell^2(\mathbb{N})}$}

In order to perform asymptotic analysis, we define the rescaled functions which depend on the positive real numbers $\xi = x/ \lambda_c t^{1/3}$ and $\zeta = y/ \lambda_c t^{1/3}$ with $x, y \in \mathbb{N}$ as

\begin{equation}\label{def:Kcresc}
\bar{K}^c_{\mathrm{W},t}(\xi,\zeta) = {(w_c + c)}^{ \lambda_c t^{1/3} (\xi - \zeta) } \lambda_c t^{1/3} K^c_{\mathrm{W}}(\lambda_c t^{1/3} \xi , \lambda_c t^{1/3} \zeta),
\end{equation}
where $\lambda_c$ is defined in \eqref{DefLambdac}.
The rescaled kernel is explicitly described as

\begin{equation*}
\label{RescKernelc2}
\begin{split}
& \bar{K}_{\mathrm{W},t}^c(\xi, \zeta) \\
=& \lambda_c t^{1/3} \oint_1 \frac{\dd z}{2 \pi \ii} \left( \frac{z + \rho'}{z + 1} \right) \e^{ f(z, t, \xi) - f(w_c , t,\xi) } \oint_{0,-\rho',-1} \frac{\dd w}{2 \pi \ii} \left( \frac{w + 1}{w + \rho'} \right) \e^{-f(w,t,\zeta) + f(w_c , t,\zeta) + g(w)}  \frac{1}{w - z},
\end{split}
\end{equation*}
where $g(w) = -\ln{(w+c)}$ and $f(z,t,\xi) = g_1(z) t + g_2(z) t^{1/2} + g_3(z,\xi) t^{1/3}$ with the functions $g_1(z)$, $g_2(z)$ and $g_3(z,\xi)$ defined as \eqref{g def1}.
Note that $\bar{K}_{\mathrm{W},t}^c$ differs from $\bar{K}^c_t$ in the sense that it has $(z + \rho')(w + 1)/[(z + 1)(w + \rho')]$ in the integrands and its contour includes also $w=-\rho'$.
Since the contour of the integral with respect to $z$ and $w$ used in the proofs of Propositions \ref{uniform convergence kernelc} and \ref{exp_bound} include also the point $w = -\rho'$, and $(z + \rho')(w + 1)/[(z + 1)(w + \rho')]$ is a regular function on the contours and equals unity at the point $(z, w) = (w_c, w_c)$, we can prove the following Propositions that correspond to Propositions \ref{uniform convergence kernelc} and \ref{exp_bound}.

\medskip
\begin{prop}
\label{uniform convergence kernelc2}
{\rm \textbf{(Uniform convergence of $\boldsymbol{\bar{K}^c_{\mathrm{W},t}}$ on a bounded set)}}
Let $m,n$ be scaled as \eqref{nmscale}. Then for any fixed $L>0$, the rescaled kernel $\bar{K}_{\mathrm{W},t}^{c} (\xi , \zeta)$ defined in \eqref{def:Kcresc} converges uniformly on $\xi,\kappa\in [-L,L]$ to
\begin{equation*}
\label{KcW limit}
\lim_{t \rightarrow \infty}
\bar{K}_{\mathrm{W},t}^{c} (\xi , \zeta)
= A(\xi + s_2 , \zeta + s_2), 
\end{equation*}
where $A(x,y)$ is defined as \eqref{def:airy_kernel} and $s_2$ is given in \eqref{s+-}. 
\end{prop}

\medskip
\begin{prop}
\label{exp_bound_W}
{\rm \textbf{(Estimate of kernel $\boldsymbol{\bar{K}^c_{\mathrm{W},t}}$)} }
For any $L > 0$ large enough, there exists $t_0(L) > 0$ such that, for any $t \geq t_0$, the estimate

\begin{equation*}\label{exp_bound_W_claim}
\left| \bar{K}_{\mathrm{W},t}^{c}(\xi, \zeta) \right| \leq \mathrm{e}^{-(\xi + \zeta)}
\end{equation*}
holds for $( \xi , \zeta ) \in {[0 , \infty)}^2 \backslash {[0 , L]}^2 $.
\end{prop}
From Lemma \ref{lem:kernel lim}, it turns out that ${\det\left( 1 - K^c_{\mathrm{W}} \right)}_{\ell^2(\mathbb{N})}$ converges the GUE Tracy-Widom distribution in the long time limit with the scaling \eqref{nmscale}.

\begin{equation}
\label{detKcW limit}
\lim_{t\rightarrow\infty} \det(1 - \bar{K}^c_{\mathrm{W},t})_{\ell^2(\mathbb{N}/(\lambda_c t^{1/3}))} =  \lim_{t\rightarrow\infty} \det(1 - \bar{K}^c_{\mathrm{W},t})_{L^2(0,\infty)} = \det(1-A)_{L^2(s_2 , \infty)} = F_2(s_2)
\end{equation}

\subsubsection{Evaluation of $I_z$}
\label{se:Iz}

We are now left with the $\vec{z}$-integrals, namely $I_z$, defined in \eqref{def:Iz}.
From the definition of $L(\vec{z})$ given in \eqref{def:L(z)}, $I_z$ turns out to be expressed as

\begin{align*}
I_z
= 
\frac{\e^{-\rho t/2}}{(\rho')^{n-1}} \left(\frac{2( 1- \rho)}{2 - \rho}\right)^m  \frac{{(-1)}^{n-1}}{(n-1)!}
\oint_{1}\dd^{n-1}z
\frac{\e^{\Lambda_{n-1,0} t}
\Delta_{n-1}(z) \Delta_{n-1}(-z)
\prod_i^{n-1}(1-\rho'z_i)}
{\prod_{i=1}^{n-1}(z_i-1)^{n}
[\frac{1}{2}(z_i+1)]^m}.
\end{align*}

It is hard to apply the method introduced in Section \ref{se:tofredholm}, since the order of the pole at $z_j=1$ is one more than the number of $\vec{z}$-variables. 
But the difficulty can be overcome by 
considering the following $n$-fold integral instead of $I_z$,

\begin{equation}
\label{J}
J^{(n)}
=
\frac{\rho^n}{n!}
\oint_{1,1/\rho'}\dd^{n}z
\frac{\e^{\Lambda_{n,0} t}
\Delta_{n}(z)\Delta_n(-z) }
{\prod_{i=1}^{n}(1-\rho'z_i)(z_i-1)^{n}
[\frac{1}{2}(z_i+1)]^m} =: \oint_{1,1/\rho'}\dd^{n}z J^{(n)}(\vec{z}).
\end{equation}

\begin{lemman}
\label{lem:rewritingI_z}
The integral $I_z$ defined in \eqref{def:Iz} can be written into $-I_z=J^{(n)}-1$ with $J^{(n)}$ given by \eqref{J}.
\end{lemman}
\begin{proof}

Calculating the residue of the simple pole at $z_{n}=1/\rho'$ of the $z_n$-integral, we have

\begin{align*}
& \oint_{1,1/\rho'} \dd^{n-1}z  \res_{z_n=1/\rho'}J^{(n)}(\vec{z})\\
=&
-\frac{\e^{-\rho t/2}}{(\rho')^{n-1}} \left(\frac{2( 1- \rho)}{2 - \rho}\right)^m \frac{(-1)^{n-1}}{n!} \oint_{1} \dd^{n-1}z
\frac{\e^{\Lambda_{n-1,0} t}
\Delta_{n-1}(z) \Delta_{n-1}(-z)
\prod_{i=1}^{n-1}(1-\rho'z_i)}
{\prod_{i=1}^{n-1}(z_i-1)^{n}
[\frac{1}{2}(z_i+1)]^m} = - I_z/n.
\end{align*}
Since $\res_{z_n=1/\rho'}J^{(n)}(\vec{z})$ has no poles at $z_j=1/\rho'$, and all the other contributions of poles at $z_j=1/\rho'$ is just $n$ times $\res_{z_n=1/\rho'}J^{(n)}(\vec{z})$, we therefore obtain

\begin{align}
\nonumber
-I_z
=&
J^{(n)}-\frac{\rho^n}{n!}
\oint_{1}\dd^{n}z
\frac{\e^{\Lambda_{n,0} t}
\Delta_{n}(z)\Delta_n(-z) }
{\prod_{i=1}^{n}(1-\rho'z_i)(z_i-1)^{n}
[\frac{1}{2}(z_i+1)]^m}
\\
\nonumber
=&
J^{(n)}-\rho^n
\oint_{1}\dd^{n}z
\frac{\e^{\Lambda_{n,0} t}
\Delta_{n}(z)}
{\prod_{i=1}^{n}(1-\rho'z_i)(z_i-1)^{i}
[\frac{1}{2}(z_i+1)]^m}
\\
=&
J^{(n)}-1 \nonumber
\label{Iz=J-1}.
\end{align}

The second equality follows from anti-symmetrisation with the identity \eqref{Vandermonde_1}. The third line is obtained by evaluating poles at $z_j=1$ sequentially. Starting with the simple pole at $z_1=1$, its residue will decrease the order of other poles at $z_j=1$ by one, and hence all the poles can be evaluated sequentially. \qed 
\end{proof}

The asymptotic behaviour of $I_z$ is now given by the asymptotics of $J^{(n)}$, which now can be analysed via the previous method. Specifically, after rearrangements, we have 

\begin{equation}
\label{Jn_Standard}
J^{(n)}
=
\frac{1}{n!}
\oint_{1,\rho'}\frac{\dd^{n}z}{\prod_{i=1}^n z_i}
\frac{\prod_{1 \leq i\neq j \leq n}(1-z_i/z_j)}
{\prod_{i=1}^{n}{(1- 1/z_i)}^n }
\prod_{i=1}^{n}
\frac{g(z_i,0)}{g(1,0)},
\end{equation}
where
\begin{align*}
g(z,x)
=\frac{z^{-x-n}}{1-\rho'/z}
\left(\frac{z}{1+z}\right)^m \e^{zt/2}.
\end{align*}
The right hand side of \eqref{Jn_Standard} fits the standard form \eqref{Inu} with

\begin{equation*}
\begin{array}{rll}
\nu &=n, \quad &\mu=m, \quad \gamma = t/2,\quad s=0, \\
v_i &=1, & 1\le i \le m, \\
a_k &=1, \quad & 1\leq k\leq n, \\
u_k &= 0 , & 1\le k\le n-1,\quad u_n =\rho'.
\end{array}
\end{equation*}
In this case, Proposition~\ref{prop:fredholm} holds for $c=0$. According to Section \ref{se:tofredholm}, $J^{(n)}$ is thus written as a Fredholm determinant with the kernel,
\begin{equation}
\label{eq:fredholm2}
K(x,y)
=
\sum_{k=0}^{n-1}
\phi_k(x)\psi_k(y),
\end{equation}
where
\begin{align*}
\phi_k(x)
=&
\oint_1 \frac{\dd w}{2 \pi \ii}
\frac{w^k(1-\rho'/w)}{w^x(w-1)^{k+1}}
\left(\frac{w+1}{w}\right)^m\e^{-wt/2},
\\
\psi_k(x)
=&
\oint_{\rho'} \frac{\dd w}{2 \pi \ii}
\frac{w^x(w-1)^{k}}{(1-\rho'/w)w^{k+2}}
\left(\frac{w}{w+1}\right)^m\e^{wt/2}.
\end{align*}

Like before, with the scaling \eqref{nmscale} and setting $k=n-\kappa_1t^{1/2},x=\xi_1t^{1/2}$, we rewrite the integrand into exponents:
\begin{equation}
\label{phiandpsi2}
\begin{split}
\phi_k(x) =& \oint_1 \frac{\dd w}{2\pi\ii}
(w-\rho')
\e^{f(w,t,\xi_1, \kappa_1)+g_4(w)} ,
\\
\psi_k(x) =& \oint_{\rho'}\frac{\dd w}{2\pi\ii}
\frac{1}{w-\rho'}
\e^{ -f(w,t,\xi_1, \kappa_1)+g_5(w)}=\frac{1}{\rho'}\e^{-f(\rho',t,\xi_1, \kappa_1)},
\end{split}
\end{equation}
where $f(w,t,\xi_1, \kappa_1)=g_1(w)t+g_2(w ,\xi_1,\kappa_1 )t^{1/2}+g_3(w)t^{1/3}$ with $g_i(w)$ and $g_2(w,\xi_1,\kappa_1)$ given below, and the pole $w=\rho'$ in $\psi_k(x)$ is of order 1 and hence can be evaluated easily. Therefore we are left with the asymptotics of function $\phi_k(w)$ only, which is given in the following proposition.
\begin{subequations}
\label{g def2}
\begin{align}
g_1(w)=&
\frac{(1+\rho)^2(2-\rho)}{16}\ln\left(\frac{w+1}{w}\right)+
\frac{\rho(3-\rho)^2}{16}\ln\left(\frac{w}{1-w}\right)-\frac{w}{2},
\\
g_2(w,\xi_1,\kappa_1)=&
-\frac{(1+\rho)c_{\rm g} s_{\rm g}}{12(1-\rho)}\ln\left(\frac{w+1}{w}\right)-
\left(\frac{(3-\rho)c_{\rm g} s_{\rm g}}{12(1-\rho)}+\kappa_1\right)
\ln\left(\frac{w}{1-w}\right)-\xi_1\ln(w),
\\
g_3(w)=&
-\frac{(2-\rho)c_2 s_2}{6(1-\rho)}\ln\left(\frac{w+1}{w}\right)-
\frac{\rho c_2 s_2}{6(1-\rho)}\ln\left(\frac{w}{1-w}\right),\\
g_4(w)=&-\ln(w)-\ln(w-1),\\
g_5(w)=&-\ln(w).
\end{align}
\end{subequations}

\begin{prop}
\label{uniform convergence phi2}
{\rm \textbf{(Uniform convergence of $\boldsymbol{\phi}$ on a bounded set)}}
Let $m,n$ be scaled as \eqref{nmscale}. Then for some constant $c_0\in\mathbb{R}$, and $\lambda_1>0,\lambda_2>0$, the function $\phi_k(x)$ defined in \eqref{phiandpsi2} converges uniformly in a bounded set of $\xi,\kappa$ to
\begin{equation*}
\label{phi2 limit}
\lim_{t \rightarrow \infty}
tc(t)
\phi_{n-\lambda_1 t^{1/2}\kappa}(\lambda_2 t^{1/2}\xi)
=c_0
(s_{\rm g}+\kappa+\xi)\e^{- \frac{1}{2} (s_{\rm g}+\kappa+\xi)^2},
\end{equation*}
where $c(t)$ is some function of $t$ that may or may not be bounded as $t\rightarrow\infty$. Specifically,
\begin{equation}\label{ConstantsGaussian}
\begin{split}
c_0=&\frac{8(2-\rho)}{9\sqrt{2\pi}(1-\rho)^2},\\
\lambda_1=&c_{\rm g}/[2(2-\rho)],
\\
\lambda_2=&c_{\rm g}/[2 \rho(2-\rho)],
\\
c(t)=&\exp\big(-
g_1(\rho')t - g_2(\rho')t^{1/2} - g_3(\rho')t^{1/3}\big),
\end{split}
\end{equation}
with $g_i(w)$ given below in \eqref{g def2}, and $c_2,c_{\rm g}$ given in \eqref{c2cg}.
\end{prop}

\begin{proof}
See Appendix~\ref{ap:uniform convergence phi2}.\qed
\end{proof}
Following the same arguments as in \cite{borodin2008transition,borodin2007fluctuationsintasep,BFPS2007}, which is also similar to the method used in the proof of Proposition~\ref{exp_bound}, we have a corresponding estimate of the function $\phi_k(w)$ defined in \eqref{phiandpsi2}.

\medskip
\begin{prop}
\label{bound phi2}
{\rm \textbf{(Estimate of function $\boldsymbol{\phi}$)} }
Consider the function $\phi_k(x)$ defined in \eqref{phiandpsi2}, with $n,m$ scaled as \eqref{nmscale}, and $f(w,t, \xi_1, \kappa_1)=g_1(w)t+g_2(w,\xi_1,\kappa_1)t^{1/2}+g_3(w)t^{1/3}$ defined in \eqref{g def2}. Then for large enough $t$ and $L$ large enough but independent of $t$,
\begin{equation*}
  \left|
\frac{t \e^{-f(\rho',t ,\xi_1, \kappa_1 )}}{\rho(1-\rho)}\phi_{n-\lambda_1 t^{1/2}\kappa}(\lambda_2 t^{1/2}\xi)
\right|
\leq \e^{-(\kappa+\xi)},
\end{equation*}
for $\kappa,\xi\geq -L$ and $\kappa+\xi \geq L$.
\end{prop}

\begin{figure}[h]
\begin{center}
\begin{tikzpicture}[scale=2.5]
%jst
\draw[->,thick] (-0.3,0) -- (2,0);
\draw[->,thick] (0,-0.8) -- (0,0.8);

\draw[fill=black] (0,0) circle (0.03);
\draw[fill=black] (1,0) circle (0.03);
\draw[fill=black] (0.35,0) circle (0.03);

\node at (0.05,-0.12) {$0$};
\node at (1,-0.12) {$1$};

\draw[thick] (0.35,-0.375278) arc (-150:150:0.750555);
\draw[thick] (0.35,-0.375278) arc (-150:150:0.750555);
\draw[thick] (0.35,0.1) arc (90:-90:0.1);

\node at (0.27,-0.12) {$\rho'$};

\coordinate (A) at (0.35,-0.375278);
\coordinate (a) at (0.35,-0.1);
\coordinate (B) at (1,0);
\coordinate (C) at (0.35,0.375278);
\coordinate (c) at (0.35,0.1);
\draw[thick,dashed] (A) -- (B) -- (C);
\draw[thick] (A) -- (a);
\draw[thick] (C) -- (c);

\end{tikzpicture}
\end{center}
\caption{Deformed steepest descent contour of integration in $\phi_k(x)$.}
\label{fig:contour gaussian 1 deform}
\end{figure}
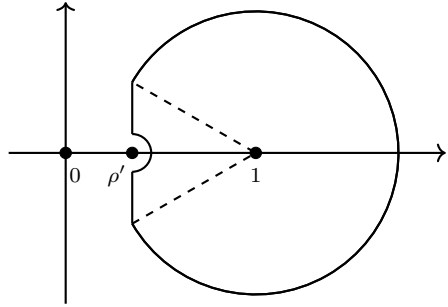

With a deformed contour given in Fig. \ref{fig:contour gaussian 1 deform}, the proof for this proposition follows exactly the same method as in the proof of Proposition~\ref{exp_bound}. Therefore we will not repeat it here.

Consequently, these two Proposition \ref{uniform convergence phi2} and \ref{bound phi2}, together with the fact that
\begin{equation}
\label{uniform convergence psi2}
  \e^{f(\rho',t,\xi_1, \kappa_1)}\psi_{n-\lambda_1 t^{1/2}\kappa}(\lambda_2 t^{1/2}\xi)=\frac{1}{1-\rho},
\end{equation}
give the following theorem.

\begin{theorem}
Consider the kernel defined in \eqref{eq:fredholm2}, and define the rescaled kernel
\begin{equation*}
 \bar{K}_t( \xi, \zeta) =\rho'^{ \lambda_2 (\xi-\zeta)t^{1/2} }
  \lambda_2 t^{1/2}K(\lambda_2 t^{1/2} \xi,\lambda_2 t^{1/2} \zeta),
\end{equation*}
with $\lambda_2$ given in \eqref{ConstantsGaussian}. Then we have

\begin{enumerate}[(i)]
\item \label{RescaledFredholm} $\displaystyle \lim_{t \to \infty} \det(1-  \bar{K}_t)_{L^2(0,\infty)}= \lim_{t \to \infty} \det(1-\bar{K}_t)_{\ell^2(\mathbb{N}/(\lambda_2 t^{1/2}))}$.

Here note that the operator $\bar{K}_t$ on the right hand side is regarded as acting on $\ell^2(\mathbb{N}/(\lambda_2 t^{1/2}))$ while the one on the left hand side on  $L^2(0,\infty)$.
\item
  For any fixed $L>0$
  \begin{equation*}
    \lim_{t\rightarrow \infty} \bar{K}_t( \xi, \zeta)=
    \frac{\e^{-(s_{\rm g}+\xi)^2 /2}}{\sqrt{2 \pi}},
  \end{equation*}
  uniformly on $(\xi,\zeta)\in[-L,L]^2$.
  \item For any fixed $L>0$ and $t$ large enough,
\begin{equation*}
  | \bar{K}_t ( \xi, \zeta)| \leq C \e^{
  -\max\{0,\xi\}},
\end{equation*}
for some constant $C>0$ and $\xi,\zeta\geq -L$.
\end{enumerate}

\end{theorem}

\begin{proof}
Obviously, (i) can be easily proved by the change of variable $(x,y)\rightarrow ( \lambda_2 t^{1/2} \xi, \lambda_2 t^{1/2} \zeta)$ in the Fredholm kernel $K(x,y)$ and then conjugating the kernel with $\rho'^{ \lambda_2 (\xi-\zeta)t^{1/2} }$. For (ii), we have
\begin{align*}
\bar{K}_t (\xi,\zeta)
=&
\rho'^{ \lambda_2 (\xi-\zeta)t^{1/2} }
\lambda_2 t^{1/2}
\sum_{k=0}^{n-1}
\phi_k(\lambda_2 t^{1/2} \xi)
\psi_k(\lambda_2 t^{1/2} \zeta)
\\
=&\rho'^{ \lambda_2 (\xi-\zeta)t^{1/2} }
\lambda_1\lambda_2 t
\left(
\int^{n t^{-1/2} / \lambda_1}_{L'}
+\int^{L'}_{t^{-1/2} / \lambda_1}
\right)
\phi_{n-\lambda_1 t^{1/2} \kappa}(\lambda_2 t^{1/2} \xi)
\psi_{n-\lambda_1 t^{1/2} \kappa}(\lambda_2 t^{1/2} \zeta)
\dd \kappa.
\end{align*}
By Proposition \ref{bound phi2}, one can see that the integrand of the first term $\rho'^{ \lambda_2 (\xi-\zeta)t^{1/2} }
\lambda_1\lambda_2 t\phi_{n-\lambda_1 t^{1/2} \kappa}(\lambda_2 t^{1/2} \xi) \times $ $
\psi_{n-\lambda_1 t^{1/2} \kappa}(\lambda_2 t^{1/2} \zeta) \leq c \e^{-\kappa}$ for some constant $c$. Therefore the first term 
\[
\rho'^{ \lambda_2 (\xi-\zeta)t^{1/2} }
\lambda_1\lambda_2 t
\int_{L'}^{n t^{-1/2}/\lambda_1}
\phi_{n-\lambda_1 t^{1/2} \kappa}(\lambda_2 t^{1/2} \xi)
\psi_{n-\lambda_1 t^{1/2} \kappa}(\lambda_2 t^{1/2} \zeta) \leq c \e^{-2L'}/2
\]
vanishes as $L'\rightarrow \i$.

For the second sum, we have, as a result of Propositions~\ref{uniform convergence phi2} and \eqref{uniform convergence psi2},
\begin{align*}
&\lim_{t\rightarrow \infty}
\rho'^{ \lambda_2 (\xi-\zeta)t^{1/2} }
\lambda_1 \lambda_2 t\int_{t^{-1/2} / \lambda_1}^{L'}
\phi_{m-\lambda_1 t^{1/2}\kappa}(\lambda_2 t^{1/2}\xi)
\psi_{m-\lambda_1 t^{1/2}\kappa}(\lambda_2 t^{1/23}\zeta)
\dd \kappa \\
=&
\frac{8(2-\rho)}{9\sqrt{ 2 \pi}(1-\rho)^2}\frac{1}{1-\rho}
\lambda_1 \lambda_2
\int^{L'}_{0}
(s_{\rm g}+\kappa+\xi)
\e^{-(s_{\rm g}+\kappa+\xi)^2 /2}
\dd \kappa
\\
\underset{ L'\rightarrow\infty }{\longrightarrow}&
\frac{1}{\sqrt{2 \pi}}
\int^{\infty}_{0}
(s_{\rm g}+\kappa+\xi)
\e^{-(s_{\rm g}+\kappa+\xi)^2 /2}
\dd \kappa 
\\
=&
\frac{ \e^{-(s_{\rm g}+\xi)^2 / 2} }{\sqrt{2 \pi}}  
\end{align*}
uniformly on $(\xi,\zeta)\in[-L,L]^2$ for any fixed $L>0$.

The bound (iii) can be easily seen from Proposition \ref{bound phi2} and (ii). \qed
\end{proof}

The above results with the Lemma \ref{lem:kernel lim} allow us to take the limit of the Fredholm determinant into the kernel, i.e.,
\begin{multline*}
\lim_{t\rightarrow\infty}\det(1-\bar{K}_t)_{\ell^2(\mathbb{N}/(\lambda_2 t^{1/2}))}
=\det(1-\lim_{t\rightarrow\infty} \bar{K}_t)_{L^2(0,\infty)}
=1-\int_{0}^{\infty}\frac{\e^{-(s_{\rm g}+\xi)^2/2}}{\sqrt{2\pi}}\dd \xi\\=
\int^{0}_{-\infty}\frac{\e^{-(s_{\rm g}+\xi)^2/2}}{\sqrt{2\pi}}\dd \xi
=
\int^{s_{\rm g}}_{-\infty}\frac{\e^{-\xi^2/2}}{\sqrt{2\pi}}\dd \xi
=F_G(s_{\rm g}),
\end{multline*}
which implies that
\begin{equation}\label{j limit}
\lim_{t\rightarrow \infty} J^{(n)} =F_G(s_{\rm g}) .
\end{equation}

Putting all things in Subsection~\ref{subsec:Ev_I21} together, we have the following theorem.

\begin{theorem}
\label{thm_I2(1) limit}
With the scaling  \eqref{nmscale}, the long time limit of the term $\mathcal{I}^{(1)}_2$ defined in \eqref{def:I2(1)} is given by 
\begin{equation}
\label{I2(1) limit}
\lim_{t\rightarrow\infty}\mathcal{I}^{(1)}_2 = \left[ 1 - F_G(s_{\rm g}) \right] F_2(s_2),
\end{equation}
where $s_{\rm g}$ and $s_{2}$ are given in \eqref{s+-}.

\end{theorem}

\begin{proof}
Applying Lemma~\ref{lem:rewritingI_z} to the right hand side of \eqref{def:I2(1)}, we obtain

\begin{equation*}
    \mathcal{I}^{(1)}_2 = (1 - J^{(n)}) \times \det(1 - \bar{K}^c_{\mathrm{W},t})_{\ell^2(\mathbb{N}/\lambda_c t^{1/3})}.
\end{equation*}
Since the long time limit of $\det(1-\bar{K}^c_{\mathrm{W},t})_{\ell^2(\mathbb{N}/\lambda_c t^{1/3})}$ and $J^{(n)}$ with the scaling \eqref{nmscale} are given by \eqref{detKcW limit} and \eqref{j limit}, respectively, we arrive at the result \eqref{I2(1) limit}. \qed

\end{proof}

\subsection{Evaluation of $\mathcal{I}^{(2)}_2$}
\label{subsec:Ev_I22}

We will show that $\mathcal{I}^{(2)}_2$ tends to zero at long time limits with the scaling \eqref{nmscale}.
As a preparation for the asymptotic analyses, we define the rescaled functions as

\begin{align}
\bar{A}^c_{t,p}(\xi) & = \e^{ - f^c(w_c , \lambda_c t^{1/3} \xi) }  A^c_p(\lambda_c t^{1/3} \xi), \label{def:Acresc} \\
\bar{B}^c_{t}(\xi) & = \lambda_c t^{1/3} \e^{f^c(w_c, \lambda_c t^{1/3} \xi)} B^c(\lambda_c t^{1/3} \xi), \label{def:Bcresc}
\end{align}
where $A^c_p(x)$ and $B^c(y)$ defined in \eqref{def:Acj} and \eqref{def:Bc}, respectively, with $f^c(z,x)$ given in \eqref{def:fc}.
Using these functions, we rewrite $\mathcal{I}_2^{(2)}$ defined in  \eqref{def:I2(2)} as

\begin{multline}
\mathcal{I}^{(2)}_2 = \sum_{k=1}^{m} \frac{(-1)^k}{ k! }
\sum_{x_1 = 1}^{\infty}  \!\! \dots \!\!\! \sum_{x_{k} = 1}^{\infty}
\sum_{r,\ell =1}^{k}{ {(-1)}^{r+\ell} {\left( \frac{1}{\lambda_c t^{1/3}} \right)}^k } \times \\ \det\left[  \bar{K}^c_{\mathrm{W},t}(x'_i , x'_j) \right]_{ \substack{ 1 \leq i , j \leq k \\  i \neq \ell , j \neq r } }  \left[ \sum_{p=1}^{n-1}{  \oint_1 \dd^{n-1} z\,  L(\vec{z}) \prod_{q=1}^{p}{ {(z_q - 1)} \bar{A}^c_{t,p}(x'_\ell) }} \right]  \bar{B}^c_t(x'_r)    . \label{def:I2(2)_rescaled} 
\end{multline}

\begin{prop}
\label{I2(2)_up_bound_AB}
For any fixed $\rho \in (0,1)$ and large enough $t$, the absolute value of the summation with respect to $p$ in \eqref{def:I2(2)_rescaled} satisfies the following bound

\begin{equation}
\label{eq:up_bound_AB}
\left| \sum_{p=1}^{n-1}{  \oint_1 \dd^{n-1} z\,  L(\vec{z}) \prod_{q=1}^{p}{ {(z_q - 1)} \bar{A}^c_{t,p}(\xi) }}   \bar{B}^c_t(\zeta) \right| \leq  D_1 t^{-1/3} \e^{-(\xi+\zeta)},
\end{equation}
for $(\xi,\zeta) \in [0 , \infty)^2$. $D_1$ is some positive constant where $L(\vec{z})$,  $\bar{A}^c_{t,p}(\xi)$ and $\bar{B}^c_t(\zeta)$ are defined in \eqref{def:L(z)}, \eqref{def:Acresc} and \eqref{def:Bcresc}, respectively.
\end{prop}

\begin{proof}
The proof consists of the four steps. As the first step, we perform the sum over $p \in [1,n-1]$ and derive the formula which is useful for the asymptotic analysis.
In fact, it turns out that the integrals obtained as the result of the first step share similar integrands to the ones treated in Propositions~\ref{uniform convergence kernelc}, \ref{exp_bound}, \ref{uniform convergence phi2} and \ref{bound phi2}.
In the second step, we show the asymptotic analysis of the Gaussian part \eqref{Gaussian_part} for which we can apply a similar method to Proposition~\ref{uniform convergence phi2} and \ref{bound phi2}.
In the third step, we investigate the asymptotic behaviours of the Airy parts \eqref{Airy_parts} in the same fashion as Propositions~\ref{uniform convergence kernelc} and \ref{exp_bound}.
In the final step, combining the results of the previous steps, we arrive at the conclusion \eqref{eq:up_bound_AB}.

\paragraph{1, Rewrite of $\sum_{p=1}^{n-1} \oint_1 \dd^{n-1} z L(\vec{z}) \prod_{q = 1}^{p}{(z_q - 1)} \bar{A}^c_{t,p}(x)$}

We will rewrite the sum over $p \in [1,n-1]$ into a suitable formula for the asymptotic analysis.
Recall the function $L(\vec{z})$  defined in \eqref{def:L(z)}, $\vec{z}$-integrals are written as

\begin{multline*}
\oint_1 \dd^{n-1} z\,  L(\vec{z}) \prod_{q=1}^{p}(z_q-1) \\
= {(-1)}^{n-1}
\frac{\e^{-\rho t/2}}{\rho'^{n-1}} \left(\frac{2(1-\rho)}{2-\rho}\right)^m
\oint_{1}\dd^{n-1}z
\frac{\e^{\Lambda_{n-1,0} t}
\Delta_{n-1}(-z)
\prod_{i=1}^{n-1} (1-\rho'z_i) }
{\prod_{i=1}^{p}(z_i-1)^{i}
\prod_{i=p+1}^{n-1}(z_i-1)^{i+1}
[\frac{1}{2}(z_i+1)]^m}.
\end{multline*}
On the right hand side, all poles at $z_i = 1$ can be easily evaluated from $i=1$ to $i=p$ in the same way as the proof of Lemma~\ref{lem:prob2term}.
Changing variables $z_{j+p} \rightarrow z_j$ for any $j \in [1,n-p-1]$ and symmetrising the $\vec{z}$-integrands by making use of the identity \eqref{Vandermonde_1}, we obtain

\begin{equation*}
\begin{split}
& \oint_1 \dd^{n-1} z\,  L(\vec{z}) \prod_{q=1}^{p}(z_q-1) \\
= & {(-1)}^{n-1}
\frac{\e^{-\rho t/2}}{\rho'^{n-1}} \left(\frac{2(1-\rho)}{2-\rho}\right)^m \rho^p
\oint_{1}\dd^{n-p-1}z
\frac{\e^{\Lambda_{n-p-1,0} t}
\Delta_{n-p-1}(-z)
\prod_{i=1}^{n-p-1}   (1-\rho'z_i)  }
{\prod_{i=1}^{n-p-1}(z_i-1)^{i+1}
[\frac{1}{2}(z_i+1)]^m} \\
= & {\left( \frac{- \rho}{\rho'} \right)}^p {(-1)}^{n-p-1}
\frac{\e^{-\rho t/2}}{\rho'^{n-p-1}} \left(\frac{2(1-\rho)}{2-\rho}\right)^m \frac{1}{(n-p-1)! } \times \\
& \qquad \qquad \qquad \qquad \qquad \qquad \oint_{1}\dd^{n-p-1}z
\frac{\e^{\Lambda_{n-p-1,0} t}
\Delta_{n-p-1}(z) \Delta_{n-p-1}(-z)
\prod_{i=1}^{n-p-1}  (1-\rho'z_i)  }
{\prod_{i=1}^{n-p-1}(z_i-1)^{n-p}
[\frac{1}{2}(z_i+1)]^m} .
\end{split}
\end{equation*}
The formula in the right hand side except for the factor ${(-\rho / \rho')}^p$ coincides with $I_z$, defined in \eqref{def:Iz}, in which $n$ is replaced with $n-p$. 
Therefore, $\vec{z}$-integrals can be written as

\begin{equation}
\label{z-integral_part}
 \oint_1 \dd^{n-1} z\,  L(\vec{z}) \prod_{q=1}^{p}(z_q-1) = - {\left( \frac{- \rho}{\rho'} \right)}^p \left( J^{(n-p)} - 1 \right) ,
\end{equation}
where $J^{(n)}$ is defined in \eqref{J}.
Following the arguments in Section~\ref{se:tofredholm}, $J^{(n-p)}$ is thus written as a Fredholm determinant with the kernel

\begin{align*}
K(x,y)
=&\oint_1 \frac{\dd w}{2 \pi \ii} {\left( \frac{w-1}{w} \right)}^p(1-\rho'/w)\e^{h(w,x)}
\oint_{\rho'} \frac{\dd z}{2 \pi \ii}
\e^{-h(z,y)}\frac{z^p}{z(w-z)(1-\rho'/z)(z-1)^p}
\\
=&
\e^{-h(\rho',y)} {\left( \frac{-\rho'}{\rho} \right)}^p 
\oint_1 \frac{\dd w}{2 \pi \ii w} {\left( \frac{w-1}{w} \right)}^p  \e^{h(w,x)},
\end{align*}
where
$
\e^{h(w,x)}:=w^{-x}
\left(\frac{w}{w-1}\right)^{n}
\left(\frac{w+1}{w}\right)^{m}
\e^{-wt/2}
$.
Since the operator with the above kernel is a rank one operator, the third and subsequent terms of its Fredholm determinant vanish.
Therefore, we obtain

\begin{equation}
\label{J_Fredholm}
J^{(n-p)}=
1- \sum_{x=1}^{\infty} K(x,x)
=1-\left(\frac{\rho'}{-\rho} \right)^{p} \sum_{x=1}^{\infty}
\int_1\frac{\dd w}{2 \pi \ii w}{\left( \frac{w-1}{w} \right)}^p   \e^{h(w,x)-h(\rho' , x)}.
\end{equation}
Since the contour includes only the pole at $w = 1$, it is allowed to choose the contour on which $|w| > \rho'$ always holds. Therefore, we can perform the sum over $x \in \mathbb{N}$ in the right hand side of \eqref{J_Fredholm}, and we obtain

\begin{equation}
\label{J_Fredholm_xsum}
{\left( \frac{-\rho}{\rho'} \right)}^p (J^{(n-p)} - 1) = \oint_{1} \frac{\dd w}{2 \pi \ii} \frac{\rho'}{w(w - \rho')} {\left( \frac{w - 1}{w} \right)}^p \e^{h(w , 0) - h(\rho',0)}.
\end{equation}
Note that the contour does not include the pole at $w = \rho'$.
Recalling the definition of $A^c_p(x)$ given in \eqref{def:Acj} and combine \eqref{z-integral_part} and \eqref{J_Fredholm_xsum}, we have 

\begin{multline*}
\label{p_sum_part_original}
\sum_{p=1}^{n-1} \oint_1 \dd^{n-1} z\,  L(\vec{z}) \prod_{q=1}^{p}(z_q-1) A^c_p(x)
\\ = \sum_{p=1}^{n-1}  \oint_{1} \frac{\dd w}{2 \pi \ii} \frac{\rho'}{w - \rho'} {\left( \frac{w-1}{w} \right)}^p  \e^{h(w,0)- h(\rho', 0) + g_J(w)} \oint_1 \frac{\dd z}{2 \pi \ii} \e^{ f^c(z,x)- f^c(w_c,x) + g_{\phi}(z)} {\left( \frac{z}{1 + z} \right)}^p 
\end{multline*}
with $g_J(w) := - \ln(w)$, $g_\phi(z):=\ln{(z+\rho')/(z(z+1))}$ and $f^c(z,x)$ defined in \eqref{def:fc}. 
In the right hand side, we can perform the sum over $p \in [1,n-1]$ as

\begin{equation}
\label{p_sum}
\sum_{p=1}^{n-1}{ {\left[ \frac{z(w-1)}{w(z+1)} \right]}^p } = \frac{z(w-1)}{z+w} \left[ 1 - {\left(  \frac{z(w-1)}{w(z+1)} \right)}^{n-1} \right],
\end{equation}
and the factor ${(w - 1)}^n$ in the second term removes the pole of the $n$\textsuperscript{th} order at $w = 1$ of the integrand $\e^{h(w,0)}$, i.e., the second term of \eqref{p_sum} does not contribute.
If we choose the contours such that $\Re(z + w) >0$ holds for any $w$ and $z$ on the contours, the factor $1/(z + w)$ can be written as $1/(z + w) = \int_{0}^{ \infty } \dd k \e^{- k (z + w)}$. Therefore, we obtain

\begin{multline}
\label{p_sum_part_taken_unscaled}
\sum_{p=1}^{n-1} \oint_1 \dd^{n-1} z\,  L(\vec{z}) \prod_{q=1}^{p}(z_q-1) A^c_p(x) \\
= \oint_{1} \frac{\dd w}{2 \pi \ii} \frac{\rho'}{w - \rho'}  \e^{ h(w,0)- h(\rho',0) + g'_J(w)} \oint_1 \frac{\dd z}{2 \pi \ii}  \e^{f^c(z,x)- f^c(w_c,x) + g'_{\phi}(z)} \int_{0}^{\infty} \dd k \e^{-k(z + w)}
\end{multline}
with $g'_J(w) : = g_J(w) + \ln(w-1)$ and $g'_\phi(z) : = g_\phi(z) + \ln(z)$.

In the next step, we will investigate the asymptotic behaviour in the long time limit with the scaling \eqref{nmscale}.
To perform the asymptotic analyses, we scale $(x,y)= ( \lambda_c t^{1/3} \xi, \lambda_c t^{1/3} \zeta )$ and $k = t^{1/3} \kappa$, and choose the contour of the integral with respect to $z$ and $w$ in \eqref{p_sum_part_taken_unscaled} as $\Gamma'$ defined in \eqref{Gamma_prime} and $\Theta'$ given in Fig. \ref{fig:contour gaussian 1 deform}, respectively.
$\Theta'$ can be expressed as

\begin{equation*}
\label{Contour_Theta_prime}
\Theta' =  \left\{ w \in \Theta \mid |w - \rho'| > \epsilon \right\} \cup \left\{ w \in \mathbb{C} \mid w = \rho' + \epsilon \e^{- \ii s } , s \in [- \pi/2 , \pi/2] \right\}
\end{equation*}
with $\epsilon \ll t^{-1/6}$ where $\Theta$ is defined in \eqref{gamma2}. It is easy to see that $\Re(z + w) > 0$ holds for any $(z , w) \in \Gamma' \times \Theta'$.
Rewriting \eqref{p_sum_part_taken_unscaled} in terms of $\bar{A}^c_{t,p}(\xi)$ and putting the integral with respect to $\kappa$ outside ones with respect to $z$ and $w$, we get 

\begin{multline}
\label{p_sum_part_taken}
\sum_{p=1}^{n-1} \oint_1 \dd^{n-1} z\,  L(\vec{z}) \prod_{q=1}^{p}(z_q-1) \bar{A}^c_{t,p}(\xi) \\
=  t^{1/3} \int_{0}^{\infty} \dd \kappa \e^{-3 \rho' t^{1/3} \kappa/2 } \oint_{\Theta'} \frac{\dd w}{2 \pi \ii} \frac{\rho'}{w - \rho'}  \e^{\bar{h}(w,t)-\bar{h}(\rho',t) - \kappa (w - \rho') t^{1/3} + g'_J(w)} \times \\ \oint_{\Gamma'} \frac{\dd z}{2 \pi \ii}  \e^{\bar{f}(z,t)-\bar{f}(w_c,t) - \kappa (z - w_c) t^{1/3} + g'_{\phi}(z)}  {\left( \frac{w_c + c}{z + c} \right)}^{\xi \lambda_c t^{1/3}},
\end{multline}
where $\bar{f}(z , t)$ is defined above \eqref{resc K suitable for asymptotics} and $\bar{h}(w , t)$ is defined as
\begin{equation*}
\bar{h}(w, t) := g_1(w) t + g_2(w , 0 ,0) t^{1/2} + g_3(w) t^{1/3}
\end{equation*}
with $g_i(w)$ and $g_2(w, \xi , \kappa)$ given in \eqref{g def2}.
Similarly, the rescaled function $\bar{B}^c_t(\zeta)$ can be written as
\begin{equation}
\label{Bc_part}
\bar{B}^c_t(\zeta) = \lambda_c t^{1/3} \oint_{\Sigma'} \frac{\dd w}{2 \pi \ii} \e^{-\bar{f}(w,t) + \bar{f}(w_c,t) + g_{\psi}(w)} {\left( \frac{w+c}{w_c + c} \right)}^{\zeta \lambda_c t^{1/3}}
\end{equation}
with $g_{\psi}(w) = - \ln{(w + \rho')(w + c)}$ and $\Sigma'$ defined in \eqref{Sigma_prime}.
It follows from \eqref{p_sum_part_taken} and \eqref{Bc_part} that the function which is to be evaluated is written as

\begin{equation*}
\sum_{p=1}^{n-1} \oint_1 \dd^{n-1} z\,  L(\vec{z}) \prod_{q=1}^{p}(z_q-1) \bar{A}^c_{t,p}(\xi) \bar{B}^c_t(\zeta) = \int_{0}^{\infty} \dd \kappa \e^{-3 \rho' t^{1/3} \kappa / 2} G_{\rm g}(\kappa, s_{\rm g}) G_2(\kappa, s_2 , \xi , \zeta) ,
\end{equation*}
where the functions $G_{\rm g}$ and $G_2$ are defined as
 
\begin{equation}
\label{Gaussian_part}
G_{\rm g}(\kappa, s_{\rm g}) := \oint_{\Theta'} \frac{\dd w}{2 \pi \ii} \frac{\rho'}{w - \rho'}  \e^{\bar{h}(w,t)-\bar{h}(\rho',t) - \kappa (w - \rho') t^{1/3} + g'_J(w)},
\end{equation}
\begin{multline}
\label{Airy_parts}
G_2(\kappa, s_2, \xi , \zeta) := t^{1/3} \oint_{\Gamma'}  \frac{\dd z}{2 \pi \ii}  \e^{\bar{f}(z,t)-\bar{f}(w_c,t) - \kappa (z - w_c) t^{1/3} + g'_{\phi}(z)}  {\left( \frac{w_c + c}{z + c} \right)}^{\xi \lambda_c t^{1/3}} \times \\ \lambda_c t^{1/3} \oint_{\Sigma'} \frac{\dd w}{2 \pi \ii} \e^{-\bar{f}(w,t) + \bar{f}(w_c,t) + g_{\psi}(w)} {\left( \frac{w+c}{w_c + c} \right)}^{\zeta \lambda_c t^{1/3}}.
\end{multline}
Note that the two integrals in $G_2(\kappa, s_2 , \xi, \zeta)$ are independent but we denote it as a single function so as to compare it with $\bar{K}^c_{t}(\xi, \zeta)$ given by \eqref{RescKernelc}.
We observe that the integral in  $G_{\rm g}(\kappa, s_{\rm g})$ has a similar form to the product of the functions $\phi_k$ and $\psi_k$ defined in \eqref{phiandpsi2}, which converges to the Gaussian function.
It is easy to see that the integrals in $G_2(\kappa, s_2 , \xi, \zeta)$ share similar integrands to those with respect to $z$ and $w$ in \eqref{KernelI1withc} except for $1/(w - z)$.

\paragraph{2, Estimate of $G_{\rm g}(\kappa, s_{\rm g})$}

By using a steepest descent analysis, in a similar way as in Propositions~\ref{uniform convergence phi2} and \ref{bound phi2}, we can show that

\begin{equation}
\label{Gaussian_part_Bound}
\left| G_{\rm g}(\kappa , s_{\rm g}) \right| \leq C_1
\end{equation}
holds for $t$ large enough with some positive constant $C_1$ independent of $\kappa$.
Since $\bar{h}(w , t)$ has a saddle point at $w = \rho'$ as shown in Appendix~\ref{ap:uniform convergence phi2}, we divide the contour $\Theta'$ into $\Theta^\delta := \left\{ w \in \Theta' \mid |w - \rho'| \leq \delta \right\}$ and $\Theta' \backslash \Theta^\delta$ with $\delta = t^{-1/6}$.

First, we focus on the main contribution from $\Theta^\delta$.
Since the Taylor expansion of $\bar{h}(w,t)$ around $w = \rho'$ is given by \eqref{g taylor expand2}, we obtain

\begin{multline}
\label{Gaussian_part_divide_1}
\int_{\Theta^\delta} \frac{\dd w}{2 \pi \ii} \frac{\rho'}{w - \rho'}  \e^{\bar{h}(w,t)-\bar{h}(\rho',t) - \kappa (w - \rho') t^{1/3} + g'_J(w)} \\
= \e^{- \frac{1}{2} {(s_{\rm g} + \rho' \kappa t^{-1/6} / \lambda_2)}^2} \int_{ \theta^{\delta t^{1/2}} } \frac{\dd v}{2 \pi \ii} \left( \frac{-\rho}{v} \right) \e^{ - \frac{1}{4} {\left[ v + \sqrt{2} \ii (s_{\rm g} +\rho' \kappa t^{-1/6} / \lambda_2  )  \right]}^2 } + R_v ,
\end{multline}
where the remainder $R_v$ is given by

\begin{equation*}
R_v = \e^{- \frac{1}{2} {(s_{\rm g} + \rho' \kappa t^{-1/6} / \lambda_2)}^2} \int_{ \theta^{\delta t^{1/2}} } \frac{\dd v}{2 \pi \ii} \left( \frac{- \rho}{v} \right) \e^{ - \frac{1}{4} {\left[ v + \sqrt{2} \ii (s_{\rm g} +\rho' \kappa t^{-1/6} / \lambda_2  )  \right]}^2  } \left( \e^{ \mathcal{O}(v^3 t^{-1/2}, v^2 t^{-1/6}, v t^{-1/2}) } - 1 \right)
\end{equation*}
with $\theta^{\delta t^{1/2}} := \{ v \in \mathbb{C} \mid \rho' + \ii v \rho' / ( \sqrt{2 t} \lambda_2) \in \Theta^\delta \}$ and $\lambda_2$ defined in \eqref{ConstantsGaussian}. In the same fashion as Appendix~\ref{Unif_conv_Kc}, we can show

\begin{equation*}
\label{Upper_bound_vintegral}
|R_v | \leq c_1 t^{-1/6} ,
\end{equation*}
where $c_1$ is some positive constant independent of $\kappa$.
Moreover, the first term of the right hand side of \eqref{Gaussian_part_divide_1} is rewritten as

\begin{multline}
\label{Gaussian_part_divide_2}
\e^{- \frac{1}{2} {(s_{\rm g} + \rho' \kappa t^{-1/6} / \lambda_2)}^2} \int_{ \theta^{\delta t^{1/2}} } \frac{\dd v}{2 \pi \ii}\left( \frac{-\rho}{v} \right) \e^{ - \frac{1}{4} {\left[ v + \sqrt{2} \ii (s_{\rm g} +\rho' \kappa t^{-1/6} / \lambda_2  )  \right]}^2 } \\ =  \e^{- \frac{1}{2} {(s_{\rm g} + \rho' \kappa t^{-1/6} / \lambda_2)}^2}   \int_{- \sqrt{2} \lambda_2 \delta t^{1/2} /\rho' }^{ \sqrt{2} \lambda_2 \delta t^{1/2} / \rho' } \frac{\dd u}{2 \pi \ii} \frac{\rho}{u - \sqrt{2} \ii \beta } \e^{ - \frac{1}{4} {\left[ u + \sqrt{2} \ii (s_{\rm g} + \rho' \kappa t^{-1/6} / \lambda_2 - \beta) \right]}^2 } + R_z
\end{multline}
for some $\beta > 0$, and $R_z$ satisfies
\begin{equation*}
\label{Upper_bound_zintegral}
| R_z | \leq c_2 \e^{- a_1 \delta^2 t }
\end{equation*}
with some positive constant $c_2$ independent of $\kappa$ and $a_1 := {\lambda_2}^2 / (2 {\rho'}^2) $.
Incidentally, the equality \eqref{Gaussian_part_divide_2} can be derived by considering the contour integration with the contour shown in Fig.~\ref{fig:contour gaussian pv}, and $R_z$ is contribution from paths $C_{+}$ and $C_{-}$.

Second, we focus on the contribution from $\Theta' \backslash \Theta^\delta$.
Obviously, we have

\begin{multline}
\label{Gaussian_part_rest_1}
\left| \int_{\Theta' \backslash \Theta^\delta} \frac{\dd w}{2 \pi \ii} \frac{\rho'}{w - \rho'}  \e^{\bar{h}(w,t)-\bar{h}(\rho',t) - \kappa (w - \rho') t^{1/3} + g'_J(w)} \right| \\ \leq \int_{\Theta' \backslash \Theta^\delta} \left| \frac{\dd w}{2 \pi \ii} \right| \left|  \frac{\rho'}{w - \rho'}  \e^{\bar{h}(w,t)-\bar{h}(\rho',t) - \kappa (w - \rho') t^{1/3} + g'_J(w)} \right|.
\end{multline}
In the integrand, $|\e^{- \kappa (w - \rho') t^{1/3}}| \leq 1$ holds because $\Re(w - \rho') \geq 0$ holds for any $w \in \Theta'$, hence we can evaluate the exponent of the right hand side of \eqref{Gaussian_part_rest_1} in the same way as Proposition~\ref{uniform convergence phi2}. 
That is to say, $g_1(w)$ given in \eqref{g def2} takes the maximum value at $w = \rho' \pm \ii \delta$ along the contour $\Theta' \backslash \Theta^\delta$, thus 

\begin{equation*}
\Re(g_1(w) - g_1(\rho')) = \Re(- a_1 \delta^2 + \mathcal{O}( \delta^3)) \leq - a_1 \delta^2 / 2
\end{equation*}
holds for any $w \in \Theta' \backslash \Theta^\delta$.
The other terms in the exponent behaves as $\mathcal{O}(t^{1/2}, t^{1/3})$ because the functions $g_2(w,0,0)$ and $g_3(w)$ given in \eqref{g def2} are regular for any $w \in \Theta' \backslash \Theta^\delta$.
Owing to $\delta = t^{-1/6}$, $a_1 \delta^2 t$ is dominant in the exponent.
Therefore, we obtain

\begin{equation*}
\label{Gaussian_part_rest_2}
\begin{split}
\left| \int_{\Theta' \backslash \Theta^\delta} \frac{\dd w}{2 \pi \ii} \frac{\rho'}{w - \rho'}  \e^{\bar{h}(w,t)-\bar{h}(\rho',t) - \kappa (w - \rho') t^{1/3} + g'_J(w)} \right| \leq c_3 \e^{ - a_1 \delta^2 t / 2 } ,
\end{split}
\end{equation*}
where $c_3$ is some positive constant independent of $\kappa$.
From the above, it turns out that

\begin{equation*}
\begin{split}
& \oint_{\Theta'} \frac{\dd w}{2 \pi \ii} \frac{\rho'}{w - \rho'}  \e^{\bar{h}(w,t)-\bar{h}(\rho',t) - \kappa (w - \rho') t^{1/3} + g'_J(w)} \\
= & \e^{ - \frac{1}{2} {(s_{\rm g} + \rho' \kappa t^{-1/6} / \lambda_2)}^2 } \int_{- \sqrt{2} \lambda_2 \delta t^{1/2} / \rho'}^{ \sqrt{2} \lambda_2 \delta t^{1/2} / \rho' } \frac{\dd u}{2 \pi \ii} \frac{\rho}{u - \sqrt{2} \ii \beta } \e^{ - \frac{1}{4} {\left[ u + \sqrt{2} \ii (s_{\rm g} + \rho' \kappa t^{-1/6} / \lambda_2 - \beta) \right]}^2 } +\mathcal{O}(t^{-1/6}, \e^{- a_1 \delta^2 t / 2 }) \\
\underset{t \rightarrow \infty}{\longrightarrow} & \e^{ - \frac{1}{2} {(s_{\rm g} + \rho' \kappa t^{-1/6} / \lambda_2)}^2 } \int_{- \infty}^{ \infty } \frac{\dd u}{2 \pi \ii} \frac{\rho}{u - \sqrt{2} \ii \beta } \e^{ - \frac{1}{4} {\left[ u + \sqrt{2} \ii (s_{\rm g} + \rho' \kappa t^{-1/6} / \lambda_2 - \beta) \right]}^2 } . 
\end{split}
\end{equation*}
We can see that the function in the third line can be bounded by some constant independent of $\kappa$, and this leads to the result \eqref{Gaussian_part_Bound}.

\begin{figure}[h]
\begin{center}
\begin{tikzpicture}[scale=2.5]
jst
\draw[->,thick] (-2,0) -- (2,0);
\draw[->,thick] (0,-1) -- (0,0.6);

\node at (0.06,-0.10) {$0$};
\node at (0.35,0.12) {$D \epsilon t^{1/2}$};
\node at (-0.35,0.12) {$-D \epsilon t^{1/2}$};
\node at (1.55,0.12) {$D \delta t^{1/2}$};
\node at (-1.55,0.12) {$-D \delta t^{1/2}$};
\node at (0.2,-0.9) {$- \sqrt{2} \ii \beta $};

\node[scale=1.5] at (0.8,0.4) {$\theta^{\delta t^{1/2}}$};
\node[scale=1.5] at (-1.7,-0.45) {$C_{-}$};
\node[scale=1.5] at (1.7,-0.45) {$C_{+}$};

\draw[->,thick] (0.2,0) arc (0:-180:0.2);

\coordinate (A) at (1.5,0);
\coordinate (B) at (0.2,0);
\coordinate (C) at (-0.2,0);
\coordinate (D) at (-1.5,0);
\coordinate (E) at (-1.5,-0.8);
\coordinate (F) at (1.5,-0.8);

\draw[->,thick] (C) -- (D);
\draw[->,thick] (D) -- (E);
\draw[->,thick] (E) -- (F);
\draw[->,thick] (F) -- (A);
\draw[->,thick] (A) -- (B) ;

\end{tikzpicture}
\end{center}
\caption{The contour of the integral chosen to derive \eqref{Gaussian_part_divide_2} where $D := \sqrt{2} \lambda_2 / \rho'$. The paths consisting of the lines along the real axis and semi-circle centered at the origin denote the path $\theta^{\delta t^{1/2}}$.}
\label{fig:contour gaussian pv}
\end{figure}
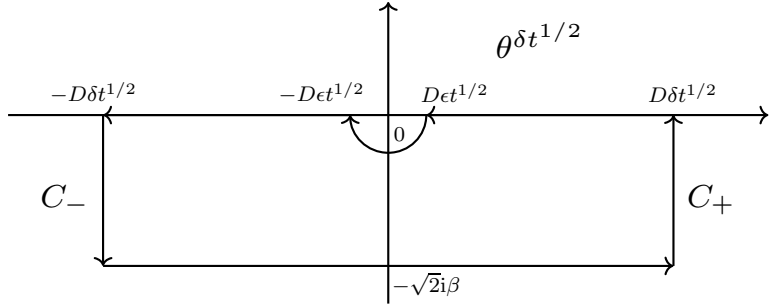

\paragraph{3, Estimate of $G_2(\kappa, s_2, \xi , \zeta)$}

In a similar way to Propositions~\ref{uniform convergence kernelc} and \ref{exp_bound},  we can show that

\begin{equation}
\label{Airy_parts_Bound}
\left| G_2(\kappa, s_2 , \xi , \zeta) \right| \leq C_2 \e^{-(\xi + \zeta)}
\end{equation}
holds for $t$ large enough and $(\xi, \zeta) \in [0, \infty )^2$ with some positive constant $C_2$ independent of $\kappa$.
Likewise the proofs of Propositions~\ref{uniform convergence kernelc} and \ref{exp_bound}, we separate the proof into two cases of $(\xi, \zeta) \in [0,L]^2$ and $ (\xi, \zeta) \in [0, \infty)^2 \backslash [0,L]^2$ for $L$ large enough but independent of $t$.

\paragraph{(i) Estimate for $(\xi , \zeta) \in [0,L]^2$}

Performing a steepest descent analysis, we can prove

\begin{equation}
\label{Airy_parts_limit}
 \lim_{t \to \infty}G_2(\kappa, s_2 , \xi , \zeta) = c_4 \mathrm{Ai}(\xi + s_2 + \kappa/\lambda) \mathrm{Ai}(\zeta + s_2) ,
\end{equation}
where $c_4 := 1/[\lambda (w_c + 1)]$ for any $\kappa \in (0, \infty)$ and $(\xi, \zeta) \in [0,L]^2$, and this yields \eqref{Airy_parts_Bound} because $|\mathrm{Ai}(x)|\leq C_a \e^{-a x}$ holds for $x \in \mathbb{R}$ with some $a > 0$ and $C_a >0$ depending on $a$.

Since the integral with respect to $w$ in \eqref{Airy_parts} does not depend on $\kappa$, by the same arguments as in Proposition~\ref{uniform convergence kernelc}, we can prove that it converges to the Airy function uniformly for $\zeta \in [0,L]$ and $\kappa \in (0, \infty)$, i.e.,

\begin{equation}
\label{Airy_parts_w_limit}
\lim_{t \to \infty} \lambda_c t^{1/3} \oint_{\Sigma'} \frac{\dd w}{2 \pi \ii} \e^{-\bar{f}(w,t) + \bar{f}(w_c,t) + g_{\psi}(w)} {\left( \frac{w+c}{w_c + c} \right)}^{\zeta \lambda_c t^{1/3}} = c_5 \mathrm{Ai}(\zeta + s_2),
\end{equation}
where $c_5 := 1/(w_c + \rho')$.

On the other hand, the integral with respect to $z$ in \eqref{Airy_parts} depends on $\kappa$, and hence we will evaluate it by slightly different calculations from Proposition~\ref{uniform convergence kernelc}.
First, we focus on the main contribution from ${\Gamma'}^\Delta$ defined in \eqref{gamma_prime_Delta} where $\Delta = t^{-1/9}$.
Since the Taylor expansion of $\bar{f}(z,t)$ around $z = w_c = \rho'/2$ is given by \eqref{gExpands}, we obtain

\begin{multline*}
\label{Airy_parts_z}
t^{1/3} \int_{{\Gamma'}^\Delta} \frac{\dd z}{2 \pi \ii} \e^{\bar{f}(z,t)-\bar{f}(w_c,t) - \kappa (z - w_c) t^{1/3} + g'_{\phi}(z)}  {\left( \frac{w_c + c}{z + c} \right)}^{\xi \lambda_c t^{1/3}}  =  c_6 \int_{\gamma^{\Delta t^{1/3}}} \frac{\dd u}{2 \pi \ii} \e^{ \frac{u^3}{3} - (s_2 + \xi + \kappa / \lambda) u } + R_u,
\end{multline*}
where $c_6 := (w_c + \rho')/[\lambda (w_c + 1)]$ and the remainder $R_u$ is given by
\begin{equation*}
R_u = c_6 \int_{\gamma^{\Delta t^{1/3}}} \frac{\dd u}{2 \pi \ii} \e^{ \frac{u^3}{3} - (s_2 + \xi + \kappa / \lambda) u } \left( \e^{\mathcal{O}(u^4 t^{-1/3}, u^2 t^{-1/6}, u t^{-1/3})} - 1 \right)
\end{equation*}
with  $\gamma^{\Delta t^{1/3}} := \{ u \in \mathbb{C} \mid w_c + u /( \lambda t^{1/3} ) \in {\Gamma'}^{\Delta} \}$ and $\lambda$ defined in \eqref{DefNormalLambda}.
In the integrand, $ \left| \e^{- (\xi + \kappa / \lambda) u } \right| \leq 1 $ holds because $\xi, \kappa, \lambda$ are positive and $\Re(u) \geq 0$ holds for any $u \in \gamma^{\Delta t^{1/3}}$.
Therefore, in the same fashion as Appendix~\ref{Unif_conv_Kc}, we can show

\begin{equation*}
|R_u| \leq c_7 t^{-1/6},
\end{equation*}
where $c_7$ is some positive constant independent of $\kappa$.
Second, we focus on the contribution from $\Gamma' \backslash {\Gamma'}^\Delta$.
Since $|\e^{- \kappa (z - w_c) t^{1/3}}| \leq 1$ holds for any $\kappa \in (0, \infty)$ and $z \in \Gamma'$, we can evaluate an upper bound in the same way as Proposition~\ref{uniform convergence kernelc}, i.e.,

\begin{equation*}
\label{Airy_parts_z_2}
\left| t^{1/3} \int_{\Gamma' \backslash {\Gamma'}^\Delta } \frac{\dd z}{2 \pi \ii} \e^{\bar{f}(z,t)-\bar{f}(w_c,t) - \kappa (z - w_c) t^{1/3} + g'_{\phi}(z)}  {\left( \frac{w_c + c}{z + c} \right)}^{\xi \lambda_c t^{1/3}} \right| \leq c_8 \e^{- a_1 \Delta^3 t/2 }
\end{equation*}
with some positive constant $c_8$ independent of $\kappa$ and $a_1$ given in \eqref{gCoefs}.

Consequently, the integral with respect to $z$ in \eqref{Airy_parts} converges to the Airy function uniformly for $\xi \in [0,L]$ and $\kappa \in (0, \infty)$, i.e.,

\begin{equation}
\label{Airy_parts_z_limit}
\lim_{t \to \infty} t^{1/3} \oint_{\Gamma'} \frac{\dd z}{2 \pi \ii} \e^{\bar{f}(z,t)-\bar{f}(w_c,t) - \kappa (z - w_c) t^{1/3} + g'_{\phi}(z)}  {\left( \frac{w_c + c}{z + c} \right)}^{\xi \lambda_c t^{1/3}} = c_6 \mathrm{Ai}(\xi + s_2 + \kappa / \lambda).
\end{equation}
Combining \eqref{Airy_parts_w_limit} and \eqref{Airy_parts_z_limit}, we arrive at
\eqref{Airy_parts_limit}.

\paragraph{(ii) Estimate for $(\xi , \zeta) \in [0, \infty )^2 \backslash [0,L]^2$}

In the same way as Proposition~\ref{exp_bound}, we can prove that \eqref{Airy_parts_Bound} holds for large enough $t$ and large enough $L$ but independent of $t$.
Obviously, we have

\begin{multline}
\label{Airy_parts_Bound_2}
\left| G_2(\kappa, s_2, \xi , \zeta) \right| \leq t^{1/3} \oint_{\Gamma'} \left| \frac{\dd z}{2 \pi \ii} \right| \left| \e^{\bar{f}(z,t)-\bar{f}(w_c,t) - \kappa (z - w_c) t^{1/3} + g'_{\phi}(z)} \right| {\left| \frac{w_c + c}{z + c} \right|}^{\xi \lambda_c t^{1/3}} \times  \\   
\lambda_c t^{1/3} \oint_{\Sigma'} \left| \frac{\dd w}{2 \pi \ii} \right| \left| \e^{-\bar{f}(w,t) + \bar{f}(w_c,t) + g_{\psi}(w)} \right| {\left| \frac{w+c}{w_c + c} \right|}^{\zeta \lambda_c t^{1/3}}.
\end{multline}
In the integrand, $|\e^{- \kappa (z - w_c) t^{1/3}}| \leq 1$ holds for any $\kappa \in (0, \infty)$ and $z \in \Gamma'$. Hence, we can apply the same arguments as in Proposition~\ref{exp_bound} to the right hand side of \eqref{Airy_parts_Bound_2}, which results in \eqref{Airy_parts_Bound}.

\paragraph{4, Conclusion} Combining the above two results \eqref{Gaussian_part_Bound} and \eqref{Airy_parts_Bound}, we arrive at

\begin{equation*}
\begin{split}
\left| \sum_{p=1}^{n-1} \oint_1 \dd^{n-1} z\,  L(\vec{z}) \prod_{q=1}^{p}(z_q-1) \bar{A}^c_{t,p}(\xi) \bar{B}^c_t(\zeta) \right|
\leq &  \int_{0}^{\infty} \dd \kappa \e^{ - \frac{3}{2} \rho' t^{1/3} \kappa } \left| G_g(\kappa, s_g) \right| \left| G_2(\kappa, s_2 , \xi , \zeta) \right| \\
\leq & C_1 C_2 \e^{-(\xi + \zeta)} \int_{0}^{\infty} \dd \kappa \e^{ - \frac{3}{2} \rho' \kappa t^{1/3} } \\
= & D_1 t^{-1/3} \e^{- (\xi + \zeta)}
\end{split}
\end{equation*}
with $D_1 := 2C_1 C_2 / (3 \rho')$.
\qed
\end{proof}

Utilising Propositions~\ref{exp_bound_W} and \ref{I2(2)_up_bound_AB}, we get the following theorem.

\begin{theorem}
\label{thm_I2(2) limit}

With the scaling \eqref{nmscale}, the long time limit of the term $\mathcal{I}_2^{(2)}$ defined in \eqref{def:I2(2)} is given by

\begin{equation}
\label{I2(2) limit}
\lim_{t\rightarrow\infty}\mathcal{I}^{(2)}_2 = 0.
\end{equation}
\end{theorem}

\begin{proof}

It follows from Proposition \ref{exp_bound_W} and the Hadamard's inequality that the determinant part of $\mathcal{I}_2^{(2)}$, given in \eqref{def:I2(2)_rescaled}, is bounded above as

\begin{equation}
\label{Kc_up_bound_lrloss}
\left| \det\left[ \bar{K}^c_{\mathrm{W},t}(x'_i, x'_j) \right]_{\substack{ 1 \leq i , j \leq k \\  i \neq l , j \neq r } } \right| \leq {(k - 1)}^{(k-1)/2} \prod_{\substack{ 1 \leq i \leq k \\  i \neq l,  r } }{ \e^{-x'_i} }.
\end{equation}
Combining Proposition \ref{I2(2)_up_bound_AB} and \eqref{Kc_up_bound_lrloss}, the absolute value of $\mathcal{I}_2^{(2)}$ turns out to be bounded above as

\begin{equation}\label{I2(2)_up_bound}
\left| \mathcal{I}_2^{(2)} \right| \leq D_1 t^{-1/3} \sum_{k=1}^{m}{ k \frac{ {(k-1)}^{(k-1)/2} }{(k-1)!} } \prod_{i = 1}^{k}{ \left\{ \frac{1}{\lambda_c t^{1/3}} \sum_{x_i = 1}^{\infty}{ \e^{ - {x_i}/{(\lambda_c t^{1/3})} } } \right\} }.
\end{equation}
From $1 + x < \e^x$ for $x>0$ and $\sum_{k=1}^{\infty}{ k {{(k-1)}^{(k-1)/2}} /{(k-1)!} } < \infty$, we obtain

\begin{equation*}
\frac{1}{\lambda_c t^{1/3}} \sum_{x_i = 1}^{\infty}{ \e^{ - {x_i}/{(\lambda_c t^{1/3})} } } =  \frac{1}{\lambda_c t^{1/3}} \frac{1}{\e^{1/{(\lambda_c t^{1/3})}} - 1} < 1,
\end{equation*}
and this inequality guarantees that the right hand side of \eqref{I2(2)_up_bound} converges to zero in the long time limit with the scaling \eqref{nmscale}.
Therefore, we have \eqref{I2(2) limit}.
\qed

\end{proof}

Eventually, it follows from Theorems~\ref{thm_I2(1) limit} and \ref{thm_I2(2) limit} the following equalities hold.

\begin{equation}
\label{I2(1)-I2(2) limit}
\lim_{t \to \infty} \mathcal{I}_2 = \lim_{t \to \infty} \left( \mathcal{I}^{(1)}_2 - \mathcal{I}^{(2)}_2 \right) = \left[ 1 - F_G(s_{\rm g}) \right] F_2(s_{2}).
\end{equation}
Obviously, the second equality \eqref{I2(1)-I2(2) limit} itself coincides with the second claim of Theorem~\ref{thm_I1I2 limit}.
As mentioned below Theorem~\ref{thm_I1I2 limit} in Section~\ref{se:scalinglimit}, we obtain the final result as Corollary~\ref{cor_Pnm limit}.

\appendix

\section*{Acknowledgment}
We are grateful to P.L.~Ferrari, I.~Kostov, H.~Spohn and M.~Wheeler for discussions. This work was initiated at KITP Santa Barbara during the program \textit{New approaches to non-equilibrium and random systems: KPZ integrability, universality, applications and experiments}, which was supported in part by the National Science Foundation under Grant No. NSF PHY11-25915. Part of this work was performed during a stay of all authors at the MATRIX mathematical research institute in Australia during the program \textit{Non-equilibrium systems and special functions}. JdG and ZC gratefully acknowledge support from the Australian Research Council. The work of TS has been supported by JSPS KAKENHI Grants No. JP15K05203, No. JP16H06338, No. JP18H01141, No. JP18H03672, No. JP19L03665. 
The work of MU has been supported by Public Trust Iwai Hisao Memorial Tokyo Scholarship Fund.

\appendix
%\newpage{\pagestyle{empty}\cleardoublepage}

\section{Bethe wave function}
\label{appx:bethe wave fn}
We consider the AHR model with $n$ plus and $m$ minus particles. Define a set of coordinates by $\mathbb{W}^k:=
\{\vec{x}=(x_1,\dots,x_k)\in\mathbb{Z}^k:
x_1<x_2<\dots<x_k\}$, and let $\vec{x}=(x_1,\dots,x_n)\in \mathbb{W}^n$ and $\vec{y}=(y_1,\dots,y_m)\in \mathbb{W}^m$ be the positions of plus and minus particles, respectively. To each set of coordinates we associate a unit basis vector $\ket{\vec{x};\vec{y}}$, and construct the state space as $S={\rm Span}\{\ket{\vec{x};\vec{y}}\}$. Let
\[
\ket{P(t)} := \sum_{\vec{x},\vec{y}} P(\vec{x},\vec{y};t) \ket{\vec{x};\vec{y}}
\]
be the vector of probabilities to observe a configuration $\ket{\vec{x},\vec{y}}$ of the AHR model at time $t$ with plus particles at positions $\vec{x}$ and minus particles at positions $\vec{y}$. The vector $\ket{P(t)}$ satisfies the master equation
\begin{equation}\label{masterEqMatrix}
\frac{\dd}{\dd t} \ket{P(t)} = M \ket{P(t)},
\end{equation}
where $M$ is the transition matrix encoding the jumping rates of the AHR model. The $\big((\vec{x};\vec{y}),(\vec{x}';\vec{y}')\big)$ entry of $M$ is the transition rate from state $(\vec{x}';\vec{y}')$ to $(\vec{x};\vec{y})$, given by
\begin{align}
M\big((\vec{x};\vec{y}),(\vec{x}';\vec{y}')\big)
=
\left\{
\begin{array}{cl}
	\beta,
	&
	\textrm{  if }(\vec{x}';\vec{y}')
	= (\vec{x}_i^-;\vec{y}), \forall i \in [1,n]
	 \\
	\alpha,
	&
	\textrm{  if }(\vec{x}';\vec{y}')
	= (\vec{x};\vec{y}_j^+),
	\forall j \in [1,m]
	 \\
	-\beta n - \alpha m  
	&
	\textrm{  if }(\vec{x}';\vec{y}') = (\vec{x};\vec{y})
	 \\
	0,
	&
	\text{  otherwise},
	\end{array}
	\right.
\end{align}
where $\vec{x}_i^{\pm}:=(x_1,\dots,x_i {\pm} 1,\dots,x_n)$. \eqref{masterEqMatrix} is equivalent to \eqref{master eq AHR}.

To derive an expression for $P(\vec{x},\vec{y};t)$, or more precisely the transition probability \eqref{Green function} considered in Section \ref{se:greenfunction}, we first look at the eigenvalue problem of the master equation with eigenvector
\begin{equation*}
\ket{\psi}
=
\sum_{\vec{x},\vec{y}}
\psi(\vec{x};\vec{y})
\ket{\vec{x};\vec{y}},\qquad M\ket{\psi} = \Lambda_{n,m} \ket{\psi},
\end{equation*}
where $\Lambda_{n,m}$ is the eigenvalue.

We will follow Bethe's method in \cite{bethe} to obtain the Bethe wave function $\psi(\vec{x},\vec{y})$ for the AHR model. First we consider the case of single species particles only, before considering the two species case. Then for each case we start with two particles and then generalise to the many particles case.

Such method of Bethe ansatz is first introduced in \cite{bethe}, and applied to many particles systems in \cite{LL1963,Y2004,YY1966}. Especially, the Bethe ansatz was used to solve the master equation of ASEP and TASEP in \cite{GS1992,S1997}. A more complicated formula would follow from the standard nested Bethe ansatz \cite{C2008}.

\subsection{Single species}

First we consider the case without minus particles and two plus particles, i.e., $n=2,m=0$. The corresponding eigenvalue equation for $x_1+1<x_2$ is
\begin{equation}
\Lambda_{2,0} \psi(x_1,x_2)
=
\beta\psi(x_1-1,x_2) + \beta\psi(x_1,x_2-1) - 2\beta\psi(x_1,x_2),
\label{+eigvaleq1}
\end{equation}
while for $x_1=x=x_2-1$ it is
\begin{equation}
\Lambda_{2,0} \psi(x,x+1)
=
\beta\psi(x-1,x+1)-\beta\psi(x,x+1).
\label{+eigvaleq2}
\end{equation}
These two equations can be solved by the trial wave function $$\psi(x_1,x_2)=A_{12}z_1^{x_1}z_2^{x_2}+A_{21}z_1^{x_2}z_2^{x_1},$$ for which
\eqref{+eigvaleq1} results in the eigenvalue expression
$\Lambda_{2,0}=\beta(z_1^{-1}+z_2^{-1}-2),$ while \eqref{+eigvaleq2} gives the scattering relation $A_{12}/A_{21}=-(1-z_1)/(1-z_2)$.

The number of equations increases as the number $n$ of plus particles increases. To solve the resulting system of equations it is convenient to replace \eqref{+eigvaleq2} by
\begin{equation}
\psi(x,x)=\psi(x,x+1),
\label{+bdrycond}
\end{equation}
because by imposing \eqref{+bdrycond} on \eqref{+eigvaleq1} for $x_1=x=x_2-1$, equation \eqref{+eigvaleq2} is automatically satisfied. In other words, \eqref{+eigvaleq1} along with the boundary condition \eqref{+bdrycond} gives the same solution as that of the eigenvalue problem. This boundary condition \eqref{+bdrycond} is the same as in \cite{S1997}, since the AHR model reduced to TASEP when there is only one species particles.

It can thus be seen that the eigenvalue problem corresponding to $m=0$ and general $n$ is equivalent to
\begin{align*}
\Lambda_{n,0}\psi(\vec{x})
=&
\beta \sum_{i=1}^n\psi(\vec{x}_i^-) - n \beta \psi(\vec{x}),
\\
\psi(x_1,\dots,x_i,x_{i+1}=x_{i},\dots,x_n)
=&
\psi(x_1,\dots,x_i,x_{i+1}=x_{i}+1,\dots,x_n), \quad i\in[1,n-1].
\end{align*}

These equations are solved by the wave function $$\psi(\vec{x})=
\sum_{\pi \in S_n}A_{\pi}\prod_{i=1}^nz_{\pi_i}^{x_i},$$ with eigenvalue $\Lambda_{n,0}=\beta\sum_{i=1}^{n}(z_i^{-1}-1)$, and with scattering relation
\begin{equation*}
\frac{A_{\pi}}{A_{s_i \pi}}
=
-\frac{1-z_{\pi_i}}{1-z_{\pi_{i+1}}},
\end{equation*}
where $s_i\; (i=1,\ldots, n-1)$ are simple transposition, i.e., the generators of the symmetric group $S_n$. This is the same ratio of amplitudes obtained in \cite{S1997}. This ratio is satisfied if $A_{\pi}$ is of the form
\begin{equation*}
A_{\pi}
=
\sign(\pi) \prod_{i=1}^n \left(
\frac{1}{1-z_{\pi_i}}
\right)^i.
\end{equation*}

The derivation of $\psi(\vec{y})$, when $n=0$ and $m$ general, is very similar to the case when $m=0$ and $n$ general. The corresponding eigenvalue problem is 
\begin{align*}
\Lambda_{0,m}\psi(\vec{y})
=&\alpha \sum_{j=1}^m\psi(\vec{y}_j^+) - m \alpha \psi(\vec{y}),\\
\psi(y_1,\dots,y_{i-1}=y_{i},y_i,\dots,y_m)
=&
\psi(y_1,\dots,y_{i-1}=y_{i}-1,y_i,\dots,y_m), \quad i\in[2,m],
\end{align*}
and we will not repeat analogous details of derivation here. Naively this observation would suggest that a wave function for two species particles is of the form
\begin{equation}
\sum_{\pi\in S_n}\sign(\pi)
\prod_{i=1}^{n}
\left(\frac{1}{1-z_{\pi_i}}\right)^i
z_{\pi_i}^{x_i}
\sum_{\sigma\in S_m}\sign(\sigma)
\prod_{j=1}^{m}
\left(\frac{1}{1-w_{\sigma_j}}\right)^{-j}
w_{\sigma_j}^{-y_j},
\label{wavefn:form}
\end{equation}
with eigenvalue $\Lambda_{n,m}
=
\beta \sum_{j=1}^n (z_j^{-1}-1) +
\alpha \sum_{k=1}^m (w_k^{-1}-1).$ Such a trial wave function would only work when there is no interaction between the plus and the minus particles. In the following, we shall consider a modification of this wave function so that the interaction between different types of particles can be taken into account.

As we have seen, the Bethe wave function of the AHR model for the single species cases is identical to the one of the TASEP \cite{S1997}. The important detail of the AHR model's wave function lies in the interaction between different types of particles.

\subsection{Two species case}
Now consider the case that two types of particles are present. We start with the simplest case with one plus and one minus particles. When $x\neq y\pm 1$, the corresponding eigenvalue equation reads as
\begin{equation}
\Lambda_{1,1} \psi(x;y)
=
\alpha  \psi(x;y+1)  + \beta \psi(x-1;y) - (\beta+\alpha )\psi(x;y),
\label{+-eigvaleq1}
\end{equation}
while for $x=y+1$ and $y=x+1$,
\begin{align}
\Lambda_{1,1}\psi(y+1;y)
=&
\psi(y;y+1) - (\beta+\alpha )\psi(y+1;y),
\label{+-eigvaleq2}
\\
\Lambda_{1,1}\psi(x;x+1)
=&
\beta\psi(x-1;x+1) + \alpha \psi(x;x+2) - \psi(x;x+1).
\label{+-eigvaleq3}
\end{align}
Expression~\eqref{+-eigvaleq1} is satisfied by the wave function \eqref{wavefn:form} and the eigenvalue $\Lambda_{1,1}=\beta(z^{-1}-1) +\alpha (w^{-1}-1)$. In the same way, \eqref{+-eigvaleq2} and \eqref{+-eigvaleq3} are satisfied by imposing appropriate boundary conditions. Comparing \eqref{+-eigvaleq1} and \eqref{+-eigvaleq2} gives the first condition,
\begin{equation*}
\psi(y;y+1)
=
\alpha \psi(y+1;y+1) +
\beta\psi(y;y).
\end{equation*}
This indicates a pieces-wise wave function 
\begin{equation*}
\psi(x;y) = \left\{
\begin{array}{l}
\displaystyle C_{+-} z^{x} w^{-y}\quad x<y\\
\displaystyle C_{-+} z^{x} w^{-y}\quad x\geq y
\end{array}\right. ,
\end{equation*}
so that \eqref{+-eigvaleq2} is satisfied by setting the amplitude ratio to $C_{+-}/C_{-+}=\alpha z+\beta w$. Fortunately, when $\beta+\alpha =1$ equation \eqref{+-eigvaleq3} agrees with \eqref{+-eigvaleq1} without any extra condition. Incidentally, the condition $\beta+\alpha =1$ leads to a factorised stationary state \cite{RSS2000}. The reasoning above suggests to modify \eqref{wavefn:form} to
\begin{equation}
\psi(\vec{x};\vec{y})
=
\sum_{\pi\in S_n}\sign(\pi)
\prod_{i=1}^{n}
\left(\frac{1}{1-z_{\pi_i}}\right)^i
z_{\pi_i}^{x_i}
\sum_{\sigma\in S_m}\sign(\sigma)
\prod_{j=1}^{m}
\left(\frac{1}{1-w_{\sigma_j}}\right)^{-j}
w_{\sigma_j}^{-y_j}C_{\vec{x},\vec{y},\pi,\sigma},
\label{wavefn}
\end{equation}
where $C_{\vec{x},\vec{y},\pi,\sigma}$ depends on the relative position of $\vec{x},\vec{y}$. For general $n,m$, we only need to impose the following boundary conditions on \eqref{wavefn},
\begin{multline}
\psi(\vec{x};
y_1,\dots,y_j=x_i+1,\dots,y_m)
=
\beta\psi(\vec{x};
y_1,\dots,y_j=x_i,\dots,y_m) +
\\
\alpha \psi(x_1,\dots,x_{i-1},x_i+1,x_{i+1},\dots,x_n;
y_1,\dots,y_j=x_i+1,\dots,y_m).
\label{+-bdrycond}
\end{multline}
As in the single species case, there are no extra boundary condition for sectors with more than two particles. Substituting \eqref{wavefn} into \eqref{+-bdrycond} gives the ratio
\begin{equation*}
\frac{C_{y_j>x_i}}{C_{y_j\leq x_i}}
=
\alpha z_{\pi_i}+\beta w_{\sigma_j}.
\end{equation*}
This suggests a form of the coefficient $C_{\vec{x},\vec{y},\pi,\sigma}$
\begin{equation*}
C_{\vec{x},\vec{y},\pi,\sigma}
=
\prod_{j=1}^m
\prod_{k=1}^{r_j}
\frac{1}{\alpha z_{\pi_{n-k+1}}+\beta w_{\sigma_j}},
\end{equation*}
where $r_j$ is the number of plus particles to the right of the $j$\textsuperscript{th} minus particle, i.e., $r_j=\#\{x_i\in\vec{x}\mid x_i \geq y_j\}$.

In conclusion, the general Bethe wave function is given by
\begin{multline}
\label{eq:BAfull}
\psi(\vec{x};\vec{y})=
\sum_{\pi\in S_n}\sign(\pi)
\prod_{i=1}^{n}
\left(\frac{1}{1-z_{\pi_i}}\right)^i
z_{\pi_i}^{x_i} \times\\
\sum_{\sigma\in S_m}\sign(\sigma)
\prod_{j=1}^{m}
\left(\frac{1}{1-w_{\sigma_j}}\right)^{-j}
w_{\sigma_j}^{-y_j}
\prod_{k=1}^{r_j}
\frac{1}{\alpha z_{\pi_{n-k+1}}+\beta w_{\sigma_j}},
\end{multline}
with eigenvalue
\begin{equation}
\Lambda_{n,m}
=
\beta \sum_{j=1}^n (z_j^{-1}-1) +
\alpha \sum_{k=1}^m (w_k^{-1}-1).
\label{eigvalue}
\end{equation}
This result leads to an integral formula for the transition probability, as stated in Section \ref{se:greenfunction}.

\section{Boundary conditions for the transition probability}
\label{se:bcs}

In this appendix we provide details of the proof that \eqref{Green function} satisfies the boundary conditions (\ref{AHR bdrycond1})-(\ref{AHR bdrycond3}). Throughout the entire proof, we shall call the factor $\prod_{k=1}^{m}\prod_{j=1}^{r_k}(\alpha z_{\pi_{n-j+1}}+\beta w_{\sigma_k})^{-1}$ the scattering factor.

\subsection{Proof of boundary condition (\ref{AHR bdrycond1})}

On both sides of \eqref{AHR bdrycond1}, i.e., when $x_{i+1}=x_i$ on the left hand side and $x_{i+1}=x_i+1$ on the right hand side, the scattering factor $\prod_{k=1}^{m}\prod_{j=1}^{r_k}(\alpha z_{\pi_{n-j+1}}+\beta w_{\sigma_k})^{-1}$ remains unchanged within the physical regions $\Omega^{n+m}$. Moreover, the scattering factor is symmetric in $z_{\pi_i},z_{\pi_{i+1}}$ since there is no minus particle between the $i$\textsuperscript{th} and the $i+1$\textsuperscript{st} plus particles.

From \eqref{Green function} we observe that $x_i$ and $x_{i+1}$ only appear as exponents of $z_{\pi_i}$ and $z_{\pi_{i+1}}$, and as there is no change in the scattering factor we only need to compare the $z_{\pi_i}$ and $z_{\pi_{i+1}}$ factors
\begin{align}
\label{AHR z component bdrycond}
\left(\frac{1}{1-z_{\pi_i}}\right)^i
\left(\frac{1}{1-z_{\pi_{i+1}}}\right)^{i+1}
z_{\pi_i}^{x_i}z_{\pi_{i+1}}^{x_{i+1}}
\end{align}
in the integrand for both sides of \eqref{AHR bdrycond1}. For the left hand side of \eqref{AHR bdrycond1}, i.e., when $x_{i+1}=x_{i}$, the factor \eqref{AHR z component bdrycond} reads as
\begin{align*}
LHS
=
\left(\frac{1}{1-z_{\pi_i}}\right)^i
\left(\frac{1}{1-z_{\pi_{i+1}}}\right)^{i+1}
(z_{\pi_i}z_{\pi_{i+1}})^{x_i}.
\end{align*}
When $x_{i+1}=x_i+1$, \eqref{AHR z component bdrycond} becomes
\begin{align*}
RHS
=
\left(\frac{1}{1-z_{\pi_i}}\right)^i
\left(\frac{1}{1-z_{\pi_{i+1}}}\right)^{i+1}
(z_{\pi_i}z_{\pi_{i+1}})^{x_i}z_{\pi_{i+1}}.
\end{align*}
It follows that
\begin{align*}
LHS-RHS
=&
\left[\frac{1}{(1-z_{\pi_i})(1-z_{\pi_{i+1}})}\right]^i
(z_{\pi_i}z_{\pi_{i+1}})^{x_i}
\left(
\frac{1}{1-z_{\pi_{i+1}}}-
\frac{z_{\pi_{i+1}}}{1-z_{\pi_{i+1}}}\right)
\\
=&
\left[\frac{1}{(1-z_{\pi_i})(1-z_{\pi_{i+1}})}\right]^i
(z_{\pi_i}z_{\pi_{i+1}})^{x_i}.
\end{align*}
This factor is symmetric in $z_{\pi_i},z_{\pi_{i+1}}$. Moreover, the scattering factor is also symmetric in $z_{\pi_i},z_{\pi_{i+1}}$, and hence summing over $\pi\in S_n$ in \eqref{Green function} gives a zero integrand due to the factor $\sign(\pi)$.

\subsection{Proof of boundary condition (\ref{AHR bdrycond2})}

The proof of this boundary condition is very similar to the one above. First we notice that when $y_{i-1}=y_i$ and $y_{i-1}=y_i-1$, the scattering factor is unchanged and symmetric in $w_{\sigma_{i-1}},w_{\sigma_i}$. Hence as before, we only need to consider
\begin{align}
\label{AHR w component bdrycond}
\left(\frac{1}{1-w_{\sigma_{i-1}}}\right)^{-i+1}
\left(\frac{1}{1-w_{\sigma_i}}\right)^{-i}
w_{\sigma_{i-1}}^{-y_{i-1}}w_{\sigma_i}^{-y_i}.
\end{align}
Then difference of this factor between the left hand side and right hand side of \eqref{AHR bdrycond2} is given by
\begin{align*}
LHS-RHS
=&
\left[
\frac{1}{(1-w_{\sigma_{i-1}})
(1-w_{\sigma_i})}\right]^{-i}
(w_{\sigma_{i-1}}w_{\sigma_i})^{-y_i}
\left(
\frac{1}{1-w_{\sigma_{i-1}}}-
\frac{w_{\sigma_{i-1}}}{1-w_{\sigma_{i-1}}}
\right)\\
=&
\left[
\frac{1}{(1-w_{\sigma_{i-1}})
(1-w_{\sigma_i})}\right]^{-i}
(w_{\sigma_{i-1}}w_{\sigma_i})^{-y_i},
\end{align*}
which is symmetric in $w_{\sigma_{i-1}},w_{\sigma_i}$. Thus summing over $\sigma\in S_m$ in \eqref{Green function} gives a zero integrand, as required.

\subsection{Proof of boundary condition (\ref{AHR bdrycond3})}

We consider two cases when $y_j=x_i$ and $y_j=x_i+1$. From the integrand of \eqref{Green function}, in these two cases, one only needs to check the factor
\begin{align*}
\left(\frac{1}{1-z_{\pi_i}}\right)^i
\left(\frac{1}{1-w_{\sigma_j}}\right)^{-j}
z_{\pi_i}^{x_i}w_{\sigma_j}^{-y_j}
\prod_{k=1}^{r_j}
\frac{1}{\alpha z_{\pi_{n-k+1}} + \beta w_{\sigma_j}},
\end{align*}
for both sides of \eqref{AHR bdrycond3}. For the left hand side of \eqref{AHR bdrycond3}, $y_j=x_i+1$, which indicates that the $i+1$\textsuperscript{st} plus particle is sitting at the right hand side of $y_j$. Therefore $r_j=n-i$.
\begin{align*}
LHS
=
\left(\frac{1}{1-z_{\pi_i}}\right)^i
\left(\frac{1}{1-w_{\sigma_j}}\right)^{-j}
\left(\frac{z_{\pi_i}}{w_{\sigma_j}}\right)^{x_i}
w_{\sigma_j}^{-1}
\prod_{k=1}^{n-i}
\frac{1}{\alpha z_{\pi_{i+k}}+ \beta w_{\sigma_j}}.
\end{align*}
Now let us consider the right hand side of \eqref{AHR bdrycond3} where $y_j=x_i$. In this case $r_j=n-i+1$, since the $i$\textsuperscript{th} plus particle is sitting at $y_j$. It follows that
\begin{align*}
RHS
=&
\beta
\left(\frac{1}{1-z_{\pi_i}}\right)^i
\left(\frac{1}{1-w_{\sigma_j}}\right)^{-j}
\left(\frac{z_{\pi_i}}{w_{\sigma_j}}\right)^{x_i}
\prod_{k=0}^{n-i}
\frac{1}{\alpha z_{\pi_{i+k}}+ \beta w_{\sigma_j}}
\\
&\qquad\qquad\qquad\qquad\qquad\qquad\qquad
+ \alpha \left(\frac{1}{1-z_{\pi_i}}\right)^i
\left(\frac{1}{1-w_{\sigma_j}}\right)^{-j}
\left(\frac{z_{\pi_i}}{w_{\sigma_j}}\right)^{x_i+1}
\prod_{k=0}^{n-i}
\frac{1}{\alpha z_{\pi_{i+k}}+\beta w_{\sigma_j}}
\\
=&
\left(\frac{1}{1-z_{\pi_i}}\right)^i
\left(\frac{1}{1-w_{\sigma_j}}\right)^{-j}
\left(\frac{z_{\pi_i}}{w_{\sigma_j}}\right)^{x_i}
\prod_{k=0}^{n-i}
\frac{1}{\alpha z_{\pi_{i+k}}+ \beta w_{\sigma_j}}
\left(\beta + \alpha \frac{z_{\pi_i}}{w_{\sigma_j}}\right)
\\
=&
\left(\frac{1}{1-z_{\pi_i}}\right)^i
\left(\frac{1}{1-w_{\sigma_j}}\right)^{-j}
\left(\frac{z_{\pi_i}}{w_{\sigma_j}}\right)^{x_i}
w_{\sigma_j}^{-1}
\prod_{k=1}^{n-i}
\frac{1}{\alpha z_{\pi_{i+k}}+\beta w_{\sigma_j}},
\end{align*}
which is exactly the same as LHS.

\section{Proof of symmetrisation identity Lemma~\ref{lem:symm identity}}
\label{se:symm identities}

We first recall the statement of Lemma~\ref{lem:symm identity},
\begin{equation}
\label{identity3}
\sum_{\pi\in S_n}\sign(\pi)\prod_{i=1}^n
\left(\frac{z_{\pi_i}-1}{z_{\pi_i}}\right)^i
\frac{1}{(1-(1-\rho)\prod_{j=1}^i z_{\pi_j})}
=
\frac{\prod_{1 \leq i<j \leq n}(z_j-z_i)
\prod_{i=1}^{n}(z_i-1)}
{\prod_{i=1}^{n}z_i^n(1-(1-\rho)z_i)}.
\end{equation}

We prove this identity by mathematical induction in $n$. One can easily see that \eqref{identity3} holds for $n=1$. We assume that it holds for $n-1$, and prove it for $n$. Let us denote the left hand side of the identity by $f_n(z_1,\dots,z_n)$. To make use of the induction assumption, we need to find the relation between $f_n$ and $f_{n-1}$. The sum over $S_n$ can be split into a sum over $k \in [1, n]$ such that $\pi_n=k$ and then sum over $S_{n-1}$. We observe that $\sign(\pi)=(-1)^{n-k}\sign(\sigma)$ where $\sigma$ is the permutation $\pi$ restricted on $\{1,2,\dots,k-1,k+1,\dots,n\}$, and $(-1)^{n-k}$ is the signature of the permutation $(1,\dots,k-1,k+1,\dots,n,k)$. Therefore,
\begin{align*}
f_n(z_1,\dots,z_n)
=
\sum_{k=1}^n(-1)^{n-k}
\left(\frac{z_k-1}{z_k}\right)^n
\frac{1}{1-(1-\rho)\prod_{i=1}^nz_{i}}
f_{n-1}(z_1,\dots,z_{k-1},z_{k+1},\dots,z_n).
\end{align*}
Substituting our induction assumption for $f_{n-1}(z_1,\dots,z_{k-1},z_{k+1},\dots,z_n)$ and after some rearrangements, we obtain

\begin{equation*}
f_n(z_1,\dots,z_n)
=
\frac{\prod_{1 \leq i<j \leq n}(z_j-z_i)}
{1-(1-\rho)\prod_{i=1}^nz_i}
\prod_{i=1}^n
\frac{z_i-1}{z_i^{n-1}}
\sum_{k=1}^n
\frac{(z_k-1)^{n-1}}{z_k}
\prod_{i = 1, i\neq k}^{n}
\frac{1}{(z_k-z_i)(1-(1-\rho)z_i)}.
\end{equation*}
In order to show the sum over $k$ gives expected result, we consider the function defined by
\begin{align*}
F(z)
=
\frac{(z-1)^{n-1}(1-(1-\rho)z)}
{z\prod_{i=1}^n(z-z_i)(1-(1-\rho)z_i)}.
\end{align*}
By the residue theorem,

\begin{align*}
\sum_{k=1}^{n}\res_{z=z_k}F(z)
=&
\sum_{k=1}^n
\frac{(z_k-1)^{n-1}}{z_k}
\prod_{i=1,i\neq k}^{n}
\frac{1}{(z_k-z_i)(1-(1-\rho)z_i)}
\\
=&
-\res_{z=0}F(z) - \res_{z=\infty}F(z)
\\
=&
\frac{1}
{\prod_{i=1}^nz_i(1-(1-\rho)z_i)}
-
\frac{(1-\rho)}
{\prod_{i=1}^n(1-(1-\rho)z_i)},
\end{align*}
which gives the required result to complete the proof. 

\section{Asymptotic of the rescaled kernel}

\subsection{Proof of Lemma~\ref{descent contour 1}} \label{appxs:SteepestContour1}
The facts that $w_c$ is a double root and $w_{2}$ is a single root  can be checked easily. To prove $\Gamma$ gives the steepest descent contour, we first show that $\R(g_1)$ is decreasing along $\Gamma_2\cup \Gamma_3$. On $\Gamma_2$, we find
\begin{multline*}
\frac{\dd \R(g_1)(s)}{\dd s}=
2 s^2 \Big[ - 3 (3 - \rho) (1 - \rho)^2 (1 + \rho) -
2 s (1 - p)(3 + 10 \rho - 5 \rho^2) - 2 s^2 (11 - 6 \rho + 3 \rho^2) \\
- 12 s^3 (1 - \rho) - 8 s^4\Big]/
\Big[\big(3s^2 + (1 + \rho - s)^2\big)
\big(3s^2 + (1 + s - \rho)^2\big)
\big(3s^2 + (3 + s - \rho)^2\big)\Big].
\end{multline*}
One can verify that $\dd \R(g_1)(s)/\dd s<0$ for any $s\in [0,2]$ and any $\rho\in[0,1]$, implying that $\R(g_1)$ is decreasing along $\Gamma_2$. It follows by symmetry that, $\R(g_1)$ is increasing along $\Gamma_1$. In fact, $\R(\ln(w))=\ln(|w|)$, so $\R(g_1)$ is symmetric with respect to the horizontal axis.

It remains to check the monotonicity of $\R(g_1)$ along $\Gamma_3$. Again taking the derivative along $\Gamma_3$, we have
\begin{multline*}
\frac{\dd \R(g_1)(\theta)}{\dd \theta}
= \Big[ 2 (1 + \rho)^2 \sin(w) \cos(w/2)^2 [17 - 5 \rho + 3 \rho^3 + \rho^3 + 8 (1 + \rho) \cos(w)]\Big] /\Big[[5 + 2 \rho +\rho^2
\\ + 4 (1 + \rho) \cos(w)] [17 + 2\rho +\rho^2 + 8 (1 + \rho) \cos(w)]\Big],
\end{multline*}
which equals zero only at $\theta=0,\pi$, i.e. $\R(g_1)$ decreases when $s\in[-2\pi/3,0]$, and increases when $s\in[0,2\pi/3]$. In fact $17 - 5 \rho + 3 \rho^3 + \rho^3 + 8 (1 + \rho) \cos(w)\geq 9 - 13 \rho + 3 \rho^3 + \rho^3$. The derivative of $9 - 13 \rho + 3 \rho^3 + \rho^3$ with respect to $\rho$ is $-13+6\rho+3\rho^2 \leq -13+6+3<0$, indicating that $9 - 13 \rho + 3 \rho^3 + \rho^3\geq9-13+3+1=0$. Similarly, one can see that $5 + 2 \rho +\rho^2 + 4 (1 + \rho) \cos(w) $ and $17 + 2\rho +\rho^2 + 8 (1 + \rho) \cos(w)$ are non-negative for any $\theta$ and $0\leq\rho\leq 1$.
	
Therefore, we conclude that $\R(g_1)$ is strictly monotone along $\Gamma$ except at its minimum point $w=2-\rho'/2$ and maximum point $w_c=\rho'/2$.
	
Similarly, the fact that the contour $\Sigma$ is a steepest descent path for $-g_1(w)$ can be proved by calculating $\tfrac{\dd \R(g_1)(\theta)}{\dd \theta}$ along $\Sigma$.

\subsection{Proof of uniform convergence of the rescaled kernel, Proposition~\ref{uniform convergence kernelc}}
\label{Unif_conv_Kc}

Consider the steepest descent contour $\Gamma\times \Sigma =(\bigcup_{i=1}^3 \Gamma_i) \times (\bigcup_{i=1}^4 \Sigma_i)$ defined in (\ref{DefGamma}) and (\ref{DefSigma}). We observe that the integrand of the kernel contains the factor $(w-z)^{-1}$, which requires that the $z,w$-contour does not intersect with each other. As a consequence, we need to deform the contour $\Gamma$ away from the saddle point $w_c$, and replaced by a vertical line through $w_c(1+\delta)$ (see Fig. \ref{fig:contour deform 1} and \eqref{DefGammavert}). We choose $\delta=t^{-1/3}$. Under such deformed contour, we now can bound $|w-z|^{-1}$ by $(w_c \delta)^{-1}$ along $\Gamma' \times \Sigma$, where

\begin{equation}
	\Gamma' = \left\{ z \in \Gamma \middle| | z - w_c | > 2 w_c \delta  \right\} + \left\{ z \in \mathbb{C} \middle| z = w_c + w_c \delta (1 - s \ii) , s \in \left[  -\sqrt{3} , \sqrt{3} \right]  \right\} . \label{DefGammavert}
\end{equation}

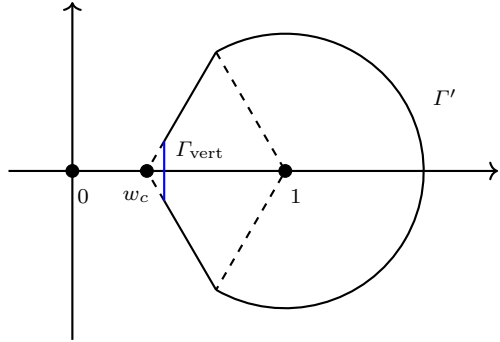
\begin{figure}[h]
	\begin{center}
		\begin{tikzpicture}[scale=2.8]
			
			\draw[->,thick] (-0.3,0) -- (2,0);
			\draw[->,thick] (0,-0.8) -- (0,0.8);

			\draw[fill=black] (0,0) circle (0.03);
			\draw[fill=black] (1,0) circle (0.03);
			\draw[fill=black] (0.35,0) circle (0.03);
			
			\node at (0.05,-0.12) {$0$};
			\node at (1.05,-0.12) {$1$};
			\node at (0.3,-0.12) {$w_c$};
			\node at (1.75,0.35) {$\Gamma'$};
			\node at (0.6,0.1) {$\Gamma_{\rm vert}$};

			\draw[thick] (0.675,-0.562917) arc (-120:120:0.65);

			\coordinate (A) at (0.675,0.562917);
			\coordinate (B) at (1,0);
			\coordinate (C) at (0.675,-0.562917);
			\coordinate (D) at (0.35,0);
			\coordinate (E) at (0.43125,0.140729);
			\coordinate (F) at (0.43125,-0.140729);
			\draw[thick,dashed] (A) -- (B) -- (C);
			\draw[thick,dashed] (E) -- (D);
			\draw[thick,dashed] (F) -- (D);
			\draw[thick] (A) -- (E);
			\draw[thick] (C) -- (F);
			\draw[thick,blue] (E) -- (F) ;
		\end{tikzpicture}
	\end{center}
	\caption{The deformed contour $\Gamma'$ with the blue part $\Gamma_{\rm vert}$.}
	\label{fig:contour deform 1}
\end{figure}

To estimate the integrand of $\bar{K}_t^{c } (\xi , \zeta)$ in $\Gamma' \times \Sigma$, we separate the contour into two parts: far away from the point $w_c$ and the neighbourhood of $w_c$, denoted by 
\begin{align}
	\Gamma^{\Delta} &= \left\{ z \in \Gamma \middle| | z - w_c | \leq \Delta  \right\},\,\,\,\,
	\Gamma'^{\Delta} = \left\{ z \in \Gamma' \middle| | z - w_c | \leq \Delta  \right\}, \label{gamma_prime_Delta} \\
	\Sigma^{\Delta} & = \left\{ w \in \Sigma \middle| | w - w_c| \leq \Delta  \right\},
\end{align} 
where we choose\footnote{We couldn't simply choose $\Delta = \delta = t^{1/3}$, since we need to choose $\delta$ and $\Delta$ such that, as $t$ goes to infinity, $\delta t^{1/3} < \infty$ and $\Delta t^{1/3} \rightarrow \infty$ are satisfied, i.e. the paths $\Gamma'^\Delta$ and $\Sigma^\Delta$ become the Airy contours defined in \eqref{def:airy_function}.} $\Delta=t^{-1/9}$ so that the integrand along the distant parts can be bounded by $\e^{-a_1 \Delta^3 t} \rightarrow 0$ as $t$ tends to infinity. We will see this later in the proof. The deformed contour is now separated into 4 parts: $\Gamma'^{\Delta} \times \Sigma^{\Delta}$, $(\Gamma' \backslash \Gamma'^{\Delta}) \times \Sigma^{\Delta}$, $\Gamma'^{\Delta} \times (\Sigma \backslash \Sigma^{\Delta})$, $(\Gamma' \backslash \Gamma'^{\Delta}) \times  (\Sigma \backslash \Sigma^{\Delta})$. As will be shown later, the main contribution in the long time limit is from the first part $\Gamma'^{\Delta} \times \Sigma^{\Delta}$, i.e., paths of integration passing near a saddle point $w_c$. It will be also shown that other parts and the error terms of the first parts converge uniformly to zero on $ {[-L , L]}^2$.

The integrand in \eqref{RescKernelc} is estimated in three parts: $z$-integrand $\mathrm{e}^{ f(z, t,\xi) - f(w_c , t, \xi) + g_{\phi}(z) }$; $w$-integrand $\mathrm{e}^{ - f(w, t,\zeta) + f(w_c , t, \zeta) + g_{\psi}(w) }$; and $(w-z)^{-1}$. We have the following bounds of the factor $|w-z|$ along the deformed contour,
\begin{subequations}
	\begin{align}\label{Minimum|w-z|}
		\underset{ (z , w) \in \Gamma'^{\Delta} \times \Sigma^{\Delta} }{\min}{ | w - z | } =& w_c \delta,\\
		\underset{ (z , w) \in (\Gamma' \times \Sigma) \backslash (\Gamma'^{\Delta} \times \Sigma^{\Delta}) }{\min}{ | w - z | } \geq & \frac{\sqrt{3}}{2} \Delta. \label{Minimum|w-z|_2} 
	\end{align}
\end{subequations}

\paragraph{(i) Main contribution on $\Gamma'^{\Delta} \times \Sigma^{\Delta}$}Let us now give the main contribution on $\Gamma'^{\Delta} \times \Sigma^{\Delta}$. Consider the $z$-integrand $\mathrm{e}^{ f(z, t,\xi) - f(w_c , t, \xi) }$ and introduce a new variable: $z  = w_c +  \lambda^{-1} Z$, where $\lambda$ is some constant which will be determined later. Using the Taylor expansions given in \eqref{gExpands}, the $z$-integrand $\mathrm{e}^{ f(z, t,\xi) - f(w_c , t, \xi) + g_{\phi}(z) }$ is rewritten as 
\begin{equation}\label{Zintegrand}
	\begin{split}
		& \mathrm{e}^{ f(z, t,\xi) - f(w_c , t, \xi) + g_{\phi}(z) } \\
		= &  \mathrm{e}^{ \{ g_1(z) - g_1(w_c) \} t + \{ g_2(z) - g_2(w_c) \} t^{1/2} + \{ g_3(z , \xi) - g_3(w_c , \xi) \} t^{1/3} + g_{\phi}(z) } \\
		= & \mathrm{e}^{ 2 a_1 {(\lambda^{-1} Z)}^3 t + b_{3,\xi} \lambda^{-1} Z t^{1/3} + g_{\phi}(w_c) + \mathcal{O}( Z^4 t , Z^2 t^{1/2} , Z^2 t^{1/3} , LZ^2t^{1/3} , Z ) } \\
		= & \mathrm{e}^{ \frac{1}{3} Z^3 t - (s_2 + \xi) Z t^{1/3}  + \mathcal{O}( Z^4 t , Z^2 t^{1/2} , Z^2 t^{1/3} , Z ) },
	\end{split}
\end{equation}
where in the last line, $g_{\phi}(w_c)=0$, and we choose $\lambda$ such that $6a_1=\lambda^3$, namely $\lambda=(6a_1)^{1/3}$. Note that $L$ is fixed and hence $\mathcal{O}(LZ^2t^{1/3})$ can be absorbed into $\mathcal{O}(Z^2t^{1/3})$. Recalling the coeffcients given in \eqref{gCoefs}, we obtain 
\begin{equation}
%\label{DefNormalLambda}
	\lambda = {\left( \frac{6}{(1 + \rho)(3 - \rho)} \right)}^{1/3}.
\end{equation}
Therefore by simple calculation, we have $b_{3,\xi} \lambda^{-1} = - s_2 - \xi \lambda_c (w_c + c)^{-1} \lambda^{-1}$. We now choose $\lambda_c$ such that $\lambda_c (w_c + c)^{-1} \lambda^{-1}=1$, i.e., 
\begin{equation*}
	\lambda_c = (w_c + c) \lambda = (1 - \rho + 2c ) {\left( \frac{3}{4 (1 + \rho) (3 - \rho)} \right)}^{1/3} ,
\end{equation*}
which agrees with \eqref{DefLambdac}.

With respect to the new variable $Z$, let $\gamma^{\Delta} $ be the corresponding path of the line integral: $
\gamma^{\Delta} = \left\{ Z \in \mathbb{C} \middle| w_c + \lambda^{-1} Z \in \Gamma'^{\Delta}  \right\}$.
It is easy to see that any $Z \in \gamma^{\Delta}$  satisfy $ |Z|   \leq \lambda \Delta $.

We then factorise the integrand into two terms: main contribution and the error term $R_z$:
\begin{equation*}
	\mathrm{e}^{ \frac{1}{3} Z^3 t - (s_2 + \xi) Z t^{1/3}  + \mathcal{O}( Z^4 t , Z^2 t^{1/2} , Z^2 t^{1/3} , L Z^2 t^{1/3} , Z ) }
	=
	\mathrm{e}^{ \frac{1}{3} Z^3 t - (s_2 + \xi) Z t^{1/3} } + R_z,
\end{equation*}
where 
\begin{equation*}
	R_z =  \mathrm{e}^{ \frac{1}{3} Z^3 t - (s_2 + \xi) Z t^{1/3} }
	\left(
	\e^{\mathcal{O}( Z^4 t , Z^2 t^{1/2} , Z^2 t^{1/3} ,  Z ) } - 1
	\right).
\end{equation*}
We claim that the error term $R_z$ gives a zero integral in long time limit, which will be shown in the second part (ii) of the proof. To get rid of the variable $t$ in the integrand, we change the variable $Z $  to $v = t^{1/3} Z$. The deformed contour now becomes $\gamma^{\Delta t^{1/3}} = \left\{ v \in \mathbb{C} \mid w_c + v /(\lambda t^{1/3}) \in {\Gamma'}^{\Delta}  \right\}$. The $z$-integrand is then rewritten into
\begin{equation}\label{zRescKernel}
	\mathrm{e}^{ \frac{1}{3} v^3 - (s_2 + \xi) v } + R_v,
\end{equation}
where 
\begin{equation}
	R_v =  \mathrm{e}^{ \frac{1}{3} v^3 - (s_2 + \xi) v }
	\left(
	\mathrm{e}^{ \mathcal{O}( v^4 t^{-1/3} , v^2 t^{-1/6} , v^2 t^{-1/3} , v t^{-1/3} ) } - 1
	\right).
\end{equation}

Let us now consider the $w$-integrand $\mathrm{e}^{ - f(w, t,\zeta) + f(w_c , t, \zeta) + g_{\psi}(w) }$. Similarly, we change the variable to $w=w_c+\lambda^{-1}W$ and $u=t^{1/3}W$, then the contour now becomes $\sigma^{\Delta t^{1/3}}  = \left\{ u \in \mathbb{C} \middle| w_c +  u /(\lambda t^{1/3}) \in \Sigma^{\Delta}  \right\}$. Following the same steps, we can see the $w$-integrand becomes 
\begin{equation}\label{wRescKernel}
	\mathrm{e}^{ - f(w, t,\zeta) + f(w_c , t, \zeta) + g_{\psi}(w)}
	=
	\frac{1}{w_c + c}\mathrm{e}^{- \frac{1}{3} u^3 + (s_2 + \zeta) u } + R_u,
\end{equation}
where 
\begin{equation}
	R_u =  \mathrm{e}^{ -\frac{1}{3} u^3 + (s_2 + \zeta) u }
	\left(
	\e^{ \mathcal{O}( u^4 t^{-1/3} , u^2 t^{-1/6} , u^2 t^{-1/3} , u t^{-1/3} ) } - 1
	\right).
\end{equation}
Note that the coefficient $(w_c + c)^{-1}$ comes from the factor $\e^{g_{\psi}(w_c)}=(w_c + c)^{-1}$, while $\e^{g_{\phi}(w_c)}=1$.

By the change of variables $v = t^{1/3} \lambda (z - w_c)$, $u = t^{1/3} \lambda (w - w_c)$, and combining \eqref{zRescKernel}, \eqref{wRescKernel} and $(z-w)^{-1}$, the integral that does not involving error terms $R_z$ and $R_w$ now becomes
\begin{multline}\label{RescKernelFirstuvMain1}
	\lambda_c t^{1/3} \int_{{\Gamma'}^{\Delta } \times {\Sigma}^{\Delta }} \frac{\mathrm{d} z}{2 \pi \ii} \frac{\mathrm{d} w}{2 \pi \ii} \frac{1}{ w-z}  \frac{ \mathrm{e}^{ f(z, t,\xi) - f(w_c , t, \xi) + g_{\phi}(z)} }{ \mathrm{e}^{ f(w, t,\zeta) - f(w_c , t, \zeta) - g_{\psi}(w)}  } \sim \\  \int_{\gamma^{\Delta t^{1/3}} \times \sigma^{\Delta t^{1/3}}} \frac{\mathrm{d} v}{2 \pi \ii} \frac{\mathrm{d} u}{2 \pi \ii} \frac{1}{u - v}\mathrm{e}^{  \frac{1}{3} v^3 - (s_2 + \xi ) v -  \frac{1}{3} {u}^3 + (s_2 + \zeta ) u},
\end{multline}
where the coefficient $\lambda_c t^{1/3} (w_c + c)^{-1} (t^{1/3} \lambda)^{-2} t^{1/3} \lambda = \lambda_c \lambda^{-1} (w_c + c)^{-1} = 1$. Since $\R(u) < \R(v)$ holds on $\gamma^{\Delta t^{1/3}} \times \sigma^{\Delta t^{1/3}}$, $1/(u - v) = - \int_{0}^{\infty} \mathrm{d} k \mathrm{e}^{ k (u - v) }$  is valid in the integrand, indicating the right hand side of (\ref{RescKernelFirstuvMain1}) can be expressed as follows
\begin{equation}\label{RescKernelFirstuvMain2}
	-\int_{0}^{\infty} \mathrm{d} k \int_{\gamma^{\Delta t^{1/3}}} \frac{\mathrm{d} v}{2 \pi \ii} \mathrm{e}^{  \frac{1}{3} v^3 - (s_2 + \xi + k) v}  \int_{ \sigma^{\Delta t^{1/3}}} \frac{\mathrm{d} u}{2 \pi \ii} \mathrm{e}^{  -\frac{1}{3} {u}^3 + (s_2 + \zeta + k) u}.
\end{equation}

Recall the Airy kernel $A(x,y)$ defined in \eqref{def:airy_kernel}. It turns out that, in the long time limit, (\ref{RescKernelFirstuvMain2}) behaves as the kernel of  GUE Tracy-Widom distribution, i.e.,
\begin{align}
	\left|- \int_{0}^{\infty} \mathrm{d} k \int_{\gamma^{\Delta t^{1/3}}} \frac{\mathrm{d} v}{2 \pi \ii} \mathrm{e}^{  \frac{1}{3} v^3 - (s_2 + \xi + k) v}  \int_{ \sigma^{\Delta t^{1/3}}} \frac{\mathrm{d} u}{2 \pi \ii} \mathrm{e}^{ - \frac{1}{3} {u}^3 + (s_2 + \zeta + k) u}   -A(s_2+\xi,s_2+\zeta) \right|
	\leq  \mathcal{O}( \mathrm{e}^{ - c_1 \Delta t^{1/3} } )
	\label{eq:final bound}
\end{align}
holds for $t$ large enough where $c_1$ is some positive constant. To see this, we first notice, by a change of variable $u' =-u$ and from the definitions of the Airy kernel \eqref{def:airy_kernel} and the Airy function \eqref{def:airy_function}, that the left hand side of \eqref{eq:final bound} is bounded above by

\begin{multline}
\left| \int_{0}^{\infty} \mathrm{d} k \int_{C - \gamma^{\Delta t^{1/3}}} \frac{\mathrm{d} v}{2 \pi \ii} \mathrm{e}^{  \frac{1}{3} v^3 - (s_2 + \xi + k) v}  \int_{ - \sigma^{\Delta t^{1/3}}} \frac{\mathrm{d} u'}{2 \pi \ii} \mathrm{e}^{ \frac{1}{3} {u'}^3 - (s_2 + \zeta + k) u'} \right| \\ 
+ \left| \int_{0}^{\infty} \mathrm{d} k \int_{C} \frac{\mathrm{d} v}{2 \pi \ii} \mathrm{e}^{  \frac{1}{3} v^3 - (s_2 + \xi + k) v}  \int_{ C - ( - \sigma^{\Delta t^{1/3}})} \frac{\mathrm{d} u'}{2 \pi \ii} \mathrm{e}^{ \frac{1}{3} {u'}^3 - (s_2 + \zeta + k) u'} \right|,
\label{eq:bound mid0}
\end{multline}
where $C$ is the Airy contour in \eqref{def:airy_function}, which starts at $\infty \e^{- \pi \ii / 3}$ and goes to $\infty \e^{\pi \ii /3}$, $-\gamma$ denotes the contour $\gamma$ with its orientation reversed and $\gamma_1 - \gamma_2$ denotes the concatenation of the contours $\gamma_1$ and $-\gamma_2$.
The contour $-\sigma^{\Delta t^{1/3}}$ differs from the Airy contour $C$ by 
\begin{equation}
\label{eq:contour difference}
C - (- \sigma^{\Delta t^{1/3}})=  \left\{ 
u\in \mathbb{C} \mid u=s\e^{ \ii \pi /3} , s\in [\Delta t^{1/3} ,\i] \right\} \cup \left\{ u\in \mathbb{C} \mid u=-s\e^{- \ii \pi /3} , s\in [-\i, - \Delta t^{1/3}] \right\},
\end{equation}
On the other hand by the Cauchy's integral theorem, we can deform the Airy contour $C$ to the contour $C_d$ which contains the same vertical line segment as $\gamma^{\Delta t^{1/3}}$ in the vicinity of the origin: $C_{d}=\{ v \in \mathbb{C} \mid v = w_c\delta\lambda t^{1/3}(1-s \ii),s\in [-\sqrt{3},\sqrt{3}]\}\cup\{v \in C \mid |v| > 2w_c\delta\lambda t^{1/3} \}$, so that the deformed Airy contour also differs from $\gamma^{\Delta t^{1/3}}$ by  \eqref{eq:contour difference} ($C_d - \gamma^{\Delta t^{1/3}}=$  \eqref{eq:contour difference}). 
After using the bound of Airy function: $|\Ai(x)| \leq C_a e^{-ax}$, which holds for any $x \in \mathbb{R}$ with some $a>0$ and $C_a>0$ depending on $a$, one then can perform the integral over $k$ on \eqref{eq:bound mid0} and show that it is bounded by 

\begin{align}
	\label{eq:bound mid}
c' \left(\int_{\Delta t^{1/3}}^{\i} \dd r \e^{-r^3/3-(s_2+\xi)r} + \int_{\Delta t^{1/3}}^{\i} \dd r \e^{-r^3/3-(s_2+\zeta)r}\right),
\end{align}
where $c'$ is some positive constant. For $t$ large enough, we can see the integrand of \eqref{eq:bound mid} is bounded by $\e^{-c_1 r}$ for some positive constant.  As a result, \eqref{eq:bound mid} is bounded by $\mathcal{O}( \mathrm{e}^{ - c_1 \Delta t^{1/3} } )$, and hence we arrive at \eqref{eq:final bound}. Because $\Delta=t^{-1/9}$, $\mathcal{O}( \mathrm{e}^{ - c_1 \Delta t^{1/3} } )$ decays exponentially with respect to $t$. Hence we have shown that without error terms, the integral in $\Gamma'^{\Delta} \times \Sigma^{\Delta}$ tends to an Airy kernel as $t$ goes to $\i$. Next we will prove the integral involving error terms vanishes.

\paragraph{(ii) Estimate of the error term $\boldsymbol{R}$ in $\boldsymbol{\Gamma'^{\Delta} \times \Sigma^{\Delta}}$}
Rewrite the target integral into the main contribution and the error terms:
\begin{multline}\label{RescKernelFirstuvError1}
	\lambda_c t^{1/3} \int_{{\Gamma'}^{\Delta } \times {\Sigma}^{\Delta }} \frac{\mathrm{d} z}{2 \pi \ii} \frac{\mathrm{d} w}{2 \pi \ii} \frac{1}{ w-z}  \frac{ \mathrm{e}^{ f(z, t,\xi) - f(w_c , t, \xi) + g_{\phi}(z)} }{ \mathrm{e}^{ f(w, t,\zeta) - f(w_c , t, \zeta) - g_{\psi}(w)}  } - \eqref{RescKernelFirstuvMain2} \\
	= \int_{\gamma^{\Delta t^{1/3}}} \frac{\mathrm{d} v}{2 \pi \ii}  \int_{ -\sigma^{\Delta t^{1/3}}} \frac{\mathrm{d} u}{2 \pi \ii} \frac{1}{u-v} \left(
	R_vR_u + R_u \e^{\frac{1}{3} v^3 - (s_2 + \xi ) v} + R_v \e^{\frac{1}{3} u^3 - (s_2 + \xi ) u} 
	\right) \\
	= \int_{\gamma^{\Delta t^{1/3}}} \frac{\mathrm{d} v}{2 \pi \ii}  \int_{ -\sigma^{\Delta t^{1/3}}} \frac{\mathrm{d} u}{2 \pi \ii} \frac{1}{u-v}  \e^{h(v,\xi) + h(u,\zeta)} 
	\left[
	\e^{ O(v) + O(u) } - 1
	\right],
\end{multline}
where  $h(v,\xi)=\frac{1}{3} v^3 - (s_2 + \xi ) v$, and $O(v)=\mathcal{O}( v^4 t^{-1/3} , v^2 t^{-1/6} , v^2 t^{-1/3} , v t^{-1/3} )$. Let us now show that the above vanishes in long time limit. First recall \eqref{Minimum|w-z|}, we 
have the estimate $|u-z|^{-1} \leq w_c^{-1} \lambda \delta^{-1} t^{-1/3} = \lambda / w_c$, as $\delta = t^{-1/3}$. Then by the inequality $|\e^x - 1| \leq |x| \e^{|x|}$, we can bound \eqref{RescKernelFirstuvError1} by the product of two line integrals as follows.
\begin{equation}\label{RescKernelErrorR}
	\eqref{RescKernelFirstuvError1} \leq \frac{\lambda}{w_c} \int_{\gamma^{\Delta t^{1/3}}}  \int_{ -\sigma^{\Delta t^{1/3}}} \left|\frac{\mathrm{d} v}{2 \pi \ii} \right| \left|\frac{\mathrm{d} u}{2 \pi \ii} \right|
	\e^{h(v,\xi) + h(u,\zeta)} 
	\e^{ O(v) + O(u) }
	\left|O(u) + O(v) \right| := R.
\end{equation}

We now only need to show that $R$ goes to zero as $t$ tends to infinity. Let us now consider $O(v)$. Clearly, $\mathcal{O}(v^4 t^{-1/3})$ is dominated by $\mathcal{O}(v^4 t^{-1/6})$ for large enough $t$. Likewise $\mathcal{O}(v^2 t^{-1/6})$ dominates $\mathcal{O}(v^2 t^{-1/3})$, and $\mathcal{O}( v t^{-1/6})$ dominates $\mathcal{O}(v t^{-1/3} )$ when $t$ is large enough. As a result, $\mathcal{O}(v^4 t^{-1/3}, v^2 t^{-1/6}, v^2 t^{-1/3} , v t^{-1/3})$ $\ll t^{-1/6} \mathcal{O}(v^4 , v^2, v)$. 
Then let us consider $\e^{O(v)}$. We observe that along $\gamma^{\Delta t^{1/3}}$, ${\rm max} |v| \leq \Delta t^{1/3} \leq t^{1/4}$ as we choose $\Delta = t^{-1/9}$. When $t$ is large enough, $|v^4 t^{-1/3}| \leq |v|^4 t^{-1/3} \leq |v|^{3} t^{-1/12} < |v|^3/6$, i.e. $\mathcal{O}(v^4 t^{-1/3}) < |v|^3/6$. Similarly, $\mathcal{O}(v^2 t^{-1/6}, v^2 t^{-1/3}) <  |v|^2/6 $ and $\mathcal{O}( v t^{-1/3}) < 1$ when $t$ is large enough. Therefore, $\mathcal{O}(v^4 t^{-1/3} , v^2 t^{-1/6} , v^2 t^{-1/3} , v t^{-1/3})$ is bounded by $(|v|^3 + |v|^2)/6+1$. Consequently, $R$ is now bounded as
\begin{equation}\label{RescKernelErrorR1}    
	R < t^{- 1/6} \frac{\lambda }{w_c } \int_{\gamma^{\Delta t^{1/3}}}  \int_{ -\sigma^{\Delta t^{1/3}}} \left|\frac{\mathrm{d} v}{2 \pi \ii} \right| \left|\frac{\mathrm{d} u}{2 \pi \ii} \right|
	\e^{\tilde{h}(v,\xi) + \tilde{h}(u,\zeta)} 
	\left|\tilde{O}(v) + \tilde{O}(u) \right|,
\end{equation}
where $\tilde{O}(v) = \mathcal{O}(v^4  , v^2 , v)$, and $\tilde{h}(v,\xi) = \frac{1}{3} v^3 - (s_2 +\xi) v + \frac{1}{6} |v|^3 + \frac{1}{6} |v|^2 + 1$. One sees that the dependence of $t$ only appear in the integration boundary $u = v = \lambda \Delta t^{1/3} e^{\pm \ii \pi /3}$. Since $ \Delta = t^{-1/9}$, i.e., $\Delta t^{1/3} \gg 1$, the integrand is dominated by the term $\e^{\frac{1}{3} v^3 + \frac{1}{6} |v|^3}\e^{\frac{1}{3} u^3 + \frac{1}{6} |u|^3} = \e^{ - \frac{1}{3} \lambda^3 \Delta^3 t }$ at the boundary. This implies that the integral \eqref{RescKernelErrorR1}, without the prefactor $t^{-1/6}$, is bounded in the limit $t\rightarrow \infty$. Namely, for large enough $t$, there exists a constant $c_3$ such that
\[R < c_3 t^{-1/6}.\]

\paragraph{(iii) Estimate of $\boldsymbol{(\Gamma' \times \Sigma) \backslash ( \Gamma'^{\Delta} \times \Sigma^{\Delta})  }$}
Next, we evaluate the rest of 3 terms in $(\Gamma' \times \Sigma) \backslash ( \Gamma'^{\Delta} \times \Sigma^{\Delta}) $: $(\Gamma' \backslash \Gamma'^{\Delta}) \times \Sigma^{\Delta}$, $\Gamma'^{\Delta} \times (\Sigma \backslash \Sigma^{\Delta})$, $(\Gamma' \backslash \Gamma'^{\Delta}) \times  (\Sigma \backslash \Sigma^{\Delta})$. Here we only show the proof sketch of $(\Gamma' \backslash \Gamma'^{\Delta}) \times \Sigma^{\Delta}$, while the proofs of the other two follow exactly in the same way.

Let us consider $(\Gamma' \backslash \Gamma'^{\Delta}) \times \Sigma^{\Delta}$. By Proposition \ref{descent contour 1}, we know $\Gamma$ is the steepest descent contour along which the real part of $g_1(z)$ takes the maximum value at $z = w_c$ on $\Gamma$. For $t$ large enough,  the real part of $g_1(z)$ takes the maximum value at the points $z = w_c + \Delta \mathrm{e}^{ \pm \ii \pi/3 }$ on $\Gamma' \backslash \Gamma'^{\Delta}$. By the Taylor expansion \eqref{g def1}, one can see that along $\Gamma' \backslash \Gamma'^{\Delta}$, we have
\begin{align*}
	\Re(g_1(z)-g_1(w_c)) \leq -2 a_1 \Delta^3 + \Delta^4 \Re(\e^{\pm 4 \ii \pi / 3} h_1(w_c + \Delta \mathrm{e}^{ \pm i \pi/3 })) 
	\leq - a_1 \Delta^3, 
\end{align*}
where $a_1>0$ is given in \eqref{gCoefs}, and the second inequality holds when $t$ is large enough. Since $\Delta = t^{-1/9}$ and the function $h_1(z)$ is bounded along $\Gamma'$, then $\Delta| h_1(z)| \leq a_1$ for $t$ large enough. Therefore, the $z$-integrand except $(w-z)^{-1}$ is bounded by
\begin{equation*}
	|\e^{f(z , t , \xi) - f(w_c , t , \xi)}|
	\leq 
	\e^{ -a_1 \Delta^3  t +\mathcal{O}(t^{1/2})}.
\end{equation*}
Note that the bound is uniform for any fixed $L$. Then by \eqref{Minimum|w-z|_2}, we have the bound of $|w-z|^{-1} \leq 2 / \sqrt{3} \Delta^{-1}  \ll 2 t^{1/3}$. Since $\Delta = t^{-1/9}$, then $\e^{-\Delta^3 t }= \e^{-t^{2/3}}$. Obviously $\e^{-t^{2/3}} t^{1/3} \e^{\mathcal{O}(t^{1/2})} \rightarrow 0$ as $t \rightarrow \i$. In conclusion, for large enough $t$, 
\begin{equation}\label{ErrorFarAway1}
	|\e^{f(z , t , \xi) - f(w_c , t , \xi)}(w-z)^{-1}|
	\leq 
	\e^{ -a_1 \Delta^3 t/2 },
\end{equation}
where $a_1>0$. 

Now we are left with the $w$-integrand (except for $|w-z|^{-1}$). Using the change of variable again $u = t^{1/3} \lambda (w - w_c)$, we have 
\begin{equation*}
	t^{1/3} \mathrm{e}^{ - f(w, t,\zeta) + f(w_c , t, \zeta) + g_{\psi}(w)}
	=\frac{\lambda^{-1}}{w_c + c}\mathrm{e}^{ h(u, \zeta)} \e^{O(u)  },
\end{equation*}
where $h(u, \zeta) =\frac{1}{3} u^3 - (s_2 + \zeta) u  $ and $O(u) = \mathcal{O}( u^4 t^{-1/3} , u^2 t^{-1/6} , u^2 t^{-1/3} , u t^{-1/3} )$. By the same analysis as in part (ii), one can see that 
\begin{equation}\label{ErrorFarAway2}
	\int_{\sigma^{\Delta t^{1/3} }} \left| \frac{\dd u}{2 \pi \ii } \right|  \left|
	\frac{\lambda^{-1}}{w_c + c}\mathrm{e}^{ h(u, \zeta)} \e^{O(u)  } \right| \leq c_4,
\end{equation}
where $c_4$ is some positive constant independent of $t$. 

Combining \eqref{ErrorFarAway1} and \eqref{ErrorFarAway2}, we obtain 
\begin{align}\label{ErrorFarCase1}
	\left| \lambda_c t^{1/3} \int_{( \Gamma' \backslash {\Gamma'}^{\Delta} ) \times  \Sigma^{\Delta}  } \frac{\mathrm{d} z}{2 \pi \ii} \frac{\mathrm{d} w}{2 \pi \ii} \mathrm{e}^{ f(z, t, \xi) - f(w_c , t, \xi) + g_{\phi}(z) }  \mathrm{e}^{ - f(w , t, \zeta) + f(w_c , t, \zeta) + g_{\psi}(w) }    \frac{1}{w-z} \right| \leq c_5 \mathrm{e}^{ - \frac{1}{2} a_1 \Delta^3 t },
\end{align}
where $c_5$ is some positive constant, and hence \eqref{ErrorFarCase1} goes to zero as $t$ goes to infinity.

\subsection{Proof of \eqref{fbar bound}}
\label{Upper_Bound_Kc}

Since $g_1(z)$ is analytic along $\Gamma'$ and $\Sigma'$, there exists some positive constant $H_1$ such that $|h_1(z)| < H_1$ for $z$ along $\Gamma'$ and $\Sigma'$. Similarly, we also have $|h_2(z)| < H_2$, $|\bar{h}_3(z)| < H_3$, $|h_{\phi}(z)| < H_{\phi}$ and $|h_{\psi}(z)| < H_{\psi}$. Consider $g_1(z)$ near $z = w_c$, we have
\begin{equation}\label{g1 estm1}
	\left|g_1(z)-g_1(w_c)\right| = \left| 2 a_1 (z-w_c)^3 + (z-w_c)^4 h_1(z) \right|  \leq 2 a_1 |z-w_c|^3 + H_1 |z-w_c|^4 \leq 3a_1 |z-w_c|^3,
\end{equation}
where the last inequality holds under the restriction that $|z-w_c| \leq H_1 / a_1$. Consequently, if $|z-w_c|<A$ where $ A:=\max \{H_1 / a_1, H_2 / |b_2|, H_3 / |\bar{b}_3|, H_{\phi} / |b_{\phi}| , H_{\psi} / |b_{\psi}| \}$, the following inequalities hold:

\begin{subequations}\label{bound Kerc eq1}
	\begin{align}
		\left|g_1(z)-g_1(w_c)\right| \leq & 3a_1 |z-w_c|^3, \\
		\left|g_2(z)-g_2(w_c)\right| \leq & 2|b_2| |z-w_c|^2, \\
		\left|\bar{g}_3(z)-\bar{g}_3(w_c)\right| \leq & 2|\bar{b}_3| |z-w_c|, \\
		\left|g_{\phi}(z)-g_{\phi}(w_c)\right| \leq & 2|b_{\phi}| |z-w_c|, \\
		\left|g_{\psi}(z)-g_{\psi}(w_c)\right| \leq & 2|b_{\psi}| |z-w_c|.
	\end{align}
\end{subequations}
Furthermore, if $|z-w_c| > 12 |b_2| t^{-1/2} / a_1 $, then $2|b_2| |z-w_c|^2 t^{1/2} \leq \frac{1}{6} a_1 |z-w_c|^3 t$. Therefore, if $ |z-w_c| > \max \{  \frac{12 |b_2|}{a_1} t^{-1/2} ,\allowbreak \sqrt{\frac{12 |\bar{b}_3|}{a_1}} t^{-1/3}, \sqrt{\frac{12 |b_{\phi}|}{a_1}} t^{-1/2}, \sqrt{\frac{12 |b_{\psi}|}{a_1}} t^{-1/2}\}$,
\begin{subequations} \label{bound Kerc eq2}
	\begin{align}
		\left|g_2(z)-g_2(w_c)\right|t^{1/2} \leq & 2|b_2| |z-w_c|^2 t^{1/2} \leq \frac{1}{6} a_1 |z-w_c|^3 t, \\
		\left|\bar{g}_3(z)-\bar{g}_3(w_c)\right| t^{1/3} \leq & 2|\bar{b}_3| |z-w_c| t^{1/3}  \leq \frac{1}{6} a_1 |z-w_c|^3 t, \\
		\left|g_{\phi}(z)-g_{\phi}(w_c)\right| \leq & 2|b_{\phi}| |z-w_c| \leq \frac{1}{6} a_1 |z-w_c|^3 t, \\
		\left|g_{\psi}(z)-g_{\psi}(w_c)\right| \leq & 2|b_{\psi}| |z-w_c| \leq \frac{1}{6} a_1 |z-w_c|^3 t.
	\end{align}
\end{subequations}
When $t$ is large enough, $t^{-1/3} \gg t^{-1/2}$, i.e., we only require $|z-w_c| > \sqrt{\frac{12 |\bar{b}_3|}{a_1}} t^{-1/3} := c_1 t^{-1/3}$, then the above inequalities hold. Using the above inequalities, we have the following estimations.

$\bullet$ Contribution from $\Gamma_{\rm vert}$: \eqref{fbar bound1} \newline For $z \in \Gamma_{\rm vert}$, $|z-w_c| \leq 2 w_c \delta \leq \Delta = t^{-1/9}$. Since $A$ is a fixed constant, one can choose $t$ large enough such that $|z-w_c| \leq \Delta < A$. Moreover, we restrict $\delta > c_1 t^{-1/3} / w_c$ so that $|z - w_c| \geq w_c \delta > c_1 t^{-1/3}$. Thus, from \eqref{bound Kerc eq1} and \eqref{bound Kerc eq2} we have,
\begin{equation*}
	\left| \e^{\bar{f}(z,t) - \bar{f}(w_c, t) + g_{\phi}(z) } \right| \leq \e^{(3a_1 +a_1/2)|z-w_c|^3 t} \e^{g_{\phi}(w_c)} \leq \e^{28 a_1 (w_c \delta)^3t} \e^{g_{\phi}(w_c)}.
\end{equation*}

$\bullet$ Contribution from $\Gamma'^{\Delta} \backslash \Gamma_{\rm vert}$: \eqref{fbar bound2} \newline Note that we can also bound this part by $\e^{28 a_1 (w_c \delta)^3}$ using the same method as above. However, since the contour now is along the direction $\e^{\pm \ii \pi /3}$, we could refine the bound of $g_1(z)$. We parameterise $z$ along $\Gamma'^{\Delta} \backslash \Gamma_{\rm vert}$ by $z=w_c + v \e^{\pm \ii \pi /3}$ where $v \in (2 w_c \delta , \Delta]$. Therefore, 
\begin{equation*}
	\left|  \e^{g_1(z) - g_1(w_c)} \right|= \left|\e^{-2a_1 v^3 + h_1(z) v^4 \e^{\pm 4 \ii \pi /3}}\right| \leq \e^{-a_1 v^3},
\end{equation*}
when $|z-w_c| < A$. In fact, we have $|z-w_c| \leq \Delta = t^{-1/9} < A$ when $t$ is large enough. Then we repeat the above estimate \eqref{bound Kerc eq2} for $g_2(z) t^{1/2} + \bar{g}_3(z) t^{1/3} + g_{\phi}(z) $ and we obtain for $\Gamma'^{\Delta} \backslash \Gamma_{\rm vert}$, 
\begin{equation}\label{fbar bound2}
	\left|	\e^{\bar{f}(z,t) - \bar{f}(w_c, t) + g_{\phi}(z) }  \right| \leq \e^{(-a_1 +a_1/2)|z-w_c|^3 t} \e^{g_{\phi}(w_c)} \leq \e^{-a_1 v^3 t /2} \e^{g_{\phi}(w_c)}.
\end{equation}
In fact, when $z \in \Gamma'^{\Delta} \backslash \Gamma_{\rm vert}$, we have $|z-w_c| \geq 2 w_c \delta > w_c \delta > c_1 t^{-1/3}$, which satisfies the restriction of \eqref{bound Kerc eq2}.

$\bullet$ Contribution from $\Gamma' \backslash \Gamma'^{\Delta}$: \eqref{fbar bound3} \newline This estimate is exactly the same as part (iii) in the Appendix~\ref{Unif_conv_Kc}. Recall that we have proved in Appendix~\ref{Unif_conv_Kc}, along  $\Gamma' \backslash \Gamma'^{\Delta}$, 
\begin{equation}\label{fbar bound3}
	\left| \e^{\bar{f}(z,t) - \bar{f}(w_c,t) + g_{\phi}(z)}  \right| \leq \e^{- a_1 \Delta^3 t + \mathcal{O}(t^{1/2})} \leq \e^{- a_1 \Delta^3 t /2},
\end{equation} 
where $\e^{\mathcal{O}(t^{1/2})} < \e^{ a_1 \Delta^3 t /2}$ for $t$ large enough.

\subsection{Proof of \eqref{zw+cBound}}\label{appx:zw+cBound}
Along $\Gamma'$, the minimum value of $z+c$ is taken at the point $z = w_c ( 1+ \delta)$ (See Fig. \ref{Fig.CotourVicinityw_c1_text}). Therefore,  
\begin{align*}
	\left| \frac{w_c + c}{z + c}\right| \leq \left| \frac{w_c + c}{w_c(1 + \delta)  + c}\right| = 
	\left| \frac{1}{ 1 + \frac{w_c\delta}{w_c + c}}\right| \leq \e^{-\frac{1}{2}\frac{w_c}{w_c + c} \delta},
\end{align*}
where the last inequality follows by $1/(1+x) < \e^{-x/2}$ when $0<x<2$. This restriction is satisfied if we suppose $\delta < 1$. Similarly, we consider the term involving $\zeta$. Along $\Sigma'$, the maximum value of $|w +c|$ given by 
\begin{align*}
	\left|w + c\right| \leq 
	\left|\sqrt{(w_c+c)^2 + 4 \delta^2 w_c^2 - 2 \delta w_c(w_c+c)}\right|
	= 
	\sqrt{w_c^2(4\delta^2 - 2\delta +1) +2w_cc(1-\delta) +c^2},
\end{align*}
This can be easily seen from Fig. \ref{Fig.CotourVicinityw_c1_text}.
Then we bound $4\delta^2-2\delta+1$ by $1-\delta$, which is valid when $0 < \delta < 1/4$. It follows that 
\begin{align*}
	|w+c|
	\leq &
	\sqrt{(w_c^2+2w_cc+c^2)(1-\delta)-\delta c^2}
	\\
	= &
	\sqrt{(w_c + c)^2 -\delta [(w_c + c)^2-c^2]} \\
	\leq &
	\sqrt{(w_c + c)^2 -\delta w_c(w_c + c)}.
\end{align*}
Therefore, 
\begin{align*}
	\left|\frac{w + c}{w_c + c}\right| \leq \sqrt{1- \frac{\delta w_c}{w_c +c}} \leq \e^{-\frac{1}{2}\frac{w_c}{w_c + c}\delta},
\end{align*}
where the last line follows by $\sqrt{1-x} \leq \e^{-x/2}$ when $x>0$.

\section{The proofs of Lemmas and Proposition in Section~\ref{sec:Asymptotics_second}}

In this appendix we prove Lemmas~\ref{commutativity sum and int}, \ref{kernel holomorphic} and Proposition~\ref{rank1 perturbation det2} in Section~\ref{sec:Asymptotics_second}.

\subsection{Proof of Lemma~\ref{commutativity sum and int}, on commutativity between the integrals and the sums}
\label{ap:commutativity}

In the following, we will prove that the determinant of the matrix $K^c(x_i , x_j , \vec{z})$ has the dominant series as
\begin{equation}
\label{eq:comsi_prop_1}
\left| \det\left[ K^c(x_i , x_j, \vec{z}) \right]_{1 \leq i , j \leq k} \right| \leq M \prod_{i=1}^{k}{ \mathrm{e}^{ - \frac{w_c}{4(w_c + c)} \delta x_i } },
\end{equation}
where $M$ is some constant depending on $n, m \in \mathbb{N}$, $t \in (0,\infty)$ and $k \in [1,m]$.
Since the series in the right hand side of \eqref{eq:comsi_prop_1} is a geometric series, we can perform the infinite sums from $x_i = 1$ to $\infty$ for all $i \in [1,k]$.

\begin{equation*}
\label{eq:comsi_prop_2}
\sum_{x_1 = 1}^{\infty}{ \sum_{x_2 = 1}^{ \infty }{ \cdots \sum_{x_k = 1}^{\infty}{ M \prod_{i=1}^{k}{ \e^{ - \frac{w_c}{4(w_c + c)} \delta x_i } } } } } = M {\left[ \frac{1}{ \e^{ (w_c \delta)/[4(w_c + c)] } - 1 } \right]}^k < \infty.
\end{equation*}
From the Weierstrass $M$-test, it turns out that we can commute the sums over $x_i \in \mathbb{N}$ for all $i \in [1,k]$ and the integrals with respect to $z_j$ for all $j \in [1,n-1]$.

In order to show \eqref{eq:comsi_prop_1}, we will prove the inequality

\begin{equation}
\label{eq:comsi_prop_3}
\left| {(w_c + c)}^{x - y} K^c(x , y , \vec{z}) \right| \leq C \e^{ - \frac{w_c}{4(w_c + c)} \delta (x + y) },
\end{equation}
where $C$ is some constant depending on $n,m \in \mathbb{N}$ and $t \in (0,\infty)$.
From the formula of $K^c(x,y,\vec{z})$, which is given by \eqref{Kxyz}, we can find the following bound:

\begin{equation*}
\begin{split}
 \label{eq:comsi_prop_4}
& \left| {(w_c + c)}^{x - y} K^c(x,y,\vec{z}) \right| \\
\leq & \oint_{1} \frac{ \left| \dd z \right| }{2 \pi } \left| \frac{z + \rho'}{z + 1} \right| \e^{ - \frac{t}{2} \Re{z} } {\left| \frac{z}{z - 1} \right|}^m {\left| \frac{1 + z}{z} \right|}^n {\left| \frac{w_c + c}{z + c} \right|}^x \prod_{j=1}^{n-1}{ \left| \frac{1 + z_j z}{1 + z} \right| } \times \\
& \oint_{0 , - \rho' ,  \{ - z_j^{-1} \}_{j=1}^{n-1} } \frac{ \left| \dd w \right| }{2 \pi} \left| \frac{w + 1}{(w + \rho')(w + c )} \right| \e^{ \frac{t}{2} \Re{w} } {\left| \frac{w - 1}{w} \right|}^m {\left| \frac{w}{1 + w} \right|}^n {\left| \frac{w_c + c}{w + c} \right|}^{-y}   \prod_{j=1}^{n-1}{ \left| \frac{1 + w}{1 + z_j w} \right| } \frac{1}{ \left| w-z \right| }
\end{split}
\end{equation*}
Let us choose the contours of $z$-integral and $w$-integral as $\Gamma'$ and $\Sigma'$ defined in the proof of Proposition~\ref{exp_bound}, respectively. 
This implies that, for any $z$ and $w$ on the contours, the parts depending on $x \in \mathbb{N}$ and $y \in \mathbb{N}$ are bounded above as follows.

\begin{align*}
{\left| \frac{w_c + c}{z + c} \right|}^x & \leq \e^{  - \frac{w_c}{4(w_c + c)} \delta x }  \label{eq:comsi_prop_6} \\
{\left| \frac{w_c + c}{w + c} \right|}^{-y} & \leq \e^{  - \frac{w_c}{4(w_c + c)} \delta y } 
\end{align*}
Since the contour of $z$-integral does not pass the points at $z=-c$ and  $|z|, |w| < \infty$ and $|w - z| > 0$ holds for any $(z,w) \in \Gamma' \times \Sigma'$, the other parts of integrands are bounded by some constant depending on $n, m \in \mathbb{N}$ and $t \in (0,\infty)$. 
In addition, considering that the lengths of the contours are finite, it turns out that \eqref{eq:comsi_prop_3} holds.

Since a determinant of a matrix whose $(i,j)$ entry is given by $a_{i,j}$ and a matrix whose $(i,j)$ entry is given by $b^{x_i - x_j} a_{i,j}$ are equivalent, we obtain the equality

\begin{equation}
\label{eq:comsi_prop_8}
\left| \det\left[ K^c(x_i , x_j , \vec{z}) \right]_{1 \leq i , j \leq k }   \right| =  \left| \det\left[ {(w_c + c)}^{x_i - x_j} K^c(x_i , x_j , \vec{z}) \right]_{1 \leq i , j \leq k }   \right|.
\end{equation}
Application of the Hadamard's inequality to the right hand side of \eqref{eq:comsi_prop_8} yields

\begin{equation}
\label{eq:comsi_prop_9}
\left| \det\left[ {(w_c + c)}^{x_i - x_j} K^c(x_i , x_j , \vec{z}) \right]_{1 \leq i , j \leq k }   \right| \leq \prod_{i=1}^{k}{ \sqrt{ \sum_{j=1}^{k}{ {\left|  {(w_c + c)}^{x_i - x_j} K^c(x_i , x_j , \vec{z})  \right|}^2 } } }.
\end{equation}
Finally, using the inequality \eqref{eq:comsi_prop_3}, we get

\begin{equation}
\label{eq:comsi_prop_10}
 \prod_{i=1}^{k}{ \sqrt{ \sum_{j=1}^{k}{ {\left|  {(w_c + c)}^{x_i - x_j} K^c(x_i , x_j , \vec{z})  \right|}^2 } } } \leq  {C }^k \prod_{i=1}^{k}{ \sqrt{ \sum_{j=1}^{k}{  \mathrm{e}^{ - \frac{w_c}{2(w_c + c)} \delta (x_i + x_j) }  } } } \leq  { k }^{\frac{k}{2}} C^k  \prod_{i=1}^{k}{ \mathrm{e}^{ - \frac{w_c}{4(w_c + c)} \delta x_i } }.
\end{equation}
It follows from \eqref{eq:comsi_prop_8}, \eqref{eq:comsi_prop_9} and \eqref{eq:comsi_prop_10} that the inequality \eqref{eq:comsi_prop_1} holds with a constant $M$ set as $M = { k }^{\frac{k}{2}} C^k$.

\subsection{Proof of the regularity of the kernel, Lemma~\ref{kernel holomorphic}}
\label{ap:kernel holomorphic}

A complex function $f : \mathbb{C}^{n-1} \rightarrow \mathbb{C}$ is complex analytic if and only if it is holomorphic.
In the following, we will show the kernel is holomorphic by proving it is complex analytic in ${S(1,r)}^{n-1}$ for $r \in (0, 1)$.
In other words, we will show that, for any $\vec{b} = (b_1,\ldots,b_{n-1}) \in {S(1,r)}^{n-1}$, there exists an open polydisc $S(b_1, r_1) \times \cdots \times S(b_{n-1}, r_{n-1}) \subset {S(1,r)}^{n-1}$ such that the kernel $K^c(x,y,\vec{z})$ has a power series expansion

\begin{equation*}
K^c(x,y,\vec{z}) =  \sum_{k_1,\ldots,k_{n-1} = 0}^{\infty}{ c_{k_1,\ldots,k_{n-1}} \prod_{j=1}^{n-1}{{(z_j - b_j)}^{k_j}}  } ,
\end{equation*}
which converges for any $\vec{z} \in S(b_1, r_1) \times \cdots \times S(b_{n-1}, r_{n-1})$, where coefficients $c_{k_1,\ldots,k_{n-1}}$ are independent of $\vec{z}$.
In this proof, we start from the form
\begin{equation}\label{Holomorphic_start_kernel}
K^c(x,y,\vec{z}) =   \oint_{1} \frac{\dd z}{2 \pi \ii} F^c(z,x)  \prod_{j=1}^{n-1}{ \frac{1 + z_j z}{1 + z} } 
 \oint_{0 , - \rho' ,  \{ - z_j^{-1} \}_{j=1}^{n-1} } \frac{\dd w}{2 \pi \ii} G^c(w,y) \prod_{j=1}^{n-1}{\frac{1 + w}{1 + z_j w}} \frac{1}{w-z},
\end{equation}
where the functions $F^c(z,x)$ and $G^c(w,y)$ are defined in \eqref{def:FcGc}.

As stated in the proof of Lemma~\ref{rank1 perturbation term1}, we can choose the contour such that

\begin{equation}
\label{eq:holo_lem_1}
\left| \frac{1 + w}{w} \right| > r > \left| b_j - 1 \right|
\end{equation}
hold for any $\vec{b} = (b_1, \ldots, b_{n-1}) \in {S(1,r)}^{n-1}$ and $w$ on the contour. Using the reverse triangle inequality and \eqref{eq:holo_lem_1}, we obtain

\begin{equation}
\label{eq:holo_lem_2}
\left| \frac{1+w}{w} + (b_j - 1) \right| \geq \left| \frac{1 + w}{w} \right| - |b_j - 1| >  r - |b_j - 1|  > 0.
\end{equation}
Suppose that the open polydisc $S(b_1, r_1) \times \cdots \times S(b_{n-1}, r_{n-1})$ is included in ${S(1,r)}^{n-1}$, $r_j$ and $b_j$ satisfy

\begin{equation}
\label{eq:holo_lem_3}
r_j < r - |b_j - 1|
\end{equation} 
for any $j \in [1,n-1]$.
It follows from \eqref{eq:holo_lem_2} and \eqref{eq:holo_lem_3} that the equality and inequalities

\begin{equation}
\label{eq:holo_lem_4}
\left| \frac{ w(b_j - z_j)}{1 + b_j w} \right| = \left| \frac{b_j - z_j}{ (1 + w)/w + (b_j - 1) } \right| \leq \frac{|z_j - b_j|}{ \left| (	1 + w)/w \right| - \left| b_j - 1 \right| } < \frac{r_j}{r - | b_j - 1 |} < 1
\end{equation}
hold for any $\vec{z} \in S(b_1, r_1) \times \cdots \times S(b_{n-1}, r_{n-1}) \subset {S(1,r)}^{n-1}$ and any $w$ on the contour.

The right hand side of \eqref{Holomorphic_start_kernel} depends on $z_j$ via the factor $(1+z_j z)/(1 + z_j w)$, and
\eqref{eq:holo_lem_4} guarantees that a Taylor series expansion of it around $z_j = b_j$

\begin{equation}
\label{Taylor_series_factor}
    \frac{1+z_j z}{1 + z_j w} = \sum_{k=0}^{\infty}{ d_k(z,w,b_j) {(z_j - b_j)}^k }
\end{equation}
converges, where
\begin{equation*}
    d_k(z,w,b_j) = 
    \begin{cases} 
    \frac{1+b_j z}{1+b_j w} & \mathrm{for}\, k=0 \\
    \frac{w-z}{w(1 + b_j w)} {\left( \frac{-w}{1+b_j w} \right)}^k & \mathrm{for}\, k \geq 1 .
    \end{cases}
\end{equation*}
Since a Taylor series \eqref{Taylor_series_factor} converges uniformly for any $z$ and $w$ on the contours, we can move the sum over $k \in \mathbb{N} \cup \{ 0 \}$ outside $z, w$-integrals.
Therefore, the kernel can be expanded as

\begin{equation}
\label{eq:holo_lem_5}
K^c(x,y,\vec{z}) = \sum_{k_1,\ldots, k_{n-1} = 0}^{\infty} c_{k_1,\ldots,k_{n-1}} \prod_{j=1}^{n-1} {(z_j - b_j)}^{k_j},
\end{equation}
where 

\begin{equation*}\label{Coef_first}
    c_{k_1,\ldots,k_{n-1}} = \oint_1 \frac{\dd z}{2 \pi \ii} \oint_{0, - \rho', {\{ - b^{-1}_j \}}_{j=1}^{n-1} } \frac{ \dd w }{2 \pi \ii} \frac{F^c(z,x) G^c(w,y) }{w - z} {\left(\frac{1+w}{1+z} \right)}^{n-1} \prod_{j=1}^{n-1}{d_{k_j}(z,w,b_j)}.
\end{equation*}

To confirm convergence of a power series in the right hand side of \eqref{eq:holo_lem_5}, we will evaluate an upper bound of $|c_{k_1,\ldots,k_{n-1}}|$. 
Using the Cauchy's integral theorem, we can choose the contours of $z,w$-integrals such that the absolute value of integrand except for $\prod_{j=1}^{n-1}{d_{k_j}(z,w,b_j)}$ is finite for any $n,m \in \mathbb{N}$ and $t \in (0,\infty)$, i.e., the contour of $z$-integral does not pass the points at $z=-c$ and the lengths of both contours are finite. 
From \eqref{eq:holo_lem_2} and \eqref{eq:holo_lem_3}, we also have
\begin{equation*}
\left| \frac{w}{1 + b_j w} \right|  < \frac{1}{r - |b_j - 1|} \leq \frac{1}{r_j}
\end{equation*}
hence we can find an upper bound of $|d_k(z,w,b_j)|$ as

\begin{subequations}
\label{Coef_second}
\begin{align}
    \left| d_0(z,w,b_j) \right| & = \left|  \frac{1+b_j z}{1 + b_j w} \right| < C_1 , \\ 
    \left| d_k(z,w,b_j) \right| & = \left|  \frac{w-z}{w(1 + b_j w)} {\left( \frac{-w}{1+b_j w} \right)}^k  \right| < \frac{C_2}{r_j^{k}} ,
\end{align}
\end{subequations}
where $C_1$ and $C_2$ are some positive constants.
Utilising \eqref{Coef_second}, it turns out that $| c_{k_1,\ldots,k_{n-1}} |$ is bounded as

\begin{equation*}
\begin{split}
\left| c_{k_1,\ldots,k_{n-1}} \right| & \leq \oint_1 \frac{\left| \dd z \right|}{2 \pi } \oint_{0, - \rho', {\{ - b^{-1}_j \}}_{j=1}^{n-1}} \frac{ \left| \dd w \right| }{2 \pi} \left| \frac{F^c(z,x) G^c(w,y) }{w-z} \right|  {\left| \frac{1 + w}{1 + z} \right|}^{n-1} \prod_{j=1}^{n-1}{ \left| d_{k_j}(z,w,b_j) \right| } \\
& < \frac{C}{r_1^{k_1} \cdots r_{n-1}^{k_{n-1}}} ,
\end{split}
\end{equation*}
where $C$ is some constant depending on $n, m \in \mathbb{N}$ and $t \in (0, \infty)$.
This implies that a power series in the right hand side of \eqref{eq:holo_lem_5} converges.

\subsection{Proof of extension of decoupling to determinant, Proposition~\ref{rank1 perturbation det2}}
\label{ap:rank1_determinant}

In order to simplify the notations, we introduce the abbreviations

\begin{align*}
K_{ij}(\vec{z}) & = K^c(x_i , x_j, \vec{z}), & K_{{\mathrm{W}},ij} & = K^c_{\mathrm{W}}(x_i , x_j), & A_i(\vec{z}) & = A^c(x_i, \vec{z}), \\
{(A_p)}_i & = A^c_p(x_i), & B_i & = B^c(x_i),
\end{align*}
and the vectors

\begin{align*}
\vec{K}^{(k)}_{i}(\vec{z}) &= \left(
\begin{array}{cccc}
K_{i1}(\vec{z}), & K_{i2}(\vec{z}), & \cdots , & K_{ik}(\vec{z})
\end{array}
\right) ,\\
\vec{K}^{(k)}_{{\mathrm{W}},i} & = \left(
\begin{array}{cccc}
K_{\mathrm{W},i1}, & K_{\mathrm{W},i2}, & \cdots , & K_{\mathrm{W},ik}
\end{array}
\right), \\
\vec{B}^{(k)} &= \left(
\begin{array}{cccc}
B_{1}, & B_{2}, & \cdots , & B_{k}
\end{array}
\right). 
\end{align*}
In the following, we represent the determinant of a $k \times k$ matrix, $\det[v_{ij}]_{1 \leq i , j \leq k}$ as follows, using $k$-dimensional row vectors $\vec{v}_i = (v_{i1}, \cdots, v_{ik})$. 

\begin{equation*}
\left|
\left[
\begin{array}{c}
\vec{v}_1  \\
\vec{v}_2  \\
\vdots \\
\vec{v}_k 
\end{array}
\right]
\right| = 
\left|
\begin{array}{cccc}
v_{11} & v_{12} & \cdots & v_{1k}  \\
v_{21} & v_{22} & \cdots & v_{2k}  \\
\vdots & \vdots & \ddots & \vdots \\
v_{k1} & v_{k2} & \cdots & v_{kk}
\end{array}
\right|
\end{equation*}
Using these notations, the claim of Proposition~\ref{rank1 perturbation det2} can be represented as

\begin{equation}
\label{eq:rank1lem5'}
\begin{split}
& \oint_{1} \prod_{j=1}^{n-1}{ \frac{\dd z_j}{2 \pi \ii} } \frac{ \prod_{1 \leq i < j \leq n-1}{(z_j - z_i)} }{ \prod_{j=1}^{n-1}{{(z_j -1)}^{j+1}} } \prod_{j=1}^{n-1}{h(z_j)} 
\left|
\left[
\begin{array}{c}
\vec{K}^{(k)}_{1}(\vec{z}) \\
\vdots \\
\vec{K}^{(k)}_{k}(\vec{z}) 
\end{array}
\right]
\right| \\
= & \oint_{1} \prod_{j=1}^{n-1}{ \frac{\dd z_j}{2 \pi \ii} } \frac{ \prod_{1 \leq i < j \leq n-1}{(z_j - z_i)} }{ \prod_{j=1}^{n-1}{{(z_j -1)}^{j+1}} } \prod_{j=1}^{n-1}{h(z_j)} 
\left|
\left[ 
\begin{array}{c}
\vec{K}^{(k)}_{\mathrm{W}, 1} - \sum_{p=1}^{n-1}{ \prod_{q=1}^{p}{ (z_q - 1) {(A_p)}_1 } } \vec{B}^{(k)}  \\
\vdots \\
\vec{K}^{(k)}_{\mathrm{W}, k} - \sum_{p=1}^{n-1}{ \prod_{q=1}^{p}{ (z_q - 1) {(A_p)}_k } } \vec{B}^{(k)}  
\end{array}
\right]
\right| .
\end{split}
\end{equation}

Since $K_{ij}(\vec{z})$ is invariant under exchanges of any two $\vec{z}$-variables and is a holomorphic function in the vicinity of $z_j = 1 $ for any $j \in [1,n-1]$ as shown in Lemma~\ref{kernel holomorphic}, we can apply Lemma~\ref{rank1 perturbation term2} to one row of $\det[ K_{ij}(\vec{z}) ]_{1 \leq i , j \leq k}$.
Applying Lemma~\ref{rank1 perturbation term1} to the first row and using the multilinearity of determinant, the left hand side of \eqref{eq:rank1lem5'} is divided into two parts. 

\begin{align}
& \oint_{1} \prod_{j=1}^{n-1}{ \frac{\dd z_j}{2 \pi \ii} } \frac{ \prod_{1 \leq i < j \leq n-1}{(z_j - z_i)} }{ \prod_{j=1}^{n-1}{{(z_j -1)}^{j+1}} } \prod_{j=1}^{n-1}{h(z_j)} 
\left|
\left[
\begin{array}{c}
\vec{K}^{(k)}_{1}(\vec{z})  \\
\vdots \\
\vec{K}^{(k)}_{k}(\vec{z}) 
\end{array}
\right]
\right| \nonumber \\
= &  \oint_{1} \prod_{j=1}^{n-1}{ \frac{\dd z_j}{2 \pi \ii} } \frac{ \prod_{1 \leq i < j \leq n-1}{(z_j - z_i)} }{ \prod_{j=1}^{n-1}{{(z_j -1)}^{j+1}} } \prod_{j=1}^{n-1}{h(z_j)} 
\left|
\left[
\begin{array}{c}
\vec{K}^{(k)}_{\mathrm{W} ,1} - (z_1 - 1) A_1(\vec{z}) \vec{B}^{(k)}\\
\vec{K}^{(k)}_{2}(\vec{z})  \\
\vdots  \\
\vec{K}^{(k)}_{k}(\vec{z}) 
\end{array}
\right]
\right| \nonumber \\
= & \oint_{1} \prod_{j=1}^{n-1}{ \frac{\dd z_j}{2 \pi \ii} } \frac{ \prod_{1 \leq i < j \leq n-1}{(z_j - z_i)}  }{ \prod_{j=2}^{n-1}{{(z_j -1)}^{j+1}} } \prod_{j=1}^{n-1}{h(z_j)} \left\{
\frac{1}{{(z_1 - 1)}^2}
\left|
\left[
\begin{array}{c}
\vec{K}^{(k)}_{\mathrm{W} ,1}  \\
\vec{K}^{(k)}_{2}(\vec{z}) \\
\vdots  \\
\vec{K}^{(k)}_{k}(\vec{z}) 
\end{array}
\right]
\right|  -   \frac{1}{(z_1 - 1)}
\left|
\left[
\begin{array}{c}
 A_1(\vec{z}) \vec{B}^{(k)} \\
\vec{K}^{(k)}_{2}(\vec{z})  \\
\vdots  \\
\vec{K}^{(k)}_{k}(\vec{z})
\end{array}
\right]
\right|  \right\} \label{eq:rank1lem2_1}
\end{align}

For the same reason, we can apply Lemma \ref{rank1 perturbation term1} to the second row of the first term of the right hand side of \eqref{eq:rank1lem2_1} and divide the term into further two terms.
Carrying out the same calculation from the second row to the $k$\textsuperscript{th} row in order, the determinant is divided into $k+1$ terms.

\begin{equation}
\label{eq:rank1lem2_2}
\begin{split}
& \oint_{1} \prod_{j=1}^{n-1}{ \frac{\dd z_j}{2 \pi \ii} } \frac{ \prod_{1 \leq i < j \leq n-1}{(z_j - z_i)} }{ \prod_{j=1}^{n-1}{{(z_j -1)}^{j+1}} } \prod_{j=1}^{n-1}{h(z_j)} 
\left|
\left[
\begin{array}{c}
\vec{K}^{(k)}_{1}(\vec{z})  \\
\vdots \\
\vec{K}^{(k)}_{k}(\vec{z}) 
\end{array}
\right]
\right| \\
= & \oint_{1} \prod_{j=1}^{n-1}{ \frac{\dd z_j}{2 \pi \ii} } \frac{ \prod_{1 \leq i < j \leq n-1}{(z_j - z_i)}  }{ \prod_{j=2}^{n-1}{{(z_j -1)}^{j+1}} } \prod_{j=1}^{n-1}{h(z_j)} \left\{
\frac{1}{{(z_1 - 1)}^2}
\left|
\left[
\begin{array}{c}
\vec{K}^{(k)}_{\mathrm{W} ,1}  \\
\vdots \\
\vdots \\
\vdots \\
\vdots \\
\vec{K}^{(k)}_{\mathrm{W}, k} 
\end{array}
\right]
\right|  -  \frac{1}{(z_1 - 1)} \sum_{\ell =1}^{k}{
\left|
\left[
\begin{array}{c}
\vec{K}^{(k)}_{\mathrm{W} ,1} \\
\vdots \\
\vec{K}^{(k)}_{\mathrm{W} ,\ell -1} \\
 A_\ell (\vec{z}) \vec{B}^{(k)} \\
\vec{K}^{(k)}_{\ell +1}(\vec{z})  \\
\vdots  \\
\vec{K}^{(k)}_{k}(\vec{z})
\end{array}
\right]
\right| } \right\}
\end{split}
\end{equation}

For convenience, we define $K'^c(x,y, \vec{z})$ and $A'^c(x, \vec{z})$ as the function which are obtained by substituting $1$ into $z_1$ of $K^c(x,y, \vec{z})$ and $z_2$ of $A(x, \vec{z})$, respectively.

\begin{equation*}
\begin{split}
K'^c(x,y, \vec{z}) = & \left. K^c(x,y,\vec{z}) \right|_{z_1 = 1} \\
= & \oint_{1} \frac{\dd z}{2 \pi \ii} F^c(z, x) \prod_{j=2}^{n-1}{ \frac{1 + z_j z}{1 + z} } \oint_{0 , - \rho' , -1 ,  \{ - z_j^{-1} \}_{j=2}^{n-1} } \frac{\dd w}{2 \pi \ii} G^c(w, y) \prod_{j=2}^{n-1}{\frac{1 + w}{1 + z_j w}} \frac{1}{w-z} \\
A'^c(x, \vec{z}) = & \left. A^c(x,\vec{z}) \right|_{z_2 = 1} \\
= & \oint_{1} \frac{ \dd z }{2 \pi \ii} F^c(z, x) \frac{1}{1 + z} \prod_{j=3}^{n-1}{ \frac{1 + z_j z}{1 + z} }
\end{split}
\end{equation*}
In addition, we introduce the abbreviations

\begin{align*}
K'_{ij}(\vec{z}) & = K'^c(x_i , x_j, \vec{z}),  &   A'_i(\vec{z}) & = A'^c(x_i, \vec{z}) ,
\end{align*}
and the vector

\begin{align*}
\vec{K}'^{(k)}_{i}(\vec{z}) &= \left(
\begin{array}{cccc}
K'_{i1}(\vec{z}), & K'_{i2}(\vec{z}), & \cdots , & K'_{ik}(\vec{z})
\end{array}
\right).
\end{align*}

We focus on each summand in the sum over $\ell \in [1,k]$ on the right hand side of \eqref{eq:rank1lem2_2}.
Since $K_{ij}(\vec{z})$ and $A_i(\vec{z})$ are holomorphic functions in the vicinity of $z_1 = 1$, $z_1$-integrand is a function with a single pole of order $1$ at $z_1 = 1$ and $z_1$-integration is carried out by substituting $1$ into $z_1$.
Hence, the $r$\textsuperscript{th} summand ($r \in [1,k]$) in the sum over $\ell \in [1,k]$ on the right hand side of \eqref{eq:rank1lem2_2} is written as

\begin{equation}
\label{eq:rank1lem2_3}
\begin{split}
h(1) \oint_{1} \prod_{j=2}^{n-1}{ \frac{\dd z_j}{2 \pi \ii} } \frac{ \prod_{2 \leq i < j \leq n-1}{(z_j - z_i)} }{ \prod_{j=2}^{n-1}{{(z_j -1)}^{j}} } \prod_{j=2}^{n-1}{h(z_j)} 
\left|
\left[
\begin{array}{c}
\vec{K}^{(k)}_{\mathrm{W}, 1}  \\
\vdots \\
\vec{K}^{(k)}_{\mathrm{W}, r -1} \\
 A_r (\vec{z}) \vec{B}^{(k)} \\
\vec{K}'^{(k)}_{r +1}(\vec{z})  \\
\vdots \\
\vec{K}'^{(k)}_{k}(\vec{z})
\end{array}
\right]
\right| .
\end{split}
\end{equation}
Removing the variable $z_{n-1}$ and shifting the indices of $\Vec{z}$-variables by $1$ as $\{ z_1,\ldots,z_{n-2} \} \to \{ z_2,\ldots,z_{n-1} \}$ in the claim of Lemma \ref{rank1 perturbation term1}, we can obtain the following equality about the kernel $K'_{ij}(\vec{z})$.

\begin{equation}\label{eq:rank1lem2_4}
\begin{split}
& \oint_{1} \prod_{j=2}^{n-1}{ \frac{ \dd z_j}{2 \pi \ii} } \frac{ \prod_{2 \leq i < j \leq n-1}{(z_j - z_i)} }{ \prod_{j=2}^{n-1}{{(z_j -1)}^{j}} } \prod_{j=2}^{n-1}{h(z_j)} K'^c(x ,y , \vec{z}) \\
= & \oint_{1} \prod_{j=2}^{n-1}{ \frac{\dd z_j}{2 \pi \ii} } \frac{ \prod_{2 \leq i < j \leq n-1}{(z_j - z_i)} }{ \prod_{j=2}^{n-1}{{(z_j -1)}^{j}} } \prod_{j=2}^{n-1}{h(z_j)} \left[ K_{\mathrm{W}}^c(x , y) -  (z_2 - 1) A'^c(x, \vec{z}) B^c(y) \right]
 \end{split}
\end{equation}
Since $K'_{ij}(\vec{z}) $ and $A_i(\vec{z})$ are symmetric under the exchange of any two variables in $\{ z_2,\ldots,z_{n-1} \}$, we can apply the equality \eqref{eq:rank1lem2_4} to one of the rows from $r +1$\textsuperscript{st}  to $k$\textsuperscript{th} row of the determinant \eqref{eq:rank1lem2_3}.
Applying the equality \eqref{eq:rank1lem2_4} to $r +1$\textsuperscript{st} row of the determinant in \eqref{eq:rank1lem2_3}, it is replaced with $\vec{K}^{(k)}_{\mathrm{W},r +1} - (z_2 - 1) A'_{r +1}(\vec{z}) \vec{B}^{(k)}$ and \eqref{eq:rank1lem2_3} is divided into two terms.

\begin{equation*}
\begin{split}
& h(1) \oint_{1} \prod_{j=2}^{n-1}{ \frac{\dd z_j}{2 \pi \ii} } \frac{ \prod_{2 \leq i < j \leq n-1}{(z_j - z_i)} }{ \prod_{j=2}^{n-1}{{(z_j -1)}^{j}} } \prod_{j=2}^{n-1}{h(z_j)} 
\left|
\left[
\begin{array}{c}
\vec{K}^{(k)}_{\mathrm{W},1} \\
\vdots \\
\vec{K}^{(k)}_{\mathrm{W},r -1} \\
 A_r (\vec{z}) \vec{B}^{(k)}   \\
\vec{K}^{(k)}_{\mathrm{W}, r +1} - (z_2 - 1) A'_{r +1}(\vec{z}) \vec{B}^{(k)} \\
\vec{K}'^{(k)}_{r +2}(\vec{z}) \\
\vdots  \\
\vec{K}'^{(k)}_{k}(\vec{z})
\end{array}
\right]
\right| \\
= & h(1) \oint_{1} \prod_{j=2}^{n-1}{ \frac{\dd z_j}{2 \pi \ii} } \frac{ \prod_{2 \leq i < j \leq n-1}{(z_j - z_i)} }{ \prod_{j=3}^{n-1}{{(z_j -1)}^{j}} } \prod_{j=2}^{n-1}{h(z_j)}
\left\{ \frac{1}{{(z_2 - 1)}^2}
 \left|
 \left[
\begin{array}{c}
\vec{K}^{(k)}_{\mathrm{W},1} \\
\vdots \\
\vec{K}^{(k)}_{\mathrm{W},r -1} \\
 A_r (\vec{z}) \vec{B}^{(k)} \\
\vec{K}^{(k)}_{\mathrm{W}, r +1}   \\
\vec{K}'^{(k)}_{r +2}(\vec{z}) \\
\vdots \\
\vec{K}'^{(k)}_{k}(\vec{z})
\end{array}
\right]
\right| 
-  \frac{A_r (\vec{z}) A'_{r +1}(\vec{z})}{(z_2 - 1)}
 \left|
 \left[
\begin{array}{c}
\vec{K}^{(k)}_{\mathrm{W},1} \\
\vdots \\
\vec{K}^{(k)}_{\mathrm{W},r -1} \\
 \vec{B}^{(k)} \\
 \vec{B}^{(k)}  \\
 \vec{K}'^{(k)}_{r +2}(\vec{z}) \\
 \vdots  \\
\vec{K}'^{(k)}_{k}(\vec{z})
\end{array}
\right]
\right| \right\}
\end{split}
\end{equation*}
Obviously, the second term vanishes because the $r$\textsuperscript{th} row and the $r+1$\textsuperscript{st} row are equivalent.
Carrying out the same calculations to the $r+2$\textsuperscript{nd} and subsequent rows of the first term, we obtain

\begin{equation*}
h(1) \oint_{1} \prod_{j=2}^{n-1}{ \frac{\dd z_j}{2 \pi \ii} } \frac{ \prod_{2 \leq i < j \leq n-1}{(z_j - z_i)} }{ \prod_{j=2}^{n-1}{{(z_j -1)}^{j}}  } \prod_{j=2}^{n-1}{h(z_j)} 
\left|
\left[
\begin{array}{c}
\vec{K}^{(k)}_{\mathrm{W},1} \\
\vdots \\
\vec{K}^{(k)}_{\mathrm{W},r -1} \\
A_r(\vec{z}) \vec{B}^{(k)}  \\
\vec{K}^{(k)}_{\mathrm{W},r +1} \\
\vdots  \\
\vec{K}^{(k)}_{\mathrm{W}, k}
\end{array}
\right]
\right| .
\end{equation*} 
Since $K_{\mathrm{W},ij}$ is independent of $\vec{z}$-variables and $A_i(\vec{z})$ is independent of $z_1$, it is allowed to revive $z_1$-integral as follows.

\begin{equation*}
\oint_{1} \prod_{j=1}^{n-1}{ \frac{\dd z_j}{2 \pi \ii} } \frac{ \prod_{1 \leq i < j \leq n-1}{(z_j - z_i)} }{ \prod_{j=1}^{n-1}{{(z_j -1)}^{j+1}} } \prod_{j=1}^{n-1}{h(z_j)} 
 \left| 
 \left[
\begin{array}{c}
\vec{K}^{(k)}_{\mathrm{W}, 1} \\
\vdots  \\
\vec{K}^{(k)}_{\mathrm{W}, r-1} \\ 
 (z_1 - 1) A_r(\vec{z}) \vec{B}^{(k)}  \\
\vec{K}^{(k)}_{\mathrm{W}, r +1} \\  
\vdots  \\
\vec{K}^{(k)}_{\mathrm{W}, k}  
\end{array}
\right]
\right|
\end{equation*}
In addition, since $K_{\mathrm{W},ij}$ is independent of $\vec{z}$-variables, we can apply Lemma~\ref{rank1 perturbation term2} to $r$\textsuperscript{th} row and obtain 

\begin{equation}
\label{eq:rank1lem2_5}
\begin{split}
 \oint_{1} \prod_{j=1}^{n-1}{ \frac{\dd z_j}{2 \pi \ii} } \frac{ \prod_{1 \leq i < j \leq n-1}{(z_j - z_i)} }{ \prod_{j=1}^{n-1}{{(z_j -1)}^{j+1}} } \prod_{j=1}^{n-1}{h(z_j)}
 \left| 
 \left[
\begin{array}{c}
\vec{K}^{(k)}_{\mathrm{W}, 1} \\
\vdots  \\
\vec{K}^{(k)}_{\mathrm{W}, r -1} \\
\sum_{p=1}^{n-1}{ \prod_{q=1}^{p}{ (z_q - 1) {(A_p)}_r } } \vec{B}^{(k)} \\
\vec{K}^{(k)}_{\mathrm{W}, r +1} \\
\vdots \\
\vec{K}^{(k)}_{\mathrm{W}, k}
\end{array}
\right]
\right| .
\end{split}
\end{equation}
This implies that we can rewrite $r$\textsuperscript{th} summand in the sum over $\ell \in [1,k]$ on the right hand side of \eqref{eq:rank1lem2_2} as \eqref{eq:rank1lem2_5} for all $r \in [1,k]$.
Therefore, we have

\begin{equation}
\label{eq:rank1lem2_6}
\begin{split}
& \oint_{1} \prod_{j=1}^{n-1}{ \frac{\dd z_j}{2 \pi \ii} } \frac{ \prod_{1 \leq i < j \leq n-1}{(z_j - z_i)} }{ \prod_{j=1}^{n-1}{{(z_j -1)}^{j+1}} } \prod_{j=1}^{n-1}{h(z_j)} 
\left|
\left[
\begin{array}{c}
\vec{K}^{(k)}_{1}(\vec{z}) \\
\vdots \\
\vec{K}^{(k)}_{k}(\vec{z}) 
\end{array}
\right]
\right| \\
= & \oint_{1} \prod_{j=1}^{n-1}{ \frac{\dd z_j}{2 \pi \ii} } \frac{ \prod_{1 \leq i < j \leq n-1}{(z_j - z_i)} }{ \prod_{j=1}^{n-1}{{(z_j -1)}^{j+1}} } \prod_{j=1}^{n-1}{h(z_j)} 
\left\{
\left|
\left[
\begin{array}{c}
\vec{K}^{(k)}_{\mathrm{W}, 1} \\
\vdots \\
\vdots \\
\vdots \\
\vdots \\
\vec{K}^{(k)}_{\mathrm{W}, k} 
\end{array}
\right]
\right| - \sum_{\ell=1}^{k}{ \left| \left[
\begin{array}{c}
\vec{K}^{(k)}_{\mathrm{W}, 1}  \\
\vdots \\
\vec{K}^{(k)}_{\mathrm{W}, \ell-1}  \\
\sum_{p=1}^{n-1}{ \prod_{q=1}^{p}{ (z_q - 1) {(A_p)}_\ell } } \vec{B}^{(k)} \\
 \vec{K}^{(k)}_{\mathrm{W}, \ell+1}  \\
 \vdots  \\
\vec{K}^{(k)}_{\mathrm{W}, k} 
\end{array}
\right]
\right| }
 \right\} .
\end{split}
\end{equation}

Moreover, we can find the following equality because the determinant is multilinear function and the determinant which has two identical rows equals zero.

\begin{equation}
\label{eq:rank1lem2_7}
\begin{split}
& \left|
\left[
\begin{array}{c}
\vec{K}^{(k)}_{\mathrm{W}, 1} - \sum_{p=1}^{n-1}{ \prod_{q=1}^{p}{ (z_q - 1) {(A_p)}_1 } } \vec{B}^{(k)}  \\
\vdots \\
\vdots \\
\vdots \\
\vdots \\
\vec{K}^{(k)}_{\mathrm{W}, k} - \sum_{p=1}^{n-1}{ \prod_{q=1}^{p}{ (z_q - 1) {(A_p)}_k } } \vec{B}^{(k)}  
\end{array}
\right]
\right| 
=  \left|
\left[
\begin{array}{c}
\vec{K}^{(k)}_{\mathrm{W}, 1} \\
\vdots \\
\vdots \\
\vdots \\
\vdots \\
\vec{K}^{(k)}_{\mathrm{W}, k} 
\end{array}
\right]
\right| - \sum_{\ell=1}^{k}{ \left| \left[
\begin{array}{c}
\vec{K}^{(k)}_{\mathrm{W}, 1}  \\
\vdots \\
\vec{K}^{(k)}_{\mathrm{W}, \ell-1}  \\
\sum_{p=1}^{n-1}{ \prod_{q=1}^{p}{ (z_q - 1) {(A_p)}_\ell } } \vec{B}^{(k)} \\
 \vec{K}^{(k)}_{\mathrm{W}, \ell+1}  \\
 \vdots  \\
\vec{K}^{(k)}_{\mathrm{W}, k} 
\end{array}
\right]
\right| }
\end{split}
\end{equation}
Finally, it follows from \eqref{eq:rank1lem2_6} and \eqref{eq:rank1lem2_7} that the equality \eqref{eq:rank1lem5'} holds.

\section{Proof of the uniform convergence of $\boldsymbol{\phi}$ on a bounded set, Proposition \ref{uniform convergence phi2}}
\label{ap:uniform convergence phi2}

A rigorous proof falls into the same pattern as in Proposition \ref{uniform convergence kernelc}.
To avoid reiterating ourselves, here we only give a basic idea of the proof. Recall that $\phi_{n-t^{1/2}\kappa_1}( t^{1/2}\xi_1)$ is given by
\[
\phi_k(x) = \oint_1 \frac{\dd w}{2\pi\ii}
(w-\rho')
\e^{f(w,t,\xi_1, \kappa_1)+g_4(w)},
\]
where $f(w,t,\xi, \kappa)=g_1(w)t+g_2(w,\xi, \kappa)t^{1/2}+g_3(w)t^{1/3}$ with $g_i(w)$ and $g_2(w,\xi,\kappa)$ given in \eqref{g def2}. Solving $g_1'(w)=0$ gives us $w_1=\rho'$ and $w_2=-\rho'/2$. One can obtain a steepest descent through $w_1=\rho'$, since the one passing $w_2$ would include extra poles at origin and hence vary the estimate of the integral.

One can see that the following contour $\Theta=\Theta_1\cup\Theta_2$ (see Fig. \ref{fig:contour gaussian 1}) is a steepest descent path of $g_1(w)$ passing through $\rho'$. Namely, $w=\rho'$ is the strict global maximum point of ${\rm Re}(g_1)$ along $\Theta$. This can be proved by calculating $\frac{\dd {\rm Re}(g_1)(s)}{\dd s}$ along $\Theta$.

\begin{subequations}
\label{gamma2}
\begin{align}
\Theta_1=&\left\{w=\rho'-\frac{s}{\sqrt{3}} \ii , \,\,\,s\in[-\rho,\rho]\right\},
%&\Sigma_1=s
\\
\Theta_2=&\left\{w=1+ \frac{2\rho}{\sqrt{3}}\e^{\ii s} , \,\,\,s\in[-5\pi/6,5\pi/6]\right\}.
\end{align}
\end{subequations}

\begin{figure}[h]
\begin{center}
\begin{tikzpicture}[scale=2.5]

\draw[->,thick] (-0.3,0) -- (2,0);
\draw[->,thick] (0,-0.8) -- (0,0.8);

\draw[fill=black] (0,0) circle (0.03);
\draw[fill=black] (1,0) circle (0.03);
\draw[fill=black] (0.35,0) circle (0.03);

\node at (0.05,-0.12) {$0$};
\node at (1,-0.12) {$1$};

\draw[thick] (0.35,-0.375278) arc (-150:150:0.750555);
\draw[thick] (0.35,-0.375278)  -- (0.35,0.375278);

\node at (0.45,-0.1) {$\rho'$};

\coordinate (A) at (0.35,-0.375278);
\coordinate (B) at (1,0);
\coordinate (C) at (0.35,0.375278);
\draw[thick,dashed] (A) -- (B) -- (C);

\end{tikzpicture}
\end{center}
\caption{steepest descent contour of integration in $\phi_k(x)$ passing the saddle point at $w_1=\rho'$.}
\label{fig:contour gaussian 1}
\end{figure}
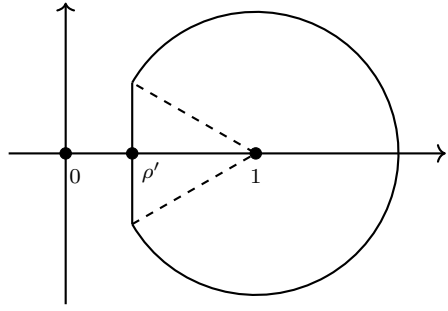

As shown in the proof of Proposition \ref{uniform convergence kernelc}, we can prove that for large enough $t$, only the part $\Theta^{\delta}:=\{w\in \Theta \mid |w-w_1|\leq\delta \}$, where $\delta=t^{-1/6}$, contributes to the integral. Near $w_1=\rho'$, the Taylor expansion of $g_i(w)$ are given by
\begin{subequations}
\label{g taylor expand2}
\begin{align}
g_1(w)-g_1(w_1)
=&
\frac{9(1-\rho)(w-\rho')^2}{16(2-\rho)\rho}+
\mathcal{O}[(w-\rho')^3],
\\
g_2(w,\xi_1,\kappa_1)-g_2(w_1,\xi_1,\kappa_1)
=&
-\frac{c_{\rm g} s_{\rm g} + 2(2-\rho)(\kappa_1+\xi_1\rho)}{2\rho(1-\rho)(2-\rho)}(w-\rho')+
\mathcal{O}[(w-\rho')^2],
\\
g_3(w)-g_3(w_1)
=&
\mathcal{O}[(w-\rho')^2],
\\
g_4(w)-g_4(w_1)=&\mathcal{O}[(w-\rho')].
\end{align}
\end{subequations}
Denote the function $g_i (w)$ without the error terms by $\bar{g}_i(w)$. As in Proposition~\ref{uniform convergence kernelc}, we can show that for large $t$, only the term $\e^{\bar{g}_1(w)t+\bar{g}_2(w,\xi_1,\kappa_1)t^{1/2}+\bar{g}_3(w)t^{1/3}+\bar{g}_4(w)}$ contributes to the integral. We re-parameterise $\Gamma_{\delta}$ by

\begin{equation*}
  w-\rho'=\ii v\frac{2 }{3}\sqrt{\frac{\rho(2-\rho)}{t(1-\rho)}},
\end{equation*}
where $-c \delta t^{1/2} \leq v \leq c \delta t^{1/2}$, and $c=\tfrac{3}{2}\sqrt{\tfrac{(1-\rho)}{\rho(2-\rho)}}$. Thus we are left with
\begin{align*}
&\lim_{t\rightarrow\infty}t\int_{\Gamma_{\delta}}\frac{\dd w}{2\pi\ii} (w-\rho') \e^{\bar{g}_1(w)t+\bar{g}_2(w,\xi_1, \kappa_1)t^{1/2}+\bar{g}_3(w)t^{1/3}+\bar{g}_4(w)-f(\rho',t,\xi_1,\kappa_1)}
\\
=&\lim_{t\rightarrow\infty}
\frac{-t}{\rho(1-\rho)}\int_{\Gamma_{\delta}}\frac{\dd w}{2\pi\ii} (w-\rho') \exp\left(
\frac{9(1-\rho)(w-\rho')^2}{16(2-\rho)\rho}t-\frac{c_{\rm g} s_{\rm g} + 2(2-\rho)(\kappa_1+\xi_1\rho)}{2 \rho(1-\rho)(2-\rho)}(w-\rho')t^{1/2}
\right)\\
=&\lim_{t\rightarrow\infty}
\frac{4(2-\rho)}{9(1-\rho)^2}\int_{c\delta t^{1/2}}^{-c\delta t^{1/2}}\frac{\dd v}{2\pi\ii}
v\e^{-v^2/4-v\ii (s_{\rm g}+\kappa+\xi)/\sqrt{2}}\\
=&\frac{4(2-\rho)}{9(1-\rho)^2}\int_{\infty}^{-\infty}\frac{\dd v}{2\pi\ii}
v\e^{-v^2/4-v\ii (s_{\rm g}+\kappa+\xi)/\sqrt{2}}
\\
=&
\frac{8(2-\rho)(s_{\rm g}+\kappa+\xi)}{9\sqrt{2\pi}(1-\rho)^2}
\e^{-(s_{\rm g}+\kappa+\xi)^2/2}.
\end{align*}
where
\begin{equation*}
  \kappa= 2(2-\rho)\kappa_1/c_{\rm g}:=\kappa_1/\lambda_1,
  \quad\quad\quad
  \xi=2(2-\rho)\rho \xi_1/c_{\rm g}:=\xi_1/\lambda_2,
\end{equation*}
as required.

\bibliographystyle{amsplain}

\bibliography{bibfile}

\end{document}